\documentclass[twocolumn,journal]{IEEEtran}
\usepackage{relsize}
\usepackage{cite}
\usepackage{amsmath,amssymb,amsfonts,comment,bm}

\usepackage[labelformat=simple]{subcaption}

\newtheorem{lemma}{Lemma}
\newtheorem{theorem}{Theorem}
\newtheorem{corollary}{Corollary}
\newtheorem{proposition}{Proposition}
\newtheorem{remark}{Remark}

\usepackage{algorithm}
\usepackage{algpseudocode}
\usepackage{enumerate}
\usepackage{graphicx}
\usepackage{textcomp}
\usepackage{xcolor}

\usepackage{hyperref}
\usepackage{balance}

\usepackage{tikz}
\usetikzlibrary{positioning}
\usetikzlibrary{arrows.meta}

\newcommand{\mb}{\mathbb}
\newcommand{\mf}{\mathbf}
\newcommand{\mc}{\mathcal}

\newcommand{\st}{\mathrm{s.t.}}
\newcommand{\sgn}{\mathrm{sgn}}

\newcommand{\mq}[2]{\left\lceil {#1} \right\rfloor_{#2} }
\newcommand{\tmq}[2]{\lceil {#1} \rfloor_{#2} }
\newcommand{\indep}{\perp \!\!\! \perp}
\newcommand{\half}{{\textstyle \frac12}}

\renewcommand{\mod}{\,{\rm mod}\,}

\def\BibTeX{{\rm B\kern-.05em{\sc i\kern-.025em b}\kern-.08em
		T\kern-.1667em\lower.7ex\hbox{E}\kern-.125emX}}

\DeclareCaptionLabelSeparator{periodspace}{.\quad}
\begin{document}
	
	\title{Modulo Quantization Coding for Primitive Relay and Diamond Channels with Correlated Noises
	}
	
	\author{Yuanxin Guo,~\IEEEmembership{Student Member,~IEEE,} Stark C. Draper,~\IEEEmembership{Senior Member,~IEEE,} and Wei Yu,~\IEEEmembership{Fellow,~IEEE}
		\thanks{
			Manuscript submitted to \emph{IEEE Transactions on Information Theory} on June 18, 2026.
			The authors are with the Edward S. Rogers Sr. Department of Electrical and Computer Engineering, University of Toronto, Toronto, ON M5S 3G4,
			Canada (e-mail: yuanxin.guo@mail.utoronto.ca; stark.draper@utoronto.ca; weiyu@ece.utoronto.ca).
			
			This work has been presented in part at the \emph{IEEE International Symposium on Information Theory (ISIT)}, Ann Arbor, USA, June 2025 \cite{guo2025modulo} and in part at the \emph{IEEE International Symposium on Information Theory (ISIT)}, Guangzhou, China, June 2026 \cite{guo2026modulo}.
			This work was supported by the Natural Science and Engineering Research Council (NSERC) of Canada via the Discovery Grant program.}
	}
	
	\maketitle
	
	\begin{abstract}
		
		This paper proposes modulo quantization (MQ) coding as a simple,
		structured, and low-complexity scheme for channels with primitive (i.e.,
		noiseless digital) relay links and \emph{correlated} Gaussian noises across
		terminals. The key component of MQ coding is the modulo quantization operation,
		which maps a real-valued symbol to its uniform-quantization index taken modulo
		a fixed integer. This operation allows effective exploitation of the common
		noise component shared across the terminals. 
		For the Gaussian primitive relay channel with perfectly correlated noises,
		where a relay has a finite-capacity link to the receiver, MQ coding can be
		shown to achieve the capacity of this channel.
		For the Gaussian primitive diamond channel with perfectly correlated noises,
		where two relays can forward information through finite-capacity links to a
		receiver that has no direct observation of the transmitted signal, 
		MQ coding yields novel achievability bounds that
		improve upon previously known bounds and coincide with the cut-set upper bound
		in certain signal-to-noise ratio (SNR) regimes. In scenarios with highly but
		non-perfectly correlated noises, MQ coding can approach the performance of
		compress-forward (CF) at significantly lower complexity, while surpassing
		decode-forward (DF) for the Gaussian primitive relay channel in certain SNR
		ranges. For the Gaussian primitive diamond channel with non-perfectly
		correlated noises, MQ can outperform both CF and DF at moderate SNR.

	\end{abstract}


	\section{Introduction} \label{sec:intro}
	
	\begin{figure*}[t]
		\centering
		\begin{subfigure}[]{0.25\linewidth}
			\centering
			\begin{minipage}[c][3.6cm][c]{\linewidth}
				\centering
				\begin{tikzpicture}[auto, thick, node distance=2cm,scale = 0.65]
					\node at (-1.5,0) (input1) {$X$};
					\node[scale= 0.8] at (-1.2,-0.65) (power) {$\mb E[X^2]\le P$};
					
					\node[draw, circle,scale = 0.4] (sum1) at (0.5,0) {\Huge +};
					\node at (2.5,0) (output1) {$Y$};
					
					\draw[-Latex] (-1.2,0) -- ++ (1.4,0);
					
					\node at (0.5,1.8) (zinput1) {$Z$};
					\draw[-Latex] (0.5,1.5) -- ++ (0,-1.23);

					\draw[-Latex] (0.77,0) -- ++ (1.45,0);
				\end{tikzpicture}
			\end{minipage}
			\caption{}
		\end{subfigure}
		\hspace{2mm}
		\begin{subfigure}[]{0.25\linewidth}
			\centering
			\begin{minipage}[c][3.6cm][c]{\linewidth}
				\centering
				\begin{tikzpicture}[auto, thick, node distance=2cm,scale = 0.65]
					
					\node at (-1.5,0) (input1) {$X$};
					\node[scale= 0.8] at (-1.2,-0.65) (power) {$\mb E[X^2]\le P$};
					
					\node[draw, circle,scale = 0.4] (sum1) at (0.5,0) {\Huge +};
					\node at (3,0) (output1) {$Y$};
					
					\draw[-Latex] (-1.2,0) -- ++ (1.4,0);
					
					\node at (0.5,1.8) (zinput1) {$Z$};
					\draw[-Latex] (0.5,1.5) -- ++ (0,-1.23);
					\node[draw, rectangle,scale = 0.8] (relay) at (3,1.8) {\bf Relay};
					
					\draw[-Latex] (zinput1) -- (relay);
					\draw[-{Latex[fill=none]},double] (3,1.5) -- ++ (0,-1.23);
					\draw[-,dotted] (zinput1) -- (relay);
					\node at (3.5,0.95) (zinput1) {$R_0$};
					\draw[-Latex] (0.77,0) -- ++ (1.95,0);
				\end{tikzpicture}
			\end{minipage}
			\caption{}
		\end{subfigure}
		\hspace{5mm}
		\begin{subfigure}[]{0.35\linewidth}
			\centering
			\begin{minipage}[c][3.6cm][c]{\linewidth}
				\centering
				\begin{tikzpicture}[auto, thick, node distance=1.5cm,scale =0.65]
					
					\node at (-1.4,0) (input) {$X$};
					\node[scale= 0.8] at (-2.2,-0.65) (power) {$\mb E[X^2]\le P$};
					
					\node[draw, circle,scale = 0.5] (sum1) at (0.4,1.5) {\Large +};
					\node[draw, circle,scale = 0.5] (sum2) at (0.4,-1.5) {\Large +};
					\node[scale = 0.95] at (0.4,2.8) (zinput1) {$Z$};
					\node[scale = 0.95] at (0.4,-2.9) (zinput2) {$Z$};
					\node at (2,1.5) (output1) {$Y_1$};
					\node at (2,-1.5) (output2) {$Y_2$};

					\node[draw, rectangle,scale = 0.8] (relay1) at (4.2,1.5) {\bf Relay 1};
					\node[draw, rectangle,scale = 0.8] (relay2) at (4.2,-1.5) {\bf Relay 2};
					\node at (6,0) (dec)  {$D$};

					
					\draw[-Latex] (input) -- node[above,scale = 0.8] {$\alpha~~$}   (sum1);
					\draw[-Latex] (input) -- node[below,scale = 0.8] {$\beta~~$}   (sum2);

					\draw[-Latex] (0.4,2.55) -- ++ (0.,-0.8);
					\draw[-Latex] (0.4,-2.6) -- ++ (0.,0.85);

					\draw[-Latex] (sum2) -- (output2);
					\draw[-Latex] (sum1) -- (output1);
					\draw[-{Latex[fill=none]}, double] (relay1) -- ++ (1.6,-1.2);
					\draw[-{Latex[fill=none]}, double] (relay2) -- ++ (1.6, 1.2);
					
					\draw[-Latex] (output1) -- (relay1);
					\draw[-Latex] (output2) -- (relay2);
					
					\node[scale = 0.9] at (5.5,1) {$R_1$};
					\node[scale = 0.9] at (5.5,-1) {$R_2$};
				\end{tikzpicture}
			\end{minipage}
			\caption{}
		\end{subfigure}
		\caption{(a) AWGN channel with power constraint $P$ and noise $Z \sim \mc N(0,1)$. (b) Gaussian primitive relay channel in which the relay observes the noise and has a noiseless digital link of capacity $R_0$ to the receiver. The capacity of this channel tends to $R_0$ as $P \to 0$. 
			(c) Gaussian primitive diamond channel with perfect noise correlation and primitive relay links of capacities $(R_1,R_2)$. The capacity of this channel tends to $\min(R_1,R_2)$ as $P \to 0$, assuming $\alpha \neq \beta$.} 
	\label{fig:intro}
\end{figure*}
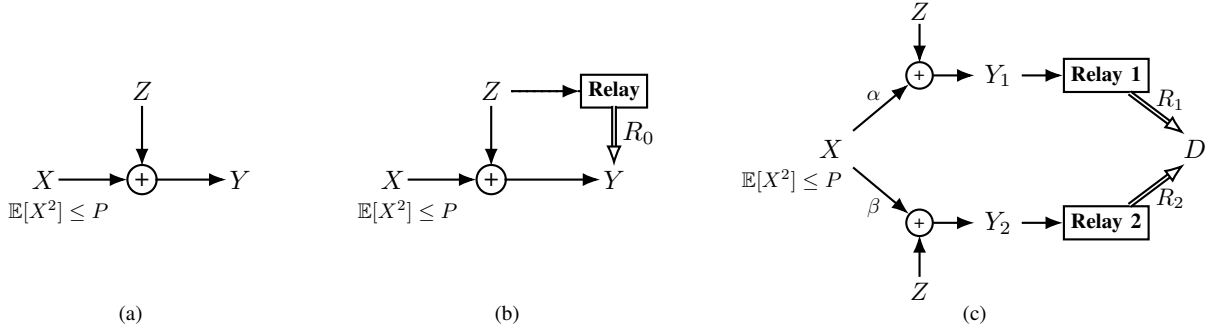


Consider an additive white Gaussian noise (AWGN) channel with transmitter power constraint $P > 0$ and noise distribution $\mc N(0,1)$ as shown in Fig.~\ref{fig:intro}(a). 
The capacity of this channel as a function of the signal-to-noise ratio (SNR) is $\psi(P) = \frac12 \log(1+ P)$. 
Clearly, as $P \rightarrow 0$, the capacity of the AWGN channel goes to zero.

Now suppose that a relay node is added to the AWGN channel model. The relay observes the noise realization perfectly and further it has a primitive (i.e., noiseless digital) relay link of finite capacity $R_0 > 0$ to the receiver, as shown in Fig.~\ref{fig:intro}(b). 
This is an example of the so-called Gaussian primitive relay channel with perfectly correlated noise, for which the capacity is known to be $\psi(P) + R_0$, as established in \cite{cover2007capacity}.

This above result is surprising in two ways. First, every bit of information sent by the relay is worth one bit in the capacity of the overall channel, despite the fact that the relay does not observe the transmitted signal. Second, even with infinitesimal amount of transmit power, the capacity of this channel remains bounded away from zero, converging to $R_0$ as $P \to 0$. This stands in contrast to the AWGN channel without the relay for which the capacity $\psi(P) \to 0$ as $P\to 0$. 

This surprising behaviour of the primitive relay channel stems from the fact that the observations at the relay and the receiver share the exact same noise. 
So if one were able to take the difference between the observed signals at the relay and at the receiver, the common noise can be cancelled completely. 
However, perfect noise cancellation is possible only if the relay link has infinite capacity.
The fact that even with \emph{finite-capacity} relay link, it is possible to achieve an overall finite capacity in this channel using only infinitesimal transmit power is highly nontrivial.



This interesting phenomenon has been pointed out for the broadcast channel with conferencing decoders in the earlier work \cite{farsani2025capacity}. The phenomenon also occurs in the 
Gaussian primitive diamond channel with perfect noise correlation, as illustrated in Fig.~\ref{fig:intro}(c). In this diamond channel model, there are two relays each observing a distinct linear combination of the transmitted signal and a \textit{common} Gaussian noise. The two relays have distinct primitive relay links of finite capacities $R_1$ and $R_2$ to a receiver, which has no direct observation of the transmitted signal. As in the single-relay case, the presence of a common noise component introduces exploitable structure. While the capacity of this channel remains unknown at finite transmit power $P$, we show in this paper that $\min\{R_1,R_2\}$ is always achievable at any power level---it is in fact the asymptotic capacity of this channel as $P \rightarrow 0$.  Remarkably, the capacity of this channel remains bounded away from zero even with infinitesimal transmit power, i.e., this primitive diamond channel exhibits a behaviour analogous to that of the Gaussian primitive relay channel with perfectly correlated noises.

The above observations motivate us to investigate practical coding schemes that can explicitly exploit the shared noise structure in the Gaussian relay channel and the Gaussian diamond channel with primitive relay links and correlated noises. 
In this paper, we propose modulo quantization (MQ) coding as an explicit, low-complexity, and inherently structured scheme for such channels.  
The proposed scheme capitalizes on the noise correlation through a simple algebraic modulo operation, which enables it to achieve an overall positive rate bounded away from zero when the noise correlation is perfect, using only infinitesimal amount of transmit power and with finite-capacity relay links for both the primitive relay channel and the primitive diamond channel, 
thereby providing an intuitive explanation of the surprising phenomenon above. 

\subsection{Related Work}

The study of the relay channel dates back to the pioneering works by van der Meulen \cite{van1971three} and Cover and El Gamal \cite{cover1979capacity}. Within this broad class, the primitive relay channel---in which the relay has a noiseless digital link to the receiver---is introduced in \cite{zhang1988partial}, where a partial converse is established, and further analyzed in \cite{kim2007coding}, which derives upper and lower bounds on its capacity. More recently, for the Gaussian primitive relay channel, upper bounds tighter than the cut-set bound have been established in \cite{wu2017cut,wu2018capacity}. 
However, the capacity of the general primitive relay channel remains elusive. 

One subclass of the primitive relay channel for which the capacity is tractable is the \emph{deterministic} variant, where the relay observation is a deterministic function of the channel input and the receiver observation. For this subclass of channels, it has been shown in \cite{cover2007capacity, kim2008capacity} that the capacity is equal to the cut-set upper bound.

When specialized to the Gaussian primitive relay channel in which the relay and the receiver observe noise-corrupted versions of the transmitted signal, the determinism as studied in \cite{cover2007capacity, kim2008capacity} reduces to the condition that 
the Gaussian noises at the relay and the receiver are perfectly correlated (assuming that the observations at the relay and at the receiver are not scaled versions of each other). In this case, the cut-set bound is the capacity, and it is equal to the sum of the direct-channel capacity and the relay-link capacity \cite{cover2007capacity, kim2008capacity}. Interestingly, each bit in the relay link is worth one bit in the overall capacity, regardless of the transmit power.


This paper aims to provide an explanation of this interesting phenomenon by 
investigating low-complexity coding strategies capable of achieving the capacity 
of the Gaussian primitive relay channel with perfectly correlated noises. 
To this end, the classical compress-forward (CF) coding strategy has been studied for the Gaussian relay
channel with correlated noise in \cite{cui_correlated, lei_correlated}.
In the CF scheme \cite{cover1979capacity}, the relay employs Wyner–Ziv coding \cite{wyner1976rate} to compress its observation while treating the observation at the receiver as side information. The CF scheme is capacity achieving when the noise correlation is perfect
\cite{cover2007capacity, kim2008capacity}.
For the discrete channel, a hash-forward (HF) scheme has also been shown to achieve the capacity of the deterministic primitive relay channels. In the HF scheme \cite{kim2007coding,kim2008capacity}, the relay randomly hashes its observed sequence and forwards the hash index to the receiver. 
Moreover, a flash helping approach is proposed in \cite{bross2019additive,bross2020decoder} for the deterministic primitive relay channels with additive noise. 
In this scheme, the relay quantizes only a small fraction of its observed sequence without binning and forwards the resulting description to the receiver. As the blocklength tends to infinity and the fraction of described symbols vanishes, this scheme also achieves capacity.


While in principle, the aforementioned coding strategies can already be applied to the Gaussian primitive relay channel and they are all capable of achieving capacity in the perfect noise correlation case, in practice these existing schemes suffer from various drawbacks.
First, the HF scheme cannot be directly implemented when the relay observation takes values from a continuous alphabet, because in the HF scheme each possible relay observation must be assigned a hash index. Moreover, both CF and flash helping are computationally intensive due to the use of high-dimensional vector quantization. As they rely on high-dimensional quantization to exploit the shared noise structure, they achieve capacity only in the limit of large block-length for quantization.

This motivates the development of MQ, a scalar implementation of quantize-and-bin that replaces high-dimensional vector quantization and random binning with a simple modular arithmetic operation, in this paper. We show that the capacity of the Gaussian primitive relay channel with perfect noise correlation can already be achieved with a scalar quantizer as a key component.

The primitive diamond channel \cite{schein2001distributed} is a model in which two relays observe noise-corrupted versions of the transmitted signal and forward information to a central receiver through separate primitive relay links. Unlike the primitive relay channel, the destination has no direct channel observation in this model, resulting in a distinct distributed relaying structure. 

The Gaussian primitive diamond channel is introduced in \cite{sanderovich2008communication}.
Most previous studies focus on the case of independent noises, for which capacity bounds are developed in \cite{sanderovich2008communication} and later refined in \cite{wu2019new}. 
Noise correlation, however, can significantly impact the capacity of diamond channel, because it gives rise to the possibility of noise cancellation. (For example, for the related models of Gaussian broadcast and interference channel with noiseless feedback, perfect noise correlation can double the pre-log of the sum capacity \cite{wigger2008pre,gastpar2014feedback}.)
For the Gaussian primitive diamond channel,
recent works \cite{katz2024gaussian,katz2024gaussian2} investigate 
the correlated noises case and establish achievable rates using CF and decode-forward (DF) strategies.


As in the primitive relay setting, existing coding schemes for the Gaussian primitive diamond channel with correlated noises have similar limitations. In particular, DF treats the relays independently and does not exploit the correlation between their observations. Meanwhile, CF retains the computational and structural limitations discussed earlier. This paper applies the MQ framework to this channel, leading to low-complexity coding schemes and improved achievability bounds in both the cases of perfect and near-perfect noise correlations.

The MQ strategy proposed in this paper relies on modulo arithmetic.
Similar ideas have appeared in prior studies of channels with related forms of relaying, for example, in \cite[Section 5.2]{lapidoth2021listsize} and \cite[Section V-C]{lapidoth2024communication}. On a conceptual level, the related modulo-based schemes in \cite{lapidoth2021listsize,lapidoth2024communication} serve purposes quite different from the MQ framework developed in this paper. The modulo-based construction in \cite{lapidoth2021listsize} is used to establish the list-size capacity, while in \cite{lapidoth2024communication}, a similar construction is used for feedback capacity. In contrast, the focus of this paper is on the classical Shannon capacity. The capacity-achieving argument in \cite{lapidoth2021listsize, lapidoth2024communication} relies on combining code construction with time-sharing, whereas the argument in this paper is based on a nested code construction. Finally, the modulo operation is introduced as an ad-hoc proof technique in \cite{lapidoth2021listsize, lapidoth2024communication}, while this paper formalizes MQ as a reusable building block applicable to a broader family of relay networks.

\subsection{Main Contributions}

In this paper, we propose the framework of MQ coding as a fundamental building block for coding over Gaussian channels with primitive relay links and correlated noises. Modulo quantization is an operation that maps a real input to the remainder of its uniform scalar quantization index modulo some positive integer, effectively ``wrapping'' the quantized real line into bins.


This paper shows that despite being only a \emph{scalar} binning operation, 
MQ coding, when used in conjunction with a Gaussian codebook for the AWGN channel, is capable of achieving the capacity of the Gaussian primitive relay channel with perfect noise correlation. 
As a first step, we construct a scalar MQ scheme capable of transmit up to $\lfloor 2^{R_0} \rfloor$ messages without error in a single channel use, irrespective of the (non-zero) transmit power. To achieve the capacity of $\psi(P) + R_0$, we propose a nested coding strategy by superposing the scalar MQ code component-wise on top of a sequence of Gaussian codes supported on a scaled integer lattice. The overall decoding complexity remains comparable to that of a standard Gaussian code, because the scalar MQ component is trivial to decode, and the MQ and the Gaussian codebook can be decoded in a successive manner.


Next, we derive novel achievability bounds based on MQ coding for the Gaussian primitive diamond channel with perfect noise correlation.
A slightly different scalar MQ scheme is constructed to accommodate this channel model. With this scheme, up to $\lfloor 2^{\min(R_1,R_2)} \rfloor$ messages can be transmitted without error with one channel use at any (non-zero) transmit power level. 
By time-sharing this scalar MQ scheme with DF, we obtain a new lower bound on the channel capacity. In both the high-SNR regime and certain low-SNR regime when the channel is asymmetric (with $R_1 \neq R_2$), this achievable rate matches the cut-set upper bound. To the best of authors' knowledge, this low-SNR capacity result is new to the literature. 

This paper also analyzes a specific anti-symmetric Gaussian diamond channel in detail, where the channel gains are equal in magnitude but opposite in sign and the primitive relay links have equal capacities. We show that the hybrid MQ-DF scheme outperforms all previously known relaying strategies across all SNRs. By comparing these results with CF, we also provide a corrected example demonstrating the suboptimality of CF with Gaussian input and Gaussian quantization for the oblivious diamond relay channel (i.e., when the relays operate without the knowledge of the codebook used by the transmitter).

Finally, this paper considers the setting of non-perfectly correlated noises
for both the Gaussian primitive relay channel and the diamond channel. 
We adapt the previously developed scalar MQ scheme to the scenario
with small perturbation in the noise correlation and compare its performance to
known strategies. In some specific SNR regimes, the achievable rate of MQ
coding approaches that of CF for the Gaussian primitive relay channel while
significantly outperforming DF. A
similar observation is made for the Gaussian primitive diamond channel: MQ
coding outperforms both CF and DF in some specific SNR regimes. These results
highlight the practicality of MQ coding as a low-complexity yet competitive
strategy for Gaussian channels with primitive relay links, when the noise correlation is close to but not perfect.

We remark that while the perfectly correlated noise setting is primarily of theoretical interest, the near-perfect case can arise in practice---for example, when multiple receivers are affected by a dominant common source of interference. As the strength of the interference increases, the effective noise correlation also increases, making the near-perfect noise correlation scenarios potentially relevant.

\subsection{Paper Organization and Notation}

The rest of the paper is organized as follows. Section II introduces modulo quantization and reviews its algebraic properties. Section III and IV discuss the Gaussian primitive relay channel and the Gaussian primitive diamond channel respectively, focusing exclusively on the case of perfect noise correlation. Section V extends the discussion to the near-perfect correlation case for the two aforementioned channels. Section VI concludes with a discussion of broader implications and future directions.

In this paper, we use uppercase letters (e.g. $X$, $Y$) to denote random variables and lowercase letters (e.g. $x$, $y$) to denote their realizations. The sequence of random variables $(X_1, \cdots ,X_n)$ is denoted $X^n$. We sometimes use uppercase boldface letters (e.g. $\bm X$) to denote random vectors. The all-one vector and the all-zero vector (of appropriate length) is denoted $\bf 1$ and $\bf 0$. The set of real numbers, positive real numbers, integers, positive integers are denoted $\mb R$, $\mb R_+$, $\mb Z$, and $\mb Z_+$ respectively. For positive integer $a \in \mb R_+$, we use $[a]$ to denote the set $\{0,1,\cdots, a-1\}$, whose cardinality is exactly $a$. 
We use $\indep$ to denote statistical independence. The multivariate Gaussian distribution with mean $\bm \mu$ and covariance matrix $\bm \Sigma$ is denoted $\mc N(\bm \mu, \bm \Sigma)$. We use the shorthand $\psi(\cdot)$ for the AWGN capacity function defined as $\psi(x) = \frac12 \log (1+x)$. All logarithms are base $2$.

\section{Modulo Quantization} \label{sec:prelim}

In this section, we introduce the modulo quantization operation and its basic properties. This operation is used to exploit the shared noise structure arising from perfect noise correlation; it lies at the core of the MQ coding framework developed in subsequent sections.

Fix positive real number $\Delta > 0$ and positive integer $k \in \mb Z_+$. The $(\Delta,k)$-\textit{modulo quantization} is the function $\mq{\cdot}{\Delta,k}: \mb R \to [k]$ defined as \begin{equation}\label{eq:mod_bin_idx}
	\mq{x}{\Delta,k} =  \left\lfloor \frac{x}{\Delta} \right\rfloor \mod k.
\end{equation}
Here, the floor function $\lfloor\cdot\rfloor: \mb R \to \mb Z$ returns the largest integer less than or equal to its argument, and the modulo-$k$ operation $(\cdot) \mod k: \mb Z \to [k]$ returns the remainder of the argument divided by $k$. When no ambiguity arises, we also use the term $(\Delta,k)$-\textit{modulo quantization} to refer to the value $\mq{x}{\Delta,k}$.

The modulo quantization operation can be interpreted as follows: $\lfloor \frac{x}{\Delta} \rfloor$ is the output of an infinite-support scalar quantizer with step-size $\Delta$ with input $x$, and $\mq{x}{\Delta,k}$ can be seen as the mod-$k$ ``wrapped'' version of the quantized output. 
Alternatively, $\mq{x}{\Delta,k}$ can be viewed as an instance of one-dimensional deterministic binning, where points are put into bins according to the colors of the intervals they belong to, as illustrated in Fig.~\ref{fig:mod_quan}.

A few properties of the floor function and the modulo operation are listed below. For all $x \in \mb R$, $a,b \in \mb Z$, and $k \in \mb Z_+$, the following relations hold:
\begin{align}
	\lfloor x + a\rfloor &= \lfloor x \rfloor + a, \label{eq:prop1} \\
	(a \mod k) \mod k &= a \mod k,\label{eq:prop2}\\
	(a + b) \mod k &= [(a \mod k) + b] \mod k  \\ &= [(a \mod k) + (b \mod k)] \mod k.\label{eq:prop3}
\end{align}
The following lemma is a direct consequence of the above.
\begin{lemma}\label{lem:mq}
	Fix any $\Delta > 0$ and $k \in \mb Z_+$, the following identities hold for all $z \in \mb R$:
	\begin{gather}
		\label{eq:colour_lemma}
		(\mq{z + a\Delta}{\Delta, k} -  \mq{z }{\Delta, k}) \mod k  = a \mod k,~~\forall\,a \in \mb Z,\\
		\label{eq:colour_lemma_2}
		\mq{z + b\Delta}{\Delta, k} =  (\mq{z }{\Delta, k} + b) \mod k, ~~\forall\,b \in \mb Z,\\
		\label{eq:lattice_colour_lemma}
		\mq{z + ck\Delta}{\Delta, k} = \mq{z}{\Delta, k},~~\forall\,c \in \mb Z,
	\end{gather}
	where $[k] = \{0,1,\cdots,k-1\}$.
\end{lemma}


Lemma~\ref{lem:mq} shows that if two inputs differ by an integer multiple of the quantization step $\Delta$, then their MQ outputs preserve this offset modulo $k$, effectively retaining their relative shift while eliminating any component common to both. This difference-preserving property is what makes MQ useful when multiple terminals observe signals having a common noise component and send only rate-limited descriptions to a destination.


\begin{figure}
	\centering
	\includegraphics[width=0.9\linewidth]{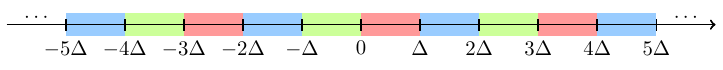}
	\caption{Illustration of the modulo quantizer $\mq{\cdot}{\Delta,3}$. Each color corresponds to one modulo quantization output.}
	\label{fig:mod_quan}
\end{figure} 

\section{Gaussian Primitive Relay Channel with Perfect Noise Correlation}
\label{sec:prim_relay}

\subsection{Channel Model} 
\label{sec:prim_relay_model}

The Gaussian primitive relay channel 
is depicted in Fig.~\ref{fig:prim_relay}. Given input $X$, the receiver and the relay observe $Y = X + Z$ and $Y_{0} =\alpha X + Z_{0}$ respectively, where $Z$ and $Z_{0}$ are the corresponding noises, and $\alpha \in \mb R$ is the relay's channel gain.
In this section, we assume that the noise variables are perfectly (and positively) correlated standard Gaussian variables, i.e., 
\begin{equation}
	\label{eq:noise_distr}
	\begin{bmatrix}
		Z \\ Z_0
	\end{bmatrix} \sim \mc N\left( \mf 0, \begin{bmatrix}
		1 & 1 \\ 1 & 1
	\end{bmatrix}\right).
\end{equation}
In effect, we have $Z_0=Z$.
The relay has a noiseless link of capacity $R_0$ bits per channel use to the receiver. Note that when $\alpha = 0$, this model reduces to the decoder-assisted Gaussian channel discussed in \cite{bross2020decoder}.

A $(2^{nR},n)$-code for this channel is specified by the triple $(f_{\sf enc}^{(n)}, f_{\sf relay}^{(n)},f_{\sf dec}^{(n)})$. The input signal is given by $X^n = f_{\sf enc}^{(n)}(M)$, where $M$ is a random message uniformly distributed over the message set $[2^{nR}]$, and $f_{\sf enc}^{(n)}:  [2^{nR}] \to \mb R^n$ is the encoding function. The encoding function maps a message $m$ to a codeword $x^n(m) = (x_1(m), \cdots, x_n(m))$. An average per-codeword power constraint is imposed such that $\frac1n\|x^n(m)\|^2 \le P$ for every $m \in [2^{nR}]$. Upon observing $Y_0^n$, the relay forwards a $nR_0$-bit message $V = f_{\sf relay}^{(n)}(Y_0^n)$ to the receiver through the noiseless link, where $f_{\sf relay}^{(n)}: \mb R^n \to [2^{nR_0}]$ is the relaying function. Upon observing the channel output $Y^n$ and receiving the description $V$ from the relay, the receiver decodes the transmitted message to be $\hat M = f^{(n)}_{\sf dec}(Y^n, V)$, where $f^{(n)}_{\sf dec}: \mb R^n \times [2^{nR_0}] \to [2^{nR}]$ is the decoding function. 

A rate $R$ is achievable if for every $\epsilon > 0$, there exists a sequence of $(2^{nR},n)$-codes such that for sufficiently large $n$, the error probability $P_e^{(n)} = \Pr\{\hat M \neq M\}$ is upper bounded by $\epsilon$. The capacity of this primitive relay channel $C_{\triangleleft}(R_0)$ is the supremum of all achievable rates.

\begin{figure}[t]
	\centering
	\begin{tikzpicture}[auto, thick, node distance=1.5cm,scale = 0.75]
		
		\node at (-4.4,0) (msg1) {$M$};
		\node[draw, rectangle,scale = 0.8] (enc) at (-2.9,0) {\bf Enc};
		\node at (-1.4,0) (input1) {$X^n$};
		
		\node[draw, circle,scale = 0.5] (sum1) at (0.45,0) {\Large +};
		\node[draw, rectangle,scale = 0.8] (relay) at (4.2,1.75) {\bf Relay};
		\node at (2.2,0) (output1) {$Y^n$};
		\node[draw, rectangle,scale = 0.8] at (4.2,0) (dec) {\bf Dec};
		\node at (5.7,0) (msg2) {$\hat M$};
		
		\draw[-Latex] (input1) -- (sum1);

		\node[scale = 0.95] at (0.45,2.75) (zinput0) {$Z_0^n$};
		\node[scale = 0.95] at (0.45,-1.05) (zinput1) {$Z^n$};
		\draw[-Latex] (0.45,2.55) -- ++ (0.,-0.55);
		\draw[-Latex] (0.45,-0.8) -- ++ (0.,0.55);
		
		\node[draw, circle,scale = 0.5] (sum2) at (0.45,1.75) {\Large +};
		\draw[-Latex] (input1) -- node[above,scale = 0.8] {$\alpha~$} ++ (sum2);

		\node at (2.2,1.75) (output2) {$Y_0^n$};
		\draw[-Latex] (sum2) -- (output2);
		\draw[-Latex] (sum1) -- (output1);
		\draw[-{Latex[fill=none]}, double] (relay) --  (dec);
		\draw[-Latex] (msg1) -- (enc);
		\draw[-Latex] (enc) -- (input1);
		\draw[-Latex] (output1) -- (dec);
		\draw[-Latex] (dec) -- (msg2);
		\draw[-Latex] (output2) -- (relay);
		\node[scale = 0.9] at (4.7,0.9) {$nR_0$};

	\end{tikzpicture}
	\caption{Gaussian primitive relay channel with correlated noises at the receiver and the relay.}
	\label{fig:prim_relay}
\end{figure}
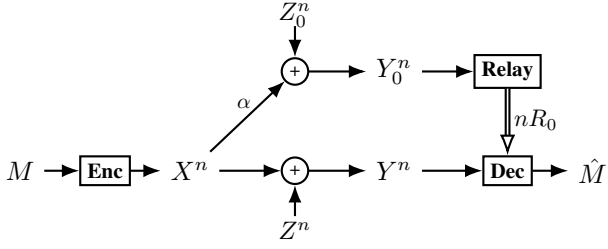


In the following discussion, we exclude the special case where the relay gain $\alpha$ is equal to $1$. In this degenerate setting, the relay and the receiver have identical observations since $Z_0 = Z$, and relaying would provide no benefit. Thus the channel reduces to an AWGN channel without any relay. 

For the non-degenerate case where $\alpha \neq 1$, the capacity of this channel has been fully characterized. Observe that the relay observation $Y_0$ is a deterministic function of the transmitted signal $X$ and the receiver observation $Y$, i.e.,
\begin{equation}\label{eq:functional}
	Y_0 = X + \alpha (Y-X).
\end{equation}
For this class of \textit{deterministic} primitive relay channels, the  capacity has been shown in \cite{kim2008capacity} to be
\begin{equation}\label{eq:cap_prim_relay}
	C = \sup_{p(x)}\min \{I(X;Y) + R_0, I(X;Y,Y_0)\}.
\end{equation}
By specializing \eqref{eq:cap_prim_relay} to the Gaussian primitive relay channel, we have the following theorem.

\begin{theorem}\label{thm:capacity_prim_relay}
	Assume $\alpha \neq 1$. The capacity of the Gaussian primitive relay channel with perfectly correlated noises is 
	\begin{equation} \label{eq:capacity}
		C_{\triangleleft}(R_0) = \psi(P)  + R_0,
	\end{equation}
	where $\psi(x) = \frac12 \log(1+x)$ is the AWGN capacity function.
\end{theorem}

The capacity of this channel can be achieved by CF \cite{cover1979capacity}. In the CF scheme, Wyner-Ziv coding \cite{wyner1976rate} is used to compress the relay observation, while accounting for the receiver observation (as side information for source coding), and the compressed observation is then sent to the receiver. The receiver recovers the transmitted message based on its own observation and the compressed version of the relay's observation. 

An alternative coding technique called \textit{flash helping} is proposed in \cite{bross2019additive,bross2020decoder} for a special case of this channel where $\alpha = 0$. In this special case, the relay observes the Gaussian noise of the direct channel. 
In the scheme of \cite{bross2019additive,bross2020decoder}, the relay compresses a vanishing fraction of its observation into $nR_0$ bits in a lossy fashion. It then forwards that description to the receiver. Through time-sharing with a capacity-achieving Gaussian code, the Gaussian primitive relay channel capacity $C_\triangleleft(R_0)$ is achieved. 

Despite that both schemes achieve capacity, these two schemes suffer from notable limitations in terms of complexity and explicitness. For the CF scheme, the Wyner-Ziv coding step involves typicality-based vector quantization and random binning. This results in high encoding and decoding complexity. For flash helping, the high-resolution quantization involves searching for a closest vector to a given vector among exponentially many possibilities. This is also computationally expensive. Furthermore, both CF and flash helping assume the use of vector quantization codebooks that achieve the rate-distortion limit. This makes practical implementation challenging. These drawbacks motivate the development of a low-complexity and fully explicit coding scheme based on \emph{scalar} quantization alone.

In the rest of this section, we first establish a coding scheme that relies on the structural properties of MQ to harness the noise correlation. This scheme achieves the relay-link capacity $R_0$. We then demonstrate that, when combined with a capacity-achieving scheme for the AWGN channel, MQ achieves the channel capacity $\psi(P) + R_0$. Finally, we extend the MQ framework to the scenario where multiple relays share a common noise component and show that it remains capacity-achieving.



\subsection{MQ Coding}
\label{sec:prim_MQ}

We now present a simple, explicit, low-complexity scalar scheme called MQ coding which achieves \textit{zero error} for transmitting $\lfloor 2^{R_0} \rfloor$ messages in a single channel use using an \textit{arbitrarily small} (but non-zero) transmit power.

Fix power $P > 0$. Let the message set size be $K = \lfloor 2^{R_0} \rfloor \in \mb Z_+$, and set $\Delta = |\alpha - 1|\sqrt{P}/K$. To encode a message $M \in [K]$, the transmit signal $X$ is set to be 
\begin{equation}
	\label{eq:prim_MQ_enc}
	X=x(M) = f_{\sf enc}(M) = M\Delta/(\alpha - 1).
\end{equation}
Note that the power constraint is satisfied, as for all $m \in [K]$,
\begin{equation}
	\label{eq:prim_MQ_pow_constr}
	|x(m)|^2 = \left| \frac{m}{\alpha - 1} \cdot \frac{|\alpha-1|\sqrt{P}}{K} \right|^2  \le P.
\end{equation}
The relay forwards the $(\Delta,K)$-modulo quantization of its observation $Y_0$:
\begin{equation}
	\label{eq:MQ_relay}
	V = f_{\sf relay}(Y_0) = \mq{Y_0}{\Delta,K}.
\end{equation}
Upon noting that the range of the $(\Delta,K)$-MQ operation has cardinality $K \le 2^{R_0}$, forwarding $V$ to the destination does not exceed the capacity of the relay link. The receiver computes the following as the decoded message: 
\begin{equation}
	\label{eq:prim_MQ_dec}
	\hat M = f_{\sf dec}(Y,V) = (V - \mq{Y}{\Delta,K}) \mod K.
\end{equation}
We now show that the decoded message $\hat M$ is always equal to the actual message $M$. Observe that
\begin{align}
	Y_0 &=  \alpha X + Z \\ &= (X + Z) + ( \alpha - 1 )X  \\ &= Y + ( \alpha - 1) \frac{M\Delta}{\alpha - 1}\\
	&= Y + M \Delta.
\end{align}
Then by Lemma \ref{lem:mq}, we obtain that 
\begin{align}
	\hat M &= (V -\mq{Y}{\Delta,K} ) \mod K \label{eq:prim_ach_4}\\&= (\mq{Y_0}{\Delta,K} - \mq{Y}{\Delta,K}) \mod K \\
	&= (\mq{Y + M\Delta}{\Delta,K} - \mq{Y}{\Delta,K})\mod K  \\
	&= M, \label{eq:prim_ach_7}
\end{align}
where \eqref{eq:prim_ach_7} is due to     \eqref{eq:colour_lemma}.

\begin{figure}[t]
	\centering
	\includegraphics[width=0.99\linewidth]{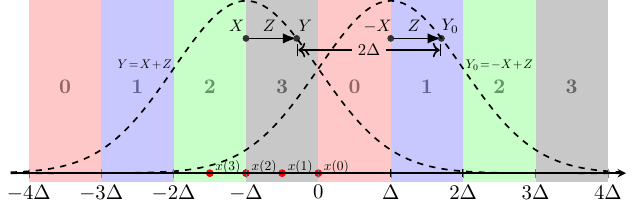}
	\caption{Illustration of the MQ coding scheme. In this example, the channel parameters are $\alpha=-1$ and $R_0 = 2$, and the message set size is $K = \lfloor 2^{R_0} \rfloor = 4$. Message $M = 2$ is transmitted. The transmitted signal is $X=x(2)=\Delta$. The observed signals $Y_0$ and $Y$ always differ by $2\Delta$. The relaying strategy is to forward $V = \mq{Y_0}{\Delta,K} = 1$. The decoder computes $\hat M = (V - \mq{Y}{\Delta,K}) \mod K = (1-3) \mod 4 = 2$, which equals $M$.}
	\label{fig:mq_illu}
\end{figure} 

To understand how MQ coding is able to exploit perfect noise correlation, despite having only a rate-limited link from the relay to the receiver, we provide a graphical illustration for the case of $\alpha=-1$ and $R_{0} = 2$ in Fig.~\ref{fig:mq_illu}. For fixed $\Delta$, the real line is partitioned into $K = 4$ different regions, indicated by different colours (red, blue, green, grey in the figure). Each region corresponds to a message $m \in \{0,1,2,3\}$ under the $(\Delta,K)$-modulo quantization operation and is a union of countably many length-$\Delta$ intervals. Suppose the message $M=m$ is transmitted, then the relay observation $Y_0$ and the receiver observation $Y$ always differ by exactly $m\Delta$, regardless of the realization of the noise $Z$. Consequently, taking the modulo-$K$ difference of $V = \mq{Y_0}{\Delta,K}$ and $\mq{Y}{\Delta, K}$ always recovers $m$ perfectly. This ensures error-free decoding, while satisfying the relay-destination link capacity $R_{0} = \log K$. Importantly, this argument holds for any choice of $\Delta > 0$, and hence the MQ scheme can operate with arbitrarily small transmit power by choosing $\Delta$ sufficiently small. 

By applying the above scalar MQ scheme independently across the channel uses, we can achieve a rate of $\log \lfloor 2^{R_0} \rfloor$.   

\begin{remark}
	\label{rem:scalar_time_sharing}
	We can match the achievable rate to $R_0$ exactly through time-sharing between a pair of MQ schemes. Recall $K = \lfloor 2^{R_0} \rfloor$. Choose $\eta \in [0,1)$ such that $\eta \log K + (1-\eta) \log (K+1) = R_0$. For any blocklength $n \in \mb Z_+$, let $n_1 = \lceil \eta n \rceil$ and $n_2 = n - n_1$. Note that 
	\begin{equation}
		\left|\frac{n_1}{n} - \eta\right| = \left|\frac{n_2}{n} - (1-\eta)\right| < \frac{1}{n}
	\end{equation}
	Apply the scalar MQ scheme with parameters $(\tilde\Delta, K)$ for $n_1$ channel uses and $(\tilde\Delta, K+1)$ for $n_2$ channel uses, where $\tilde\Delta = \frac{|\alpha - 1|\sqrt{P}}{K+1}$. It is straightforward to verify that the power constraint is satisfied. The average coding rate, which is equal to the average rate at the primitive relay link, is given by 
	\begin{align}
		\tilde{R}_0 &= \frac{n_1}{n}\log K + \frac{n_2}{n}\log(K+1) \notag \\
		&= R_0 - \frac{(\lceil \eta n \rceil - \eta n) \log(1+\frac1k) }{n}.
	\end{align}
	We observe that $R_0 - \frac{1}{n} \log(1+\frac1K)< \tilde R_0 < R_0$. On one hand, $\tilde R_0$ is bounded above by $R_0$, implying that the primitive relay capacity is not exceeded. On the other hand, $\tilde R_0$ can be made as close as possible to $R_0$ by choosing $n$ sufficiently large. Since each scalar MQ transmission incurs no error, the overall block error probability remains zero, and hence $R_0$ is achievable. 
\end{remark}

This achievability result is summarized below.

\begin{proposition}
	\label{prop:scalar_prim_relay}
	Assume $\alpha \neq 1$. MQ coding is able to achieve rate $R_0$ for the Gaussian primitive relay channel with perfectly correlated noises using any positive average transmit power $P > 0$.
\end{proposition}

\begin{remark}
	\label{rem:thp}
	The modulo operation plays an important role in a variety of practical communication strategies. A well-known example is Tomlinson-Harashima precoding (THP) \cite{tomlinson1971new,harashima1972matched}, which is a practical scalar implementation of dirty-paper coding \cite{costa1983writing}. However, the purpose of the modulo operation in THP differs from that of the MQ coding for the relay channel. In THP, the main purpose of the modulo operator is to reduce the transmit power. In contrast, the main purpose of the modulo quantizer in MQ coding for the primitive relay channel is to reduce the rate of the relay messages.
	
	Likewise, in a related setting of unlimited sampling \cite{bhandari2020unlimited}, the modulo operation is utilized to guarantee that the amplitude of the output does not exceed a certain threshold. Furthermore, a related construction appears in the instantaneous relaying setting \cite{khormuji2010instantaneous}, where the relay applies a modulo operation to a scaled version of its observation before shaping. In the asymptotic limit, the modulo operation removes the dependence between the relay input and output. These uses of modulo operation are again distinct from the purpose of MQ operation proposed in this paper.
\end{remark}

\begin{remark}
	\label{rem:wyner_ziv}
	MQ coding can be thought of as a scalar implementation of binning in Wyner-Ziv coding for the CF relay strategy. The remarkable result here is that as $P \rightarrow 0$, this scalar implementation is already capacity-achieving for the Gaussian primitive relay channel. This is due to the fact as $P \rightarrow 0$, the relay observation $Z$ and the signal at the receiver $X+Z$ become more and more highly correlated\footnote{In fact, as $P \rightarrow 0$, the residual noise power after CF by Wyner-Ziv coding also goes to zero. One can verify by computation that the effective SNR in the limit is such that the capacity of the overall channel is exactly $R_0$.}. In this case, even a scalar implementation of binning is already very effective. 
\end{remark}

\begin{remark}
	\label{remark:additive}    
	The proposed MQ coding scheme does not depend on the fact that the noise is Gaussian. In fact, it applies to any primitive relay channel with real-valued input and output alphabets and with additive noises.
\end{remark}

\subsection{Superposed MQ Coding}
\label{sec:superposed_MQ}

The scalar MQ scheme in the previous section 
achieves the relay-link capacity $R_0$ by exploiting the common noise structure. However, it does not make efficient use of the available input power. The achievable rate in fact remains $R_0$ regardless of the input power $P$. To attain the full capacity $\psi(P)+R_0$, the coding scheme must also exploit the capability of the direct AWGN channel. 
In this section, we propose the \textit{superposed MQ coding} scheme to achieve the full capacity. 

The superposed MQ coding scheme is a nested code construction that combines the scalar MQ code and a Gaussian channel code subject to an additional structural constraint. Specifically, the Gaussian channel code is supported on a scaled integer lattice. The construction resembles superposition coding for the broadcast channel \cite{cover1972broadcast}.

The message is split into two components, encoded independently using the MQ code and the Gaussian code, respectively. The transmitted signal is the superposition, i.e., entry-wise sum, of the two encoded sequences. The relay conveys the (symbol-wise) modulo quantization of the its observation to the receiver via the primitive link. By comparing the modulo quantization of the relay's and its own observations, the receiver is able to perfectly decode the message component encoded via the MQ code. The receiver then recovers the remaining message component after cancelling the decoded MQ component. This is reminiscent of successive cancellation decoding for the broadcast channel. 

This construction yields a code for the Gaussian primitive relay channel whose error probability matches that of the underlying Gaussian code, whose rate equals the sum of the Gaussian code rate and the primitive-link capacity, and whose average power constraint can be made arbitrarily close (from above) to that of the Gaussian code. We now proceed to describe this construction in detail. 


Let ${\mc C}_{\sf G}$ be a Gaussian $(2^{n R_{\sf G}}, n)$-code for the direct channel supported on a scaled integer lattice. The direct channel is an AWGN channel with input $X$ and output $Y = X + Z$. Let the codewords of ${\mc C}_{\sf G}$ be $\{x_{\sf G}^n(1), \cdots, x_{\sf G}^n(2^{n R_{\sf G}})\}$ and the decoding function be $f_{{\sf dec},{\sf G}}^{(n)}: \mb R^n \to [2^{nR_{\sf G}}]$. We impose an average power constraint on the codewords. The normalized squared norm of each codeword is upper bounded by $P_{\sf G}$, i.e., 
\begin{equation}
	\frac1n  \|x_{\sf G}^n(i)\|^2 \le  P_{\sf G},~\forall \,i \in [2^{n R_{\sf G}}].
\end{equation}
As noted above, we require that all codewords of $\mc C_{\sf G}$ lie on a scaled integer lattice. In other words, there exists some $\Delta_{\sf G} > 0$ such that \begin{equation}\label{eq:lattice}
	x_{\sf G}^n(i) \in (\Delta_{\sf G} \mb Z)^n,~\forall\, i \in [2^{n R_{\sf G}}].
\end{equation}
We later argue that there exist Gaussian codes with this structure that can achieve the AWGN channel capacity. 

\begin{remark}
	\label{rem:arb_small_delta_G}
	Given a code $\mathcal C_{\sf G}$ satisfying \eqref{eq:lattice} for some $\Delta_{\sf G}$, the constraint \eqref{eq:lattice} also holds for $\Delta_{\sf G}' = \Delta_{\sf G}/a$ for any $a \in \mb Z_+$. Therefore, the scaling factor of the integer lattice can be made arbitrarily small, which turns out to be a useful fact when dealing with the power constraint. 
\end{remark}

The average probability of error for the code $\mc C_{\sf G}$ on the direct channel is denoted by $P_e({\mc C}_{\sf G})$:
\begin{equation} \label{eq:P_e_C}
	P_e({\mc C}_{\sf G}) = 2^{-nR_{\sf G}} \sum_{i=1}^{2^{n R_{\sf G}}} \Pr\{ m \neq f_{{\sf dec},{\sf G}}^{(n)}( x_{\sf G}^n(i) + Z^n) \}.
\end{equation}

Upon the construction of one such $(2^{nR_{\sf G}}, n)$-code $\mc C_{\sf G}$ with average power constraint $P_{\sf G}$, we now construct a $(2^{n(R_{\sf G}+R_0)}, n)$-code for the Gaussian primitive relay channel with perfect noise correlation and under an average power constraint $P>P_{\sf G}$. We show that the error probability $P_e^{(n)}$ of this code is equal to $P_e({\mc C}_{\sf G})$.

Without loss of generality, we consider the case where $R_0 = \log K$ for some $K \in \mb Z_+$. (If $\log K \notin \mb Z_+$ we could use a similar time-sharing approach as in the discussion in Remark \ref{rem:scalar_time_sharing}.)
Identify the message set $[2^{n( R_{\sf G} + R_0)}]$ with the Cartesian product $[2^{n R_{\sf G}}] \times [K]^n$, i.e., consider messages of the form $(m_{\sf G}, m^n)$, where $m_{\sf G} \in [2^{n R_{\sf G}}]$ and $m^n = (m_1, \cdots, m_n) \in [K]^n$. Recall from Remark \ref{rem:arb_small_delta_G} that for a code supported on the scaled integer lattice, we can make the scaling factor arbitrarily small. Therefore we can choose $\Delta_{\sf G} > 0$ such that ${\mc C}_{\sf G} \subset (\Delta_{\sf G} \mb Z)^n$ and $P_{\sf G} \le P - \Delta_{\sf G}^2$. Let $\Delta = |\alpha - 1|\Delta_{\sf G}/K$. The proposed encoding function $f_{\sf enc}^{(n)}: [2^{n R_{\sf G}}] \times [K]^n \to \mb R^n$ is
\begin{equation}\label{eq:codebook}
	f_{\sf enc}^{(n)}(m_{\sf G}, m^n) = x^n(m_{\sf G}, m^n)=x_{\sf G}^n (m_{\sf G}) + \frac{\Delta}{\alpha-1} m^n.
\end{equation}
The above choice of $\Delta_{\sf G}$ ensures that the power constraint is satisfied for the Gaussian primitive relay channel, i.e., $\forall\,(m_{\sf G}, m^n)\in [2^{n R_{\sf G}}] \times [K]^n$,
\begin{align}
	&\frac1n\|x^n(m_{\sf G}, m^n)\|^2 \notag \\
	&\le \frac1n \left( \|x_{\sf G}^n(m_{\sf G})\|^2  + \frac{\Delta^2}{|\alpha-1|^2}  \|m^n\|^2\right) \notag \\
	&\le  P_{\sf G} + {\textstyle\frac1n} \cdot n K^2 \frac{\Delta^2}{(\alpha-1)^2} \notag \\
	&\le (P - \Delta_{\sf G}^2) + \Delta_{\sf G}^2 = P.
\end{align}

Given message $(M_{\sf G}, M^n)$, the relay observes $Y_0^n = \alpha x^n( M_{\sf G}, M^n) + Z^n$. The relay forwards the symbol-wise $(\Delta,K)$-modulo quantization $V^n$ to the receiver, where $V_{i} = \mq{Y_{0,i}}{\Delta,K}, \forall\,1\le i\le n$. The receiver receives $V^n$ from the relay and observes $Y^n =  x^n(M_{\sf G}, M^n) + Z^n$. It first aims to recover $M^n$. To do this, it 
declares
\begin{equation}
	\hat M_i =  \left(V_i - \mq{Y_i}{\Delta,K}\right) \mod K.
\end{equation}

We now show that the receiver is able to decode the $M^n$ part of the message perfectly, i.e., $\hat M_i = M_i$ holds for all $1\le i \le n$. Observe that for any such $i$,  
\begin{align}
	Y_{0,i} &= Y_{i} + (\alpha - 1) x_i(M_{\sf G},M^n) \notag \\
	&=  Y_{i} + (\alpha - 1) \left(x_{{\sf G},i}(M_{\sf G}) + \frac{\Delta}{\alpha-1} M_i \right) \notag \\
	&=   Y_{i} + (\alpha - 1) x_{{\sf G},i}(M_{\sf G}) + \Delta M_i 
\end{align}
The lattice support condition \eqref{eq:lattice} ensures $x_{{\sf G},i}(M_{\sf G}) \in \Delta_{\sf G} \mb Z$. Since $\Delta_{\sf G} = K\Delta/|\alpha-1|$, there exists some $c \in \mb Z$ such that $(\alpha-1) x_{{\sf G},i}(M_{\sf G})  = cK\Delta$. We can therefore establish the connection between $V_i$ and $Y_i$ through the following:
\begin{align}
	V_i &= \mq{Y_{0,i}}{\Delta,K} \notag \\ 
	&= \mq{Y_{i} + (\alpha - 1) x_{{\sf G},i}(M_{\sf G}) + \Delta M_i }{\Delta,K} \notag \\
	&= \mq{Y_{i} +  ck\Delta + \Delta M_i }{\Delta,K} \notag \\
	&= \mq{Y_{i} + \Delta M_i }{\Delta,K},
\end{align}
where the last step is due to \eqref{eq:lattice_colour_lemma} in Lemma \ref{lem:mq}. Using a similar argument as in \eqref{eq:prim_ach_4}-\eqref{eq:prim_ach_7}, we have
\begin{align}
	\hat M_i &= (V_i -\mq{Y_i}{\Delta,K} ) \mod K \notag \\
	&= (\mq{Y_{i} + \Delta M_i}{\Delta,K} - \mq{Y}{\Delta,K})\mod K \notag  \\
	&= M_i,
\end{align}
concluding that $\hat M_i = M_i$ exactly.

Upon recovering $M^n$, the receiver can subtract $\frac{\Delta}{\alpha-1} M^n$ from $Y^n$ and get 
\begin{equation}
	Y_{\sf G}^n = Y^n - \frac{\Delta}{\alpha-1} M^n = \beta x_{\sf G}^n(M_{\sf G}) + Z^n.
\end{equation}
Observe that $Y_{\sf G}^n$ is exactly the output of the direct channel when the channel input is $x_{\sf G}^n(M_{\sf G})$. By \eqref{eq:P_e_C}, the receiver can decode $M_{\sf G}$ correctly with error probability $P_e(\mc C_{\sf G})$. Since the recovery of $M^n$ is perfect, the overall average error probability for this superposed modulo quantization coding scheme is $P_e^{(n)} = P_e(\mc C_{\sf G})$. 

Based on the above rate expression and error probability analysis, to complete the proof that this superposed MQ coding scheme indeed achieves the overall capacity \eqref{eq:capacity}, it remains to show the existence of a sequence of capacity-achieving codes supported on a scaled integer lattice for the direct channel. Using a random coding argument, the problem can be reduced to finding a sequence of distributions $\{\mu^{(j)}\}_{j \in \mb Z_+}$ with support on some one-dimensional lattice and second moment strictly less than $P$, such that the following mutual information converges to the Gaussian channel capacity:
\begin{equation}\label{eq:mi_convergence}
	\lim_{j \to \infty} I\left(X_{\sf G}^{(j)}; X_{\sf G}^{(j)} + Z\right) \to \psi(P), ~~ X_{\sf G}^{(j)} \sim \mu^{(j)}.
\end{equation}
One choice of such sequence of distributions has been explored in \cite[Section VII]{wu2010impact}. In particular, let $B^{(j)}\sim {\rm Binomial}(j-1,1/2)$, then $\mu^{(j)}$ is the distribution of 
\begin{equation}
	X_{\sf G}^{(j)} = c_j \left(B^{(j)} - \frac{j-1}{2}\right),
\end{equation}
where $c_j > 0$ is chosen such that $\mb E[(X_{\sf G}^{(j)})^2]$ approaches $P$ from below as $j \to \infty$. The convergence in \eqref{eq:mi_convergence} follows from the asymptotic normality of binomial distributions. We summarize this achievability result in the following theorem.

\begin{theorem}
	\label{thm:prim_MQ_capacity}
	Assume $\alpha \neq 1$. The superposed MQ coding scheme is able to achieve the capacity $C_{\triangleleft}(R_0) = \psi(P) + R_0$ for the Gaussian primitive relay channel with perfectly correlated noises using any positive average transmit power $P > 0$.
\end{theorem}

\begin{remark}
	The code construction in the proof above involves both a coarser lattice and a finer lattice. This is reminiscent of nested lattice coding \cite{zamir2002nested,yu2005trellis}. However, the roles of the lattices in superposed MQ coding are quite different from that of nested lattice coding. In nested lattice coding, only the finer lattice is used for channel coding. The coarser lattice is used for source coding. In contrast, in the superposed MQ coding scheme, both lattices are used for channel coding. 
\end{remark}

\begin{remark}
	The superposed MQ coding scheme is not the only way to achieve the capacity of this channel via MQ; it is also possible to do so through time-sharing between the scalar MQ scheme and the usual AWGN channel codes.
\end{remark}  

In the following we describe such an alternative construction based on time-sharing. Let $n$ be the total blocklength of the transmission. In this scheme, the relay link operates \textit{only} in the first channel use at the \textit{high} rate of $nR_0$. For this single channel use, we implement the scalar MQ scheme discussed in Section \ref{sec:prim_MQ}. In the remaining channel uses the relay link is inactive, and the channel becomes a regular AWGN channel. For these channel uses, we use an AWGN channel code of length $(n-1)$ and average power constraint $P$.

By Proposition \ref{prop:scalar_prim_relay},  there exists an error-free $(\lfloor 2^{nR_0} \rfloor,1)$-code for this channel when the relay rate is $nR_0$. Furthermore, this code can be made to satisfy average power constraint $P$. The overall transmit power is then upper bounded by $nP$.
Meanwhile, the rate $\psi(P)$ is achievable for the AWGN channel over the $n-1$ channel uses when the relay is inactive. The time-shared rate is lower bounded as
\begin{align}
	\label{eq:time_shared_rate}
	\frac{1}{n} \log \lfloor 2^{nR_0} \rfloor  \ge  \psi(P) + R_0 - O\left(\frac{1}{n}\right),
\end{align}
which approaches $C_\triangleleft = \psi(P) + R_0$ as $n \to \infty$. Hence, this construction also achieves the channel capacity.

We note here that although the time-sharing argument above is conceptually simpler than the earlier proof involving Gaussian AWGN code supported on the lattice, the MQ code in the time-sharing argument would require high-precision implementation when $n$ is large, because such an MQ code requires a lattice with a spacing that scales as $2^{-nR_0}$. In contrast, the earlier proof only requires a lattice of a spacing proportional to $2^{-R_0}$, independent of $n$. Nevertheless, due to the conceptual simplicity of the time-sharing argument, we will adopt this argument when deriving results for the diamond channel in Section \ref{sec:diamond}.



\begin{remark}
	The time-sharing argument in the above proof resembles the flash helping approach introduced in \cite{bross2020decoder}. In flash helping, the relay is active only in $\tau$-fraction of time, where $\tau > 0$ tends to zero as the blocklength $n$ goes to infinity. The relay helps by quantizing a length-$\tau n$ subsequence of its observation into $nR_0$ bits. This requires a rate-distortion code. The error probability of the rate-distortion code vanishes only as the length of the subsequence grows without bound, i.e., $\tau n \to \infty$. In comparison, the MQ coding scheme does not make any error even when the relay is active only in one channel use.  
\end{remark}

\subsection{Multiple Relays}

In this subsection, we generalize the MQ coding framework to the Gaussian primitive relay channel with multiple relays under perfect noise correlation. In particular, we show that the capacity of this channel can also be achieved via an MQ-based scheme.

For a positive integer $L \in \mb Z_+$, the Gaussian primitive $L$-relay channel with perfect noise correlation is depicted in Fig.~\ref{fig:prim_L_relay}. Each relay observes a noisy version of the input $X$ according to $Y_\ell = \alpha_\ell X + Z$ for $1 \le \ell \le L$, where the $\alpha_\ell$'s are the respective channel gains at the relays, and $Z \sim \mc N(0,1)$ is the common Gaussian noise affecting all relays under the perfect correlation assumption. Similar to the single-relay case, the input is subject to an average-power constraint $P$. The destination observes $Y = X + Z$, together with the messages forwarded from the relays over noiseless bit pipes of capacity $R_1, \cdots, R_{L}$. As before, we assume that $\alpha_\ell \neq 1$ for all $1 \le  \ell \le L$, for the same reasons discussed in Section \ref{sec:prim_relay_model}. The capacity of this channel is denoted $C_{\triangleleft,L}(R_1, \cdots, R_L)$.

\begin{figure}[t]
	\centering
	\begin{tikzpicture}[auto, thick, node distance=1.5cm,scale = 0.7]
		
		\node at (-4.4,0) (msg1) {$M$};
		\node[draw, rectangle,scale = 0.75] (enc) at (-2.9,0) {\bf Enc};
		\node at (-1.4,0) (input1) {$X^n$};
		
		\node[draw, circle,scale = 0.5] (sum1) at (0.45,0) {\Large +};
		\node[draw, rectangle,scale = 0.75] (relay) at (4.2,1.25) {\bf Relay\,$L$};
		\node at (2.2,0) (output1) {$Y^n$};
		\node[draw, rectangle,scale = 0.75] at (5.7,0) (dec) {\bf Dec};
		\node at (7.2,0) (msg2) {$\hat M$};
		
		\draw[-Latex] (input1) -- (sum1);

		\node[scale = 0.95] at (0.45,2.25) (zinput0) {$Z^n$};
		\node[scale = 0.95] at (0.45,-1.05) (zinput1) {$Z^n$};
		\draw[-Latex] (0.45,2.05) -- ++ (0.,-0.55);
		\draw[-Latex] (0.45,-0.8) -- ++ (0.,0.55);
		\draw[-Latex] (0.45,4.2) -- ++ (0.,-0.55);
		\node[draw, circle,scale = 0.5] (sum3) at (0.45,3.4) {\Large +};
		
		\node[draw, circle,scale = 0.5] (sum2) at (0.45,1.25) {\Large +};
		\draw[-Latex] (input1) -- node[above,scale = 0.8] {$\alpha_L~$} ++ (sum2);
		\draw[-Latex] (input1) -- node[left,scale = 0.8] {$\alpha_1$} ++ (sum3);

		\node at (2.2,1.25) (output2) {$Y_L^n$};

		\node at (2.2,3.4) (output3) {$Y_1^n$};
		\node[scale = 0.95] at (0.45,4.4) (zinput0) {$Z^n$};
		\node[draw, rectangle,scale = 0.75] (relay3) at (4.2,3.4) {\bf Relay\,$1$};
		
		\node at (2.2,2.5) {$\vdots$};

		\draw[-Latex] (sum2) -- (output2);
		\draw[-Latex] (sum1) -- (output1);
		\draw[-{Latex[fill=none]}, double] (relay) --  (dec);
		\draw[-Latex] (msg1) -- (enc);
		\draw[-Latex] (enc) -- (input1);
		\draw[-Latex] (output1) -- (dec);
		\draw[-Latex] (dec) -- (msg2);
		\draw[-Latex] (output2) -- (relay);
		\draw[-Latex] (sum3) -- (output3);
		\draw[-Latex] (output3) -- (relay3);
		\draw[-{Latex[fill=none]}, double] (relay3) --  (dec);
		\node[scale = 0.8] at (4.5,0.52) {$nR_L$};
		\node[scale = 0.8] at (5.3,2) {$nR_1$};

	\end{tikzpicture}
	\caption{Gaussian primitive $L$-relay channel with perfectly correlated noises at the receiver and the relays.}
	\label{fig:prim_L_relay}
\end{figure}
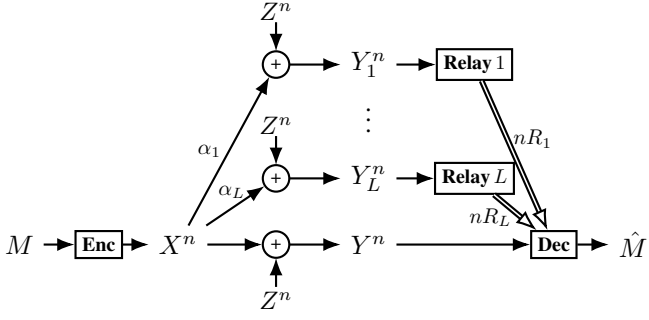

From the cut-set bound, it is immediate that the capacity of this channel cannot exceed the sum of the AWGN channel capacity and all the primitive link capacities, namely $\psi(P) + \sum_{\ell = 1}^L R_\ell$. We now construct a time-sharing scheme based on MQ coding that attains this upper bound.

Let $n \in \mb Z_+$ be the transmission blocklength with $n \gg L$. In each of the first $L$ channel uses, only one of the primitive relay links is active while all other remain inactive. Specifically, during the $\ell$-th channel use, we apply the scalar MQ scheme where $\ell$-th relay link operates at the rate of $nR_\ell$. Using the scalar MQ coding scheme in Section \ref{sec:prim_MQ}, $2^{nR_\ell}$ messages can be transmitted from the transmitter to the receiver without error in this channel use. For all the remaining $(n-L)$ channel uses, no relay is active and the channel reduces to an AWGN channel, and as $n \to \infty$, an additional $\psi(P)$ bits of information can be conveyed reliably. Summing up, the overall achievable rate of this scheme is equal to $\psi(P) + \sum_{\ell = 1}^L R_\ell$, which is exactly the cut-set upper bound. This leads to the following corollary.

\begin{corollary}
	Assume $\alpha_\ell \neq 1$ for all $1\le \ell \le L$. The capacity of the Gaussian primitive $L$-relay channel is 
	\begin{equation}
		\label{eq:multi_relay_capacity}
		C_{\triangleleft,L}(R_1, \cdots, R_L) = \psi(P) + \sum_{\ell=1}^L R_\ell.
	\end{equation}
	It is achievable by using an MQ-based time-sharing scheme. 
\end{corollary}

\section{Gaussian Primitive Diamond Channel with Perfect Noise Correlation}
\label{sec:diamond}

\subsection{Channel Model}
\label{sec:diamond_model}

The Gaussian primitive diamond channel 
is depicted in Fig.~\ref{fig:diamond}. This channel has been studied in \cite{sanderovich2008communication}. Given input $X$, the two relays observe $Y_{1} = \alpha X + Z_{1}$ and $Y_{2} =\beta X + Z_{2}$ respectively, where $Z_{1}$ and $Z_{2}$ are the respective noises at the relays, and $\alpha, \beta \in \mb R$ are the corresponding channel gains. In this section, we assume that the noise variables are perfectly (and positively) correlated standard Gaussian variables, namely, 
\begin{equation}
	\label{eq:diam_noise_distr}
	\begin{bmatrix}
		Z_1 \\ Z_2
	\end{bmatrix} \sim \mc N\left( \mf 0, \begin{bmatrix}
		1 & 1 \\ 1 & 1
	\end{bmatrix}\right).
\end{equation}
For simplicity of exposition, we can introduce a new Gaussian random variable $Z$ equal to both $Z_1$ and $Z_2$, thus we can write $Y_1 = \alpha X +Z$ and $Y_2 = \beta X + Z$.
The two relays communicate with a central receiver through primitive relay links of capacities $R_1$ and $R_2$ bits per channel use, respectively. Importantly, there is no direct link from the transmitter to the central receiver; the receiver obtains its information exclusively through the relays. 

\begin{figure}[t]
	\centering
	\begin{tikzpicture}[auto, thick, node distance=1.5cm,scale =0.75]
		\node at (-4.2,0) (msg1) {$M$};
		\node[draw, rectangle,scale = 0.8] (enc) at (-2.8,0) {\bf Enc};
		\node at (-1.4,0) (input) {$X^n$};
		
		\node[draw, circle,scale = 0.5] (sum1) at (0.4,1.5) {\Large +};
		\node[draw, circle,scale = 0.5] (sum2) at (0.4,-1.5) {\Large +};
		\node[scale = 0.95] at (0.45,2.5) (zinput1) {$Z_1^n$};
		\node[scale = 0.95] at (0.45,-2.6) (zinput2) {$Z_2^n$};
		\node at (2,1.5) (output1) {$Y_1^n$};
		\node at (2,-1.5) (output2) {$Y_2^n$};

		\node[draw, rectangle,scale = 0.8] (relay1) at (3.7,1.5) {\bf Relay 1};
		\node[draw, rectangle,scale = 0.8] (relay2) at (3.7,-1.5) {\bf Relay 2};
		\node [draw, rectangle,scale = 0.8] (dec) at (5.5,0) {\bf Dec};
		\node at (6.8,0) (msg2) {$\hat M$};

		\draw[-Latex] (msg1) -- (enc);
		\draw[-Latex] (enc) -- (input);
		\draw[-Latex] (input) -- node[above,scale = 0.8] {$\alpha$}   (sum1);
		\draw[-Latex] (input) -- node[below,scale = 0.8] {$\beta$}   (sum2);

		\draw[-Latex] (0.4,2.3) -- ++ (0.,-0.55);
		\draw[-Latex] (0.4,-2.3) -- ++ (0.,0.55);

		\draw[-Latex] (sum2) -- (output2);
		\draw[-Latex] (sum1) -- (output1);
		\draw[-{Latex[fill=none]}, double] (relay1) --  (dec);
		\draw[-{Latex[fill=none]}, double] (relay2) --  (dec);
		
		\draw[-Latex] (output1) -- (relay1);
		\draw[-Latex] (output2) -- (relay2);
		\draw[-Latex] (dec) -- (msg2);
		\node[scale = 0.9] at (5,0.95) {$nR_1$};
		\node[scale = 0.9] at (5,-1) {$nR_2$};
	\end{tikzpicture}
	\caption{Gaussian primitive diamond channel with correlated noises at the two relays.}
	\label{fig:diamond}
\end{figure}
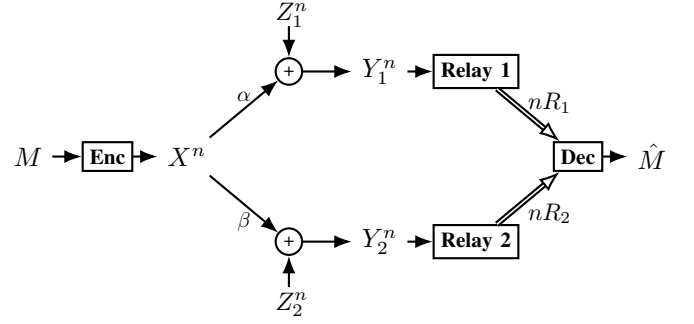

A $(2^{nR}, n)$-code for this channel is specified by the quadruple $(f_{\sf enc}^{(n)}, f_{\sf relay,1}^{(n)}, f_{\sf relay,2}^{(n)}, f_{\sf dec}^{(n)})$. The encoding function $f_{\sf enc}^{(n)}:  [2^{nR}] \to \mb R^n$ maps a message $M$ to a codeword $X^n=x^n(M)$ such that the average power constraint $\frac1n\|x^n(m)\|^2 \le P$ is satisfied for all $m \in [2^{nR}]$. 

For $j = 1,2$, the $j$-th relay forwards an $nR_j$-bit message based on its observation $V_j = f_{\sf relay,j}^{(n)}(Y_j^n)$ to the receiver through a noiseless link, where $f_{\sf relay,j}^{(n)}: \mb R^n \to [2^{nR_j}]$ is the corresponding relay function. The decoding function $f_{\sf dec}^{(n)}: [2^{nR_1}] \times [2^{nR_2}] \to [2^{nR}]$ for this channel takes input $V_1^n$ and $V_2^n$ and outputs the decoded message $\hat M$. The definitions of achievable rate and channel capacity are similar to those in Section \ref{sec:prim_relay_model}. We denote the capacity of the Gaussian primitive diamond channel by $C_{\diamond}(R_1,R_2)$.

Consider first the special case where the two relays have identical gains, i.e. $\alpha = \beta$. In this case, the relay observations are exactly the same, and thus the channel model reduces to a two-hop channel consisting of three nodes: the intermediate node observes $\alpha X + Z$ and is able to communicate to the receiver through a noiseless bit pipe of capacity $R_1 + R_2$. The capacity of this simplified model is just the minimum between the AWGN channel capacity and the capacity of the bit pipe, i.e., 
\begin{equation}
	\label{eq:capacity_diamond_alpha=beta}
	C_\diamond(R_1,R_2) = \min\{ \psi(\alpha^2 P), R_1 + R_2 \}.
\end{equation}
The converse follows directly from the cut-set bound, while the achievability is established via DF, i.e., the intermediate node fully decodes the transmitted message and forwards it to the receiver. This case is trivial. For the remainder of this paper, we focus on the more interesting case where $\alpha \neq \beta$.

In the following, we first give an overview of the known capacity bounds of this channel. We then show that, with a slight modification, the scalar MQ coding scheme introduced in the previous section can be extended to the diamond channel. This scheme is able to achieve a rate of $\min\{R_1, R_2\}$, namely, the smaller of the two relay link capacities. Finally, we propose a way to combine this MQ scheme with DF to yield a new achievability bound for the diamond channel. We show that this bound is tight in certain parameter regimes, thereby revealing a new capacity result.


\subsection{Known Capacity Bounds}
\label{sec:known_bd_diam}

In contrast to the Gaussian primitive relay channel, the capacity of Gaussian primitive diamond channel is not yet known, even if the noise correlation is perfect. In this subsection, we summarize the best-known upper and lower bounds for this channel model. These results serve as benchmarks for analyzing the MQ coding schemes and evaluating their achievable rates in Section \ref{sec:diamond_MQ}.

For the converse, the cut-set bound for the diamond channel can be constructed by a minimization over the broadcast cut-set bound, the multiple-access cut-set bound, and two cross cut-set bounds \cite[Chapter 3.2]{schein2001distributed}. In the present setting where the relay observations have a common noise component, the channel input can be written as a function of the relay outputs, hence the broadcast cut-set bound becomes vacuous. 

For the achievability, we consider lower bounds derived from the DF and CF schemes. In the DF scheme, the relays recover the message sent by the transmitter. Specifically, the channel from the transmitter to the two relays can be regarded as a degraded broadcast channel \cite[Chapter 5.4]{el2011network}, in which a common message is decoded by both relays, and an additional private message is decoded by the stronger relay. The decoded messages are then forwarded to the central receiver, subject to the constraint that the rate of each forwarded message does not exceed the capacity of the corresponding relay link. 

In the CF scheme, the relays do not perform decoding; instead, each relay compresses its observed sequence via typicality-based vector quantization followed by binning. The bin indices are then forwarded to the receiver. The receiver recovers the transmitted message by performing joint decompression and decoding. This approach has been studied in \cite{sanderovich2008communication} for the case of independent noises, and a correlated-noise variant has been introduced in \cite{katz2024gaussian}.

The bounds discussed above are listed in the proposition below. For completeness, we provide a proof in Appendix \ref{app:bd_diamond}.

\begin{proposition}
	\label{prop:bd_diamond}
	Assume $\alpha \neq \beta$ and $|\alpha| \ge |\beta|$. The capacity of the Gaussian primitive diamond channel with perfectly correlated noises is bounded from above by the following expression of the cut-set bound:
	\begin{equation} \label{eq:diamond_cutset}
		\textnormal{Cutset}_{\diamond}(R_1, R_2) =  \min\left\{\begin{matrix}
			R_1 + R_2 \\ R_1 + \psi(\beta^2 P)\\ R_2 + \psi(\alpha^2 P)
		\end{matrix}\right\}.
	\end{equation}
	The capacity is bounded from below by the following expression of DF lower bound:
	\begin{align} 
		\label{eq:diamond_DF}
		R_{\diamond,\textnormal{DF}}(R_1,R_2) =\max_{\gamma \in [0,1]} \min \left\{\begin{matrix}
			R_1 + R_2 \\ R_1 + \psi\left(\frac{\bar\gamma \beta^2 P }{\gamma \beta^2 P + 1}  \right) \\ \psi (\gamma \alpha^2 P )  + \psi \left( \frac{\bar\gamma \beta^2 P }{\gamma \beta^2 P + 1}\right)
		\end{matrix} \right\},
	\end{align}
	where $\gamma$ is the power-splitting parameter in the underlying superposition coding scheme and $\bar \gamma = 1-\gamma$. Specifically, a fraction $\gamma$ of the total transmit power is allocated to the private message decoded only by the stronger relay, while the remaining fraction $\bar\gamma$ is allocated to the common message decoded by both relays.
	
	The capacity is also bounded from below by the following expression of CF lower bound:
	\begin{align} 
		\label{eq:diamond_CF}
		&R_{\diamond,\textnormal{CF}}(R_1,R_2)  \notag \\&=  \max_{N_1,N_2 > 0} \min \left\{\begin{matrix}
			\frac12\log \left(1+\frac{((\alpha - \beta)^2 + \alpha^2 N_2 + \beta^2 N_1 )P}{N_1N_2 + N_1 + N_2}\right) \\ R_1 + \frac{1}{2}\log \frac{N_1(\beta^2 P + 1 + N_2)}{N_1N_2 + N_1 +N_2}  \\
			R_2 + \frac{1}{2}\log \frac{N_2(\alpha^2 P + 1 + N_1)}{N_1N_2 + N_1 +N_2} \\ 
			R_1 + R_2 + \frac{1}{2}\log \frac{N_1 N_2}{N_1N_2 + N_1 +N_2}
		\end{matrix} \right\}.
	\end{align}
\end{proposition}

\begin{remark}
	The CF scheme underlying \eqref{eq:diamond_CF} uses Gaussian test channels $\hat Y_i=Y_i+U_i$, where $U_i\sim\mathcal N(0,N_i)$ is the independent quantization noise used by relay $i$ for $i = 1,2$. Each relay compresses its observation and sends the corresponding bin index to the destination, which performs joint decompression and decoding. Notably, the first term in \eqref{eq:diamond_CF} is $I(X;\hat Y_1,\hat Y_2)$ and plays a role analogous to the broadcast
	cut-set term $I(X;Y_1,Y_2)$. Under perfect noise correlation,
	$X$ is a deterministic function of $(Y_1,Y_2)$, so $I(X;Y_1,Y_2)=\infty$ and the
	broadcast cut does not constrain the cut-set bound. In CF, however,
	the added independent quantization noises destroy this deterministic functional
	relationship, making $I(X;\hat Y_1,\hat Y_2)$ finite and potentially
	rate-limiting.
\end{remark}

\subsection{MQ Coding}
\label{sec:diamond_MQ}

In Proposition \ref{prop:scalar_prim_relay}, we introduce a zero-error scalar coding scheme for the Gaussian primitive relay channel with perfect noise correlation. We now extend the underlying ideas to develop a coding scheme for the Gaussian primitive diamond channel with perfect correlation. Similar to Proposition~\ref{prop:scalar_prim_relay}, the resulting scheme achieves a rate equal to the smaller of the two relay-link capacities.

To understand this extension, we briefly recall the code construction as described in the proof of Proposition \ref{prop:scalar_prim_relay}. In the decoding step, the receiver does not require the complete observation to decode the transmitted message; rather, a modulo quantization of the observation suffices. The adaptation of the scheme to the diamond channel is analogous: each relay forwards a modulo-quantized version of its observation to the central destination. In order to decode the message, the parameters of the modulo quantization at each of the two relays must be identical. In addition, the decodability condition of MQ requires that the logarithm of the message set size does not exceed either relay-link capacity. The resulting achievable rate for MQ coding applied to the Gaussian primitive diamond channel is summarized below.

\begin{proposition}
	\label{prop:scalar_diamond_relay}
	Assume $\alpha \neq \beta$. For the Gaussian primitive diamond channel with perfectly correlated noises at the relays, there exists a $(\lfloor 2^{R_{\textnormal{min}}} \rfloor,1)$ code with zero error probability with any positive average transmit power $P > 0$, where $R_{\textnormal{min}}=\min(R_1,R_2)$.
\end{proposition}

\begin{IEEEproof}
	The following scheme is a variant of the MQ coding scheme described in the proof of Proposition \ref{prop:scalar_prim_relay}. Let the message set be of size $K = \lfloor 2^{R_{\textnormal{min}}} \rfloor$.
	Set $\Delta = |\beta - \alpha|\sqrt{P}/K$. The encoding function is $f_{\sf enc}(m) = m\Delta/(\beta - \alpha)$. The power constraint is satisfied, since for all $m \in [K]$
	\begin{equation}
		|x(m)|^2 = |f_{\sf enc}(m)|^2 = \left| \frac{m}{\beta - \alpha} \cdot \frac{|\beta - \alpha|\sqrt{P}}{K} \right|^2  \le P.
	\end{equation}
	
	To encode a message $M \in [K]$, the transmitter sends $X=M\Delta/(\beta - \alpha)$. The two relays both compute the following function based on their observations: $f_{\sf relay,i}(y) = \mq{y}{\Delta,K}$ for $i = 1,2$. In other words, each relay forwards the
	$(\Delta,K)$-modulo quantization of its observation to the receiver:
	\begin{align}
		V_1 =  \mq{Y_1}{\Delta,K}= \mq{\alpha X + Z}{\Delta,K},\\
		V_2 = \mq{Y_2}{\Delta,K} = \mq{\beta X + Z}{\Delta,K}.
	\end{align}
	Note that since $K \le \lfloor 2^{R_i} \rfloor$ for $i = 1,2$, forwarding $V_i$ to the destination satisfies the capacity limit of the noiseless link between the $i$-th relay and the destination.
	
	The decoding function $f_{\sf dec}(v_1,v_2)$ is defined as:
	\begin{equation}
		f_{\sf dec}(v_1,v_2) = (v_2 - v_1) \mod K.
	\end{equation}
	We now show that $\hat M = (V_2 - V_1) \mod K$ is equal to $M$. Observe that
	\begin{align}
		Y_2 &= \beta X + Z \notag \\
		&= (\beta - \alpha) X + \alpha X+ Z \notag \\
		&= (\beta - \alpha) \frac{M\Delta}{\beta - \alpha} + Y_1 \notag \\
		&= M\Delta + Y_1.
	\end{align}
	Then,
	\begin{align}
		\hat M &= (V_2 - V_1 ) \mod K \notag \\
		&= (\mq{Y_2}{\Delta,K} - \mq{Y_1}{\Delta,K}) \mod K \notag \\
		&= (\mq{Y_1 + M\Delta}{\Delta,K} - \mq{Y_1}{\Delta,K})\mod K  = M, \label{eq:MAC_ach}
	\end{align}
	by Lemma \ref{lem:mq}. This concludes the proof. 
\end{IEEEproof}

\begin{remark}
	Similar to the discussion in Section \ref{sec:prim_MQ}, we can achieve the rate $\log \lfloor 2^{R_{\rm min}} \rfloor$ by applying the scalar MQ scheme introduced in Proposition \ref{prop:scalar_diamond_relay} independently across channel uses. By employing time-sharing among MQ schemes with different parameters (cf. Remark \ref{rem:scalar_time_sharing}), we obtain the following achievable rate
	\begin{equation}
		\label{eq:naive_MQ_diam}
		R_{\diamond,\textnormal{MQ}}(R_1,R_2) = R_\textnormal{min}.
	\end{equation}
\end{remark}

From the above discussion, we see that any rate below $R_{\text{min}}$ can be achieved using MQ coding with \emph{arbitrarily small} (non-zero) power. 
Moreover, MQ coding is asymptotically optimal in the low power limit. 
This is because 
\begin{align}
	&\lim_{P \to 0} \textnormal{Cutset}_{\diamond}(R_1, R_2) \notag \\
	&=  \lim_{P \to 0} \min\{R_1+R_2, R_1 + \psi(\beta^2 P),R_2 + \psi(\alpha^2 P)\} \notag \\
	&= \min\{R_1+R_2, R_1 ,R_2 \} \notag \\
	&= R_{\rm min}.
\end{align}
Thus, $R_{\rm min}$ is the asymptotic capacity of the Gaussian primitive diamond channel as $P \to 0$. Just as for the Gaussian primitive relay channel, the capacity is bounded away from zero even when $P$ is infinitesimally small.

In the following, we derive higher achievable rates by combining MQ with other types of relay strategies.

\subsection{Hybrid MQ-DF Coding}
\label{sec:MQ_DF}

When the relay-link capacities are unequal, MQ coding can only exploit $R_{\text{min}}$ bits per channel use of the higher-capacity relay link.  This leaves the balance of the higher capacity link unused. In addition, as the MQ rate is independent of the transmit power, it cannot leverage higher operating SNRs.

The aforementioned limitations motivate the following length-$n$ block hybrid scheme that combines MQ and DF. We partition the block into one single channel use for MQ and $n-1$ channel uses for DF. 
We use a constant transmit power $P$ in each channel use. The total energy $nP$ is thereby partitioned so that $P$ is allocated to MQ and $(n-1)P$ to DF.
Finally, we partition the relay-link capacities such that each relay uses $n\rho \in [0, nR_{\text{min}}]$ bits for MQ and the excess capacity is used for DF. These choices are made because MQ requires only one channel use and arbitrarily small power to achieve a rate equal to the minimum of the two relay link rates. Here, we do not necessarily allocate the full $nR_{\text{min}}$ to MQ (i.e., $n\rho < nR_{\rm min}$ is allowed), because it may be advantageous to allocate some additional relay link rate to DF. In the DF phase, we further use parameter $\gamma \in [0,1]$ to control the power allocation between the private and common messages. 

Per the discussion above, the MQ phase of the proposed scheme takes place only in the first channel use. In this single channel use, each relay sends $n\rho$ bits to the receiver.  According to Proposition~\ref{prop:scalar_diamond_relay}, $\log \lfloor 2^{n\rho}\rfloor = n(\rho - o(1))$ bits of information can be conveyed at zero error using any non-zero power.

For the remaining $(n-1)$ channel uses, we employ a DF scheme with power $P$ per channel use. 
Note that the Gaussian primitive diamond channel can be regarded as a degraded Gaussian BC cascaded with a noiseless digital MAC. To apply DF, we use superposition coding with a common and a private message for the Gaussian BC component. The private message is encoded using a Gaussian codebook with i.i.d. $\mc N(0, \gamma P)$ entries. The common message is encoded using a Gaussian codebook with i.i.d. $\mc N(0, (1-\gamma) P)$ entries. Both relays decode the common message. The stronger relay (the relay with the larger-magnitude channel gain) also  decodes the private message. The stronger relay forwards its private message; the remaining capacity of the stronger relay's digital link is used to forward part of the common message. The weaker relay forwards the balance of the common message. 

The receiver decodes the MQ message at zero-error, and obtains the DF messages from the digital links. The following theorem gives the rate achieved using this scheme.

\begin{theorem}
	\label{thm:diam_MQ_DF}
	Assume $|\alpha| \ge |\beta|$ and $\alpha \neq \beta$. The capacity of the Gaussian primitive diamond channel with perfectly correlated noises at the relays is bounded from below by
	\begin{align}
		\label{eq:R_MQ_DF}
		& R_{\diamond,\textnormal{MQ-DF}}(R_1, R_2)  \notag \\
		&=\!\max_{\substack{\gamma \in [0,1]\\ \rho \in [0,R_{\textnormal{min}}]}}\!\!\! \min \left\{ 
		\begin{matrix}
			R_1 + R_2 - \rho\\
			R_1 + \psi(\frac{\bar{\gamma} \beta^2 P}{\gamma \beta^2 P + 1})\\
			\psi(\gamma \alpha^2 P) + \psi(\frac{\bar{\gamma} \beta^2 P}{\gamma \beta^2 P + 1}) + \rho
		\end{matrix}\right\},
	\end{align}
	where $\bar \gamma = 1- \gamma$. 
\end{theorem}


\begin{IEEEproof}[Proof]
	Fix $\gamma,\rho$. Denote the common and private message rate (normalized by $n$) of the DF component by $R_{\text{DF-comm}}$ and $R_{\text{DF-priv}}$ respectively. As $n \to \infty$, the superposition coding inner bound of BC gives
	\begin{align}  
		R_{\text{DF-comm}} &\le \psi\left(\frac{(1-\gamma) \beta^2 P}{\gamma \beta^2 P + 1}\right),\label{eq:MQ_DF_rate_constr_DF_11} \\[1mm]
		R_{\text{DF-priv}}& \le \psi(\gamma \alpha^2 P). \label{eq:MQ_DF_rate_constr_DF_12}
	\end{align}
	Meanwhile, in order to forward the MQ message and the DF messages to the receiver via the relay links, the following bounds hold as $n \to \infty$:
	\begin{align}
		R_{\text{DF-priv}} + \rho& \le R_1, \label{eq:MQ_DF_rate_constr_21}\\
		R_{\text{DF-priv}} + R_{\text{DF-comm}} + 2\rho &\le R_1 + R_2. \label{eq:MQ_DF_rate_constr_22}
	\end{align}
	The total message rate for this scheme is $\rho + R_{\text{DF-priv}} + R_{\text{DF-comm}}$. By eliminating the parameters $R'_{\text{DF-comm}}$ and $R'_{\text{DF-priv}}$ via Fourier-Motzkin elimination, and further optimizing over the parameters $\gamma$ and $\rho$, we conclude that the rate expression in \eqref{eq:R_MQ_DF} is achievable.
\end{IEEEproof}

\begin{remark}
	The above proof relies on a time-sharing argument, analogous to the alternative proof of Theorem \ref{thm:capacity_prim_relay}. Theorem \ref{thm:diam_MQ_DF} could also be proved using a nested code construction. Specifically, it suffices to construct a DF code supported on a scaled integer lattice and superimpose on top of the one-shot MQ code described in the proof of Proposition \ref{prop:scalar_diamond_relay}. The resulting proof closely parallels our first proof of Theorem \ref{thm:capacity_prim_relay}. We therefore omit the details.
\end{remark}

The achievability bound \eqref{eq:R_MQ_DF} in Theorem \ref{thm:diam_MQ_DF} is tight in several parameter regimes. In particular, it coincides with the 
cut-set upper bound \eqref{eq:diamond_cutset} at high SNR. Moreover, when the two relay-link capacities differ, the lower and upper bounds also coincide at low SNR. 
The following theorem formalizes this observation. We provide the proof in Appendix \ref{app:capacity_diamond}.

\begin{theorem}
	\label{thm:capacity_diamond}
	Assume $|\alpha| \ge |\beta|$. The capacity of the Gaussian primitive diamond channel with perfectly correlated noises is characterized in the following cases:
	\begin{enumerate}[(a)]
		\item For $P \ge \frac{2^{2R_2} - 1}{\beta^2} + \frac{2^{2R_2}(2^{2R_1} - 1)}{\alpha^2}$,
		\begin{equation}
			\label{eq:C_diam_high_SNR}
			C_{\diamond}(R_1,R_2) = R_1 + R_2.
		\end{equation}
		\item If $R_1 > R_2$, then for $P \le  \frac{2^{2(R_1 - R_2)} - 1}{\alpha^2}$, 
		\begin{equation}
			\label{eq:C_diam_low_SNR_1}
			C_{\diamond}(R_1,R_2) = R_2+\psi(\alpha^2 P).
		\end{equation}
		If $R_2 > R_1$, then for $P \le  \frac{2^{2(R_2 - R_1)} - 1}{\beta^2}$,
		\begin{equation}
			\label{eq:C_diam_low_SNR_2}
			C_{\diamond}(R_1,R_2) = R_1+\psi(\beta^2 P).
		\end{equation}
	\end{enumerate}
	Furthermore, if $R_1 = R_2 = \bar R$, then as $P \to 0$, 
	\begin{equation}
		\label{eq:C_diam_asymp_low_SNR}
		C_{\diamond}(\bar R, \bar R) \to \bar R.
	\end{equation}
\end{theorem}

We now provide intuition on why
the MQ-DF scheme achieves capacity on the diamond channel with perfectly correlated noises in these two distinct regimes. In the high-SNR regime, the optimal rate-allocation parameter in \eqref{eq:R_MQ_DF} is $\rho^* = 0$.  In this regime, MQ is not used, and the hybrid scheme reduces to DF. This is capacity achieving, because when the transmit power $P$ is sufficiently large each relay can decode the transmitted message at a rate equal to its outgoing bit-pipe capacity. This saturates the multiple-access cut-set bound. 

In the more interesting low-SNR regime, the optimum $\rho$ is $\rho^* = R_{\textnormal{min}}$.  This means that the MQ component operates at its largest permissible rate. The DF component becomes degenerate. It now operates on the two-hop relay channel involving only the relay with higher forward link capacity with outgoing link capacity $R_{\text{ex}} \triangleq |R_1 - R_2|$. In this low-SNR regime, the bottleneck of this two-hop relay channel is the first hop, i.e., the single-user Gaussian channel between the transmitter and the relay with excess forward link capacity. The overall achievable rate is therefore the sum of $R_{\rm min}$ and this Gaussian channel capacity, which coincides with one of the cross cut-set bounds. Notably, when the two relay links have the same capacity, this scheme reduces to the pure MQ strategy and achieves the channel capacity in the low-SNR limit.


In Fig.~\ref{fig:bounds_diam}(a) we plot the new achievability bounds \eqref{eq:naive_MQ_diam} and \eqref{eq:R_MQ_DF} together with other known capacity bounds discussed in Section \ref{sec:known_bd_diam}. These bounds are plotted as functions of SNR, and the model parameters are given by $\alpha = -\beta = 1$, $R_1 = 2$, $R_2 = 1$. It is observed that MQ–DF outperforms all other schemes and matches the cut-set bound, and hence is \emph{capacity-achieving} in both high-SNR and low-SNR regimes.

\subsection{Anti-Symmetric Gaussian Primitive Diamond Channel}
\label{sec:bounds_sym}

To better understand how the achievability bound in Theorem \ref{thm:diam_MQ_DF} behaves in the moderate-SNR regime (not covered by the capacity results in Theorem \ref{thm:capacity_diamond}), we investigate a scenario where the relay gains have equal magnitudes but opposite signs, and relay-link capacities are the same, then compare the rate achieved by MQ-DF with other known capacity bounds for this special case.

Formally, the \textit{anti-symmetric Gaussian primitive diamond channel} is a special instance of the Gaussian primitive diamond channel where $R_1 = R_2 = \bar R$ and $\alpha = -\beta = 1$. It is straightforward that the MQ achievable rate \eqref{eq:naive_MQ_diam} is given by 
\begin{equation}
	\label{eq:naive_MQ_sym_diam}
	R_{\textnormal{sym},\diamond,\textnormal{MQ}}(\bar R)= R_{\diamond,\textnormal{MQ}}(\bar R, \bar R) = \bar R.
\end{equation}
It turns out that several other capacity bounds can also be computed explicitly in this anti-symmetric setting. The following proposition shows that the cut-set upper bound, the DF lower bound, and the MQ–DF lower bound all take simpler forms. We present a proof of this proposition in Appendix \ref{app:sym_diam_simplify_bds}.

\begin{figure*}[t]
	\centering
	\begin{subfigure}[t]{0.4\linewidth}
		\centering
		\includegraphics[width=\linewidth]{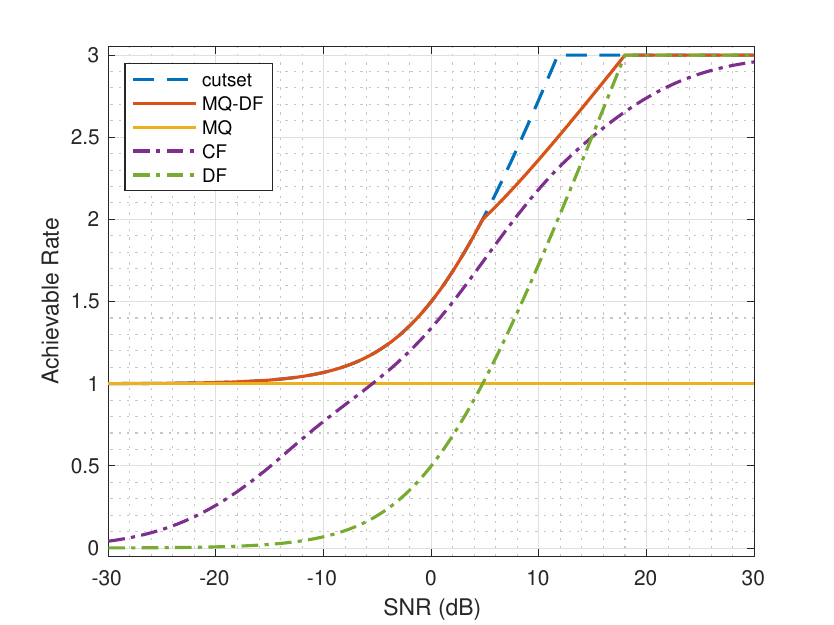}
		\caption{$\lambda=1$, $\alpha = -\beta = 1$, $R_1 = 2$, $R_2 = 1$}
	\end{subfigure}        
	\hspace{2mm}
	\begin{subfigure}[t]{0.4\linewidth}
		\centering
		\includegraphics[width=\linewidth]{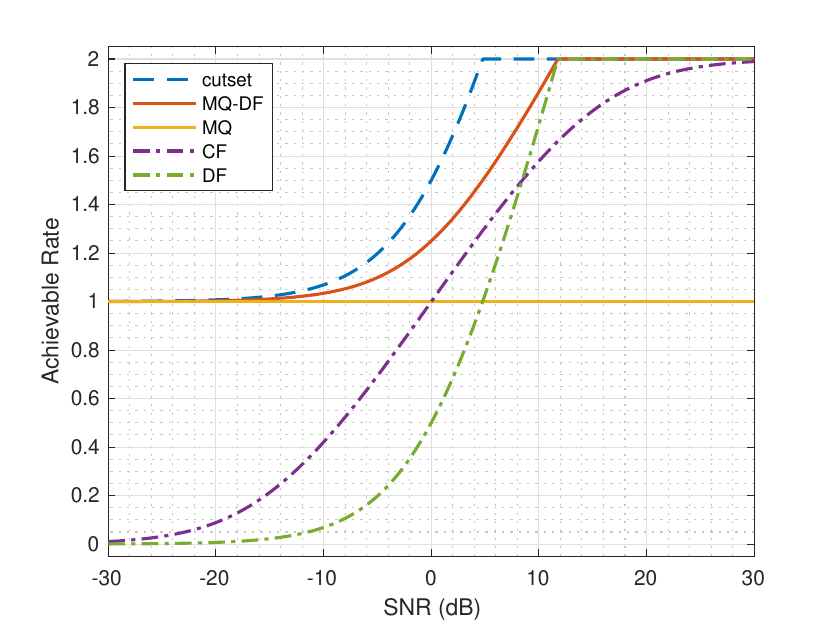}
		\caption{$\lambda=1$, $\alpha = -\beta = 1$, $R_1 = R_2 = 1$}
	\end{subfigure}        
	\hspace{1mm}
	
	\caption{Comparison of capacity bounds for the Gaussian primitive diamond channel.}
	\label{fig:bounds_diam}
\end{figure*} 

\begin{proposition}
	\label{prop:sym_diam_simplify_bds}
	For the anti-symmetric Gaussian primitive diamond channel, the cut-set upper bound \eqref{eq:diamond_cutset}, the DF lower bound \eqref{eq:diamond_DF}, and the MQ-DF lower bound \eqref{eq:R_MQ_DF} are given as below:
	\begin{align}
		\textnormal{Cutset}_{\textnormal{sym},\diamond}(\bar R) &=  \begin{cases}
			\bar R + \psi(P), &\psi(P) < \bar R,\\
			2\bar R, &\psi(P) \ge \bar R.\\
		\end{cases}\label{eq:sym_diam_cutset_2}\\
		R_{\textnormal{sym},\diamond,\textnormal{DF}}(\bar R) &= \begin{cases}
			\psi(P), &\psi(P) < 2\bar R,\\
			2\bar R, &\psi(P) \ge 2\bar R.\\
		\end{cases}\label{eq:sym_diam_DF_2}\\
		R_{\textnormal{sym},\diamond,\textnormal{MQ-DF}}(\bar R) 
		&= \begin{cases}
			\bar R + \frac{\psi(P)}{2}, &\psi(P) < 2\bar R,\\
			2\bar R, &\psi(P) \ge  2\bar R.
		\end{cases}\label{eq:sym_diam_MQ_DF_2}
	\end{align}
\end{proposition}

By comparing the bounds above, we observe that the cut-set bound is attainable by both the DF scheme and the MQ-DF scheme whenever $\psi(P) \ge 2\bar R$. This behaviour is consistent with the high SNR result \eqref{eq:C_diam_high_SNR} in Theorem \ref{thm:capacity_diamond}. In contrast, when $\psi(P) < 2\bar R$, the MQ-DF lower bound strictly improves upon the DF lower bound, demonstrating the advantage of incorporating MQ in this regime. In particular, the optimal rate-allocation parameter $\rho$ for the MQ-DF lower bound \eqref{eq:R_MQ_DF} is 
\begin{equation}
	\label{eq:rho_star}
	\rho^* = \max\left\{0,\bar R - \frac{\psi(P)}{2}\right\}.
\end{equation}
For fixed $\bar R$, this optimum $\rho^*$ converges to $\bar R$ when $P \to 0$, implying that the MQ-DF scheme reduces to the naive MQ scheme in the low power limit. This is again consistent with the discussion in the previous subsections.

For the CF lower bound \eqref{eq:diamond_CF}, due to the symmetry at the two relays, it is optimal to choose identical quantization noise levels, i.e., setting $N_1 = N_2 = N$. As a result, the CF rate for the anti-symmetric Gaussian primitive diamond channel reduces to
\begin{align}
	\label{eq:sym_cf_1}
	R_{\textnormal{sym},\diamond,\textnormal{CF}}(\bar R) = \max_{N > 0} \min \left\{\begin{matrix}
		\frac12 \log \left(1+\frac{2P}{N}\right)\\
		\bar R + \frac12 \log \left(1 + \frac{P-1}{N+2}\right)\\
		2\bar R + \frac12 \log \left(\frac{N}{N+2}\right)
	\end{matrix}\right\}.
\end{align}
The optimization program in \eqref{eq:sym_cf_1} does not admit a simple closed-form solution. Nevertheless, in the following proposition we compare the CF bound to the MQ and the hybrid MQ-DF bounds. The proof is deferred to Appendix \ref{app:sym_cf}.

\begin{proposition}
	\label{prop:sym_cf}
	The CF lower bound $R_{\textnormal{sym},\diamond,\textnormal{CF}}(\bar R)$ for the anti-symmetric Gaussian primitive diamond channel satisfies the following:
	
	(a) For all $P > 0$, $R_{\textnormal{sym},\diamond,\textnormal{CF}}(\bar R) < R_{\textnormal{sym},\diamond,\textnormal{MQ-DF}}(\bar R)$;
	
	(b)
	\begin{equation}
		R_{\textnormal{sym},\diamond,\textnormal{CF}}(\bar R) \begin{cases}
			< R_{\textnormal{sym},\diamond,\textnormal{MQ}}(\bar R) = \bar R, & P < 1,\\
			= R_{\textnormal{sym},\diamond,\textnormal{MQ}}(\bar R)= \bar R, & P = 1,\\
			> R_{\textnormal{sym},\diamond,\textnormal{MQ}}(\bar R)=\bar R, & P > 1.\\
		\end{cases}
	\end{equation}
\end{proposition}

The above proposition reveals the suboptimality of CF for this channel. In particular, part (a) shows that the MQ-DF coding scheme strictly outperforms CF scheme for all values of $P$, while part (b) further asserts that when the SNR is smaller than $1$, even MQ alone achieves a higher rate than CF. Taken together, these results emphasize that MQ-based schemes, despite their simplicity, are inherently more suited to the structure of the anti-symmetric Gaussian primitive diamond channel with perfectly correlated noises, while CF with Gaussian codebooks and independent Gaussian quantization fails to exploit this structure.

Inspired by the achievability result in Theorem \ref{thm:diam_MQ_DF}, one might wonder whether time-sharing between MQ and CF could yield a better achievable rate than either MQ or CF alone. Imitating the derivation in the proof of Theorem \ref{thm:diam_MQ_DF}, we derive an achievable rate for the MQ-CF coding scheme. The rate simplifies to the following for the anti-symmetric setting: 
\begin{align}
	\label{eq:sym_mq_cf}
	R_{\textnormal{sym},\diamond,\textnormal{MQ-CF}}(\bar R)  = \max_{\substack{N > 0\\\rho \in [0,\bar R]}} \!\!\min \left\{\begin{matrix}
		\frac12 \log \left(1+\frac{2P}{N}\right) + \rho\\
		\bar R + \frac12 \log \left(1 + \frac{P-1}{N+2}\right)\\
		2\bar R + \frac12 \log \left(\frac{N}{N+2} \right)- \rho
	\end{matrix}\right\}.
\end{align}
The following surprising result asserts that rate achievable by the MQ-CF scheme \eqref{eq:sym_mq_cf} is in fact equal to the higher of the CF rate \eqref{eq:sym_cf_1} and the naive MQ rate $\bar R$. Combining with Propositions \ref{prop:sym_diam_simplify_bds} and \ref{prop:sym_cf}, we conclude that the MQ-CF lower bound is dominated by the MQ-DF lower bound at all power levels. The proof of Proposition \ref{prop:sym_mq_cf} is given in Appendix \ref{app:sym_mq_cf}.
\begin{proposition}
	\label{prop:sym_mq_cf}
	The MQ-CF lower bound $R_{\textnormal{sym},\diamond,\textnormal{MQ-CF}}(\bar R)$ for the anti-symmetric Gaussian primitive diamond channel is given by
	\begin{align}
		R_{\textnormal{sym},\diamond,\textnormal{MQ-CF}}(\bar R) &= \max\{R_{\textnormal{sym},\diamond,\textnormal{CF}}(\bar R), \bar R\}  \\ &= \begin{cases}
			\bar R,~&P \le 1,\\
			R_{\textnormal{sym},\diamond,\textnormal{CF}}(\bar R),~&P> 1.
		\end{cases}
		\label{eq:prop_MQ_CF}
	\end{align}
\end{proposition}

In Fig.~\ref{fig:bounds_diam}(b) we compare all known capacity bounds (except the MQ-CF bound in Proposition \ref{prop:sym_mq_cf}) for the anti-symmetric Gaussian primitive diamond channel. The bounds are plotted across the entire SNR range. Again, it can be seen that the hybrid MQ-DF scheme achieves the highest rate among all schemes, in accordance to the discussion above; it is \emph{capacity-achieving} in the low SNR limit.

\begin{remark}
	In \cite{sanderovich2008communication}, the Gaussian primitive diamond
	channel with zero noise correlation and a \textit{nomadic} transmitter is
	studied. Nomadic means that the relays do not have the knowledge of the 
	codebook used by the transmitter.  In this setting, \cite{sanderovich2008communication} 
	shows that for the CF strategy, if the codebook at the transmitter is restricted to 
	be Gaussian, then it is optimal to use Gaussian quantization at the relays. 
	In addition, the authors state that it is not always optimal to use CF with a Gaussian
	codebook for this channel. They provide a counterexample where CF with binary phase-shift keying (BPSK) outperforms CF with Gaussian codebook and Gaussian quantization.
	However, 
	there appears to be an error in the numerical calculation in \cite{sanderovich2008communication} (see Appendix \ref{app:bpsk_Gaussian_CF}).
	The counterexample no longer stands after the error is fixed. 
	
	Nonetheless, when the noises at the relays are perfectly correlated, 
	Proposition \ref{prop:sym_cf}(b) provides a corrected example demonstrating that 
	the use of CF with Gaussian codebook and Gaussian quantization is suboptimal. 
	Consider the naive MQ
	coding scheme for the anti-symmetric Gaussian primitive diamond channel with
	relay-link capacities $\bar R = \log M$ for some positive integer $M$. 
	When $P < 1$, the MQ approach achieves a strictly higher rate than CF with Gaussian codebooks and Gaussian quantization.  
\end{remark}


\section{Non-Perfect Correlation}

In this section, we extend the MQ coding framework to the Gaussian primitive relay and diamond channels with non-perfectly correlated noises. The primary focus is the regime of strong positive correlation, i.e., $\lambda \approx 1$, which can be viewed as a perturbation of the perfectly correlated case. In this regime, the high noise correlation allows suitably adapted scalar MQ schemes (i.e., without superposition or time-sharing, cf. Section \ref{sec:superposed_MQ} and Section \ref{sec:MQ_DF}) to retain some of their advantages. We show that such MQ schemes can still yield achievability results comparable to classical coding strategies for both channels. 

\subsection{Gaussian Primitive Relay Channel}
\label{sec:prim_relay_nonperfect}

The Gaussian primitive relay channel with perfect noise correlation has been introduced in Section \ref{sec:prim_relay_model}. In this subsection, the noises at the receiver and the relay are not perfectly correlated, i.e., the assumption in \eqref{eq:noise_distr} is replaced by 
\begin{equation}
	\begin{bmatrix}
		Z \\ Z_0
	\end{bmatrix} \sim \mc N\left( \mf 0, \begin{bmatrix}
		1 & \lambda \\ \lambda & 1
	\end{bmatrix}\right), ~\text{where}~\lambda \in (-1,1).
\end{equation}
We refer to such noises as \textit{$\lambda$-correlated}. By rewriting $Z_0 = \lambda Z + Z'$, where $Z' \sim \mc N(0, 1-\lambda^2)$ and $Z' \indep Z$, the relay observation can thus be expressed as 
\begin{equation}
	Y_0 = \alpha X + Z_0 = \alpha X + \lambda Z + Z'.
\end{equation}
Note that when $\lambda = \alpha$, $Y_0 = \alpha (X+Z) + Z' = \alpha Y + Z'$,
making $X - Y - Y_0$ a Markov chain. In this case, the relay’s observation is a degraded version of the receiver’s, providing no additional information about 
$X$. Relaying is therefore futile, analogous to the $\alpha = 1$ case under perfect noise correlation (see Section \ref{sec:prim_relay_model}). In the following, we assume $\alpha \neq \lambda$.

The capacity of the Gaussian primitive relay channel with $\lambda$-correlated noises is not yet known (except for the trivial case where $\alpha = 
\lambda$). Below, we state several known achievability and converse bounds. The proof is deferred to Appendix \ref{app:prim_nonperf_bd}.
\begin{proposition}
	\label{prop:prim_nonperf_bd}
	Assume $\alpha \neq \lambda$. The capacity of the Gaussian primitive relay channel with $\lambda$-correlated noise is bounded from above by the cut-set bound
	\begin{align}
		&\textnormal{Cutset}_{\triangleleft}(R_0;\lambda)   \notag \\
		&= \min\left\{ \psi\left( \frac{(\alpha - \lambda)^2 P}{1-\lambda^2} + P\right), \psi(P) + R_0 \right\}.
		\label{eq:nonperf_cutset}
	\end{align}
	It is bounded from below by the CF lower bound
	\begin{align}
		&R_{\triangleleft,\textnormal{CF}}(R_0;\lambda)  \notag  \\ 
		&= \psi(P) + \psi\left( \frac{(2^{2R_0} - 1)(\alpha - \lambda)^2 P}{ 2^{2R_0} (1-\lambda^2)(P+1) + (\alpha - \lambda)^2 P}\right),
		\label{eq:nonperf_cf_lb}
	\end{align}
	and by the DF lower bound
	\begin{align}
		&R_{\triangleleft,\textnormal{DF}}(R_0;\lambda)  \notag \\
		&= \begin{cases}
			\psi(P),&\!\!\alpha^2 \le 1,\\
			\psi(\alpha^2 P), &\!\!1 < \alpha^2 \le \frac{2^{2R_0}-1}{P}, \\[1mm]
			\psi\!\left(\! P \!+\! \frac{(\alpha^2-1)(2^{2R_0} - 1)}{\alpha^2}\! \right)\!, &\!\! \alpha^2 \!>\! \max \!\left\{1, \frac{2^{2R_0}-1}{P}\right\}\!.
		\end{cases} \label{eq:nonperf_df_lb}
	\end{align}

\end{proposition}

In the following, we present a scalar MQ coding scheme for an example of this channel where $R_0 = 1$. The proposed scheme transmits one of two possible messages in a single channel use, albeit with a nonzero probability of error due to imperfect correlation between the noises. This nonzero error probability, however, does not fundamentally limit the practicality of the scheme. In effect, the one-shot MQ construction induces a binary symmetric channel between the transmitted and decoded messages, and by applying standard block coding across multiple uses of this induced channel, the overall decoding error probability can be driven arbitrarily low. Moreover, when the correlation coefficient tends to $1$, the crossover probability of the induced binary symmetric channel diminishes.

Fix power $P > 0$, let $\Delta = 2|\frac{\alpha}{\lambda} - 1|\sqrt{P}$. Note that $\Delta >0$ since $\alpha\neq \lambda$. To encode a message $M \in \{0,1\}$, the transmit signal is set to be
\begin{equation}
	\label{eq:nonperf_mq_enc}
	X  = \frac{(M-\frac12)\Delta}{\frac{\alpha}{\lambda} - 1}= \sqrt{P}\cdot\sgn\left( \frac{(M-\frac12)\Delta}{\frac{\alpha}{\lambda} - 1}\right).
\end{equation}
Observe that the power constraint is satisfied, because for either $M \in \{0,1\}$, $|X^2| = P$. 

The relay observes $Y_0 = \alpha X + Z_0 = \alpha X + \lambda Z + Z'$. It transmits 
\begin{equation}
	\label{eq:nonperf_mq_relay}
	V = \mq{\frac{Y_0}{\lambda}}{\Delta,2} = \mq{\frac{\alpha}{\lambda} X + Z + \frac{1}{\lambda} Z'}{\Delta,2}
\end{equation} 
Given the encoding rule \eqref{eq:nonperf_mq_enc} and that the receiver observation is $Y = X + Z$, we can rewrite \eqref{eq:nonperf_mq_relay} as
\begin{align}
	V &= \mq{\left(\frac{\alpha}{\lambda}-1\right) X + Y + \frac{1}{\lambda} Z'}{\Delta,2}  \notag \\
	&= \mq{(M-\frac12)\Delta + Y + \frac{1}{\lambda} Z'}{\Delta,2} \notag \\
	&= \left(\mq{ \tilde Y + \frac{1}{\lambda} Z'}{\Delta,2} + M \right) \mod 2, \label{eq:nonperf_mq_relay_2}
\end{align} 
where $\tilde Y = Y - \frac12 \Delta$, and the last equality is due to \eqref{eq:colour_lemma_2}.

To decode the message, the receiver first computes $\tilde V = \tmq{\tilde Y}{\Delta,2} = \mq{Y - \frac12\Delta}{\Delta,2}$. Upon receipt of the relay's transmission $V$, the receiver can form
\begin{align}
	\hat M &= (V - \tilde V) \mod 2 \notag \\
	&= \left(M+ \mq{ \tilde Y + \frac{1}{\lambda} Z'}{\Delta,2}-\mq{ \tilde Y }{\Delta,2}\right) \mod 2.
\end{align}
Observe that $\hat M \neq M$ if and only if $\tmq{\tilde Y + \frac{1}{\lambda} Z'}{\Delta,2} \neq \tmq{\tilde Y}{\Delta,2}$. 
Now for $m \in \{0,1\}$, we denote
\begin{equation}
	\breve{\bm Y}_m = \begin{bmatrix}
		\tilde Y \\ \tilde Y + \frac{1}{\lambda} Z'
	\end{bmatrix}
\end{equation}
conditioned on $M = m$.
We observe that  
\begin{equation}
	\breve{\bm Y}_m \sim \mc N\left(\bm \mu_m, \begin{bmatrix}
		1 & 1\\
		1 & 1/\lambda^2
	\end{bmatrix}\right),
\end{equation}
where 
\begin{equation}
	\bm\mu_0 = \frac{\alpha \Delta}{2(\lambda - \alpha)} \begin{bmatrix}
		1 \\ 1
	\end{bmatrix} ,~\bm\mu_1 =  \begin{bmatrix}
		-\Delta \\ -\Delta
	\end{bmatrix} - \bm\mu_0.
\end{equation}
Define $\mc R_{\Delta}$ to be the following two-dimensional region which is a countable union of squares of side length $\Delta$:
\begin{equation}
	\label{eq:mc_R_Delta}
	\mc R_{\Delta} = \bigcup_{i,j \in \mb Z} [i \Delta, (i+1)\Delta] \times [(i + 2j - 1 )\Delta, (i +2j)\Delta],
\end{equation}
where $\mc A \times \mc B$ denotes the Cartesian product of two sets $\mc A$ and $\mc B$.
It is straightforward from the definition of modulo quantization that $\mc R_{\Delta} = \{[x,y]^T: \mq{x}{\Delta,2}\neq \mq{y}{\Delta,2}\}$. We then have
\begin{align}
	\Pr\{\hat M = 1 | M = 0\} &= \Pr\{\breve{\bm Y}_0\in \mc R_{\Delta} \} \\
	\Pr\{\hat M = 0 | M = 1\} &=  \Pr\{\breve{\bm Y}_1 \in \mc R_{\Delta} \}.
\end{align}
We claim that 
\begin{equation}
	\label{eq:pe_0=pe_1}
	\Pr\{\breve{\bm Y}_0\in \mc R_{\Delta} \}  = \Pr\{\breve{\bm Y}_1\in \mc R_{\Delta} \}.
\end{equation}
This is because $\breve{\bm Y}_1$ is the point reflection of $\breve{\bm Y}_0$ about the point $[-\frac{\Delta}{2}, -\frac{\Delta}{2}]^T$, and that $\mc R_\Delta$ is centrally symmetric about the same point. Therefore, the overall error probability of this scheme is equal to 
\begin{align}
	P_{e,\triangleleft} &=  \frac{\Pr\{\hat M = 1 | M = 0\} + \Pr\{\hat M = 0| M =1\} }{2} \notag \\
	&=\Pr\{\breve{\bm Y}_0\in \mc R_{\Delta} \}.	\label{eq:pe=pe_0}
\end{align}

\begin{figure*}[t]
	\centering
	\begin{subfigure}{0.34\linewidth}
		\centering
		\includegraphics[width=\linewidth]{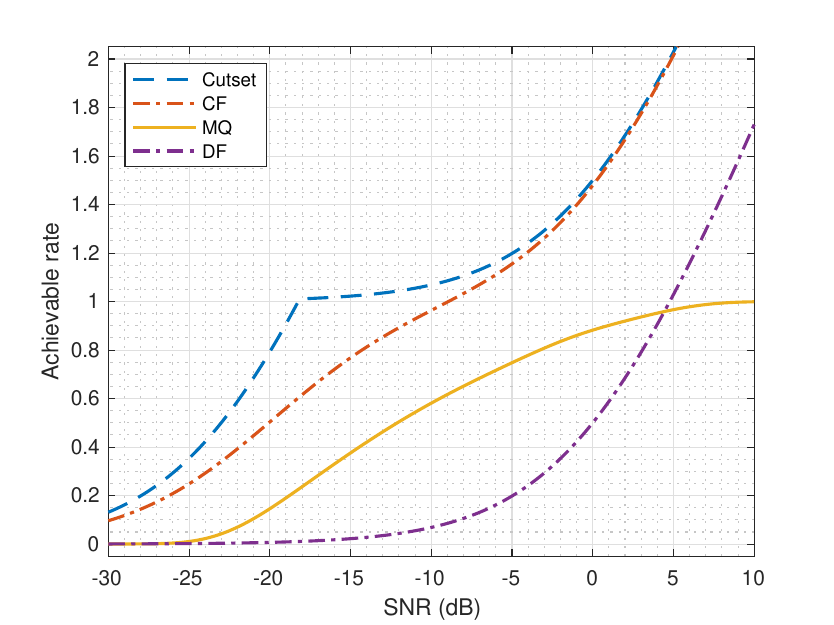}
		\caption{Fix $\alpha = -1$ and $\lambda = 0.99$, vary SNR.}
	\end{subfigure}
	\hspace{-5mm}
	\begin{subfigure}{0.34\linewidth}
		\centering
		\includegraphics[width=\linewidth]{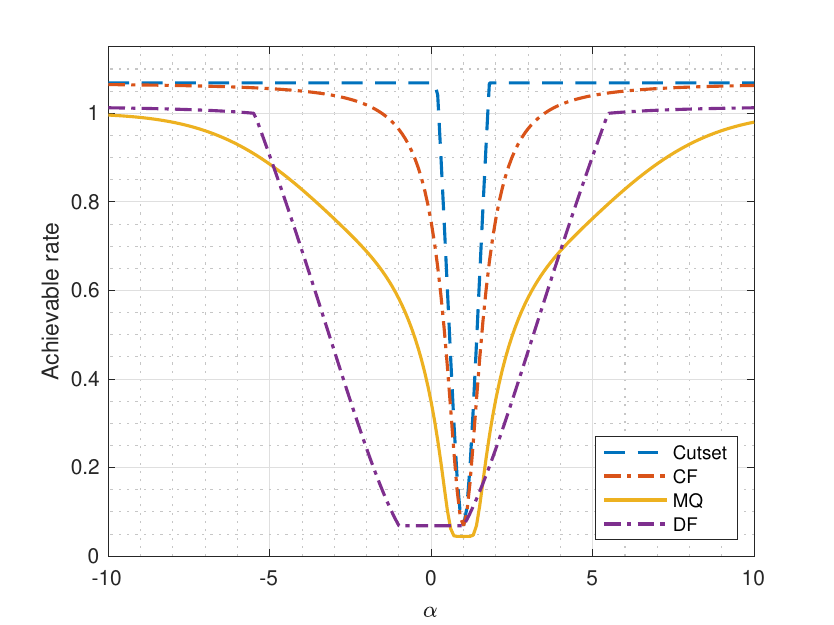}
		\caption{Fix SNR $= -10$dB and $\lambda = 0.99$, vary $\alpha$.}
	\end{subfigure}
	\hspace{-5mm}
	\begin{subfigure}{0.34\linewidth}
		\centering
		\includegraphics[width=\linewidth]{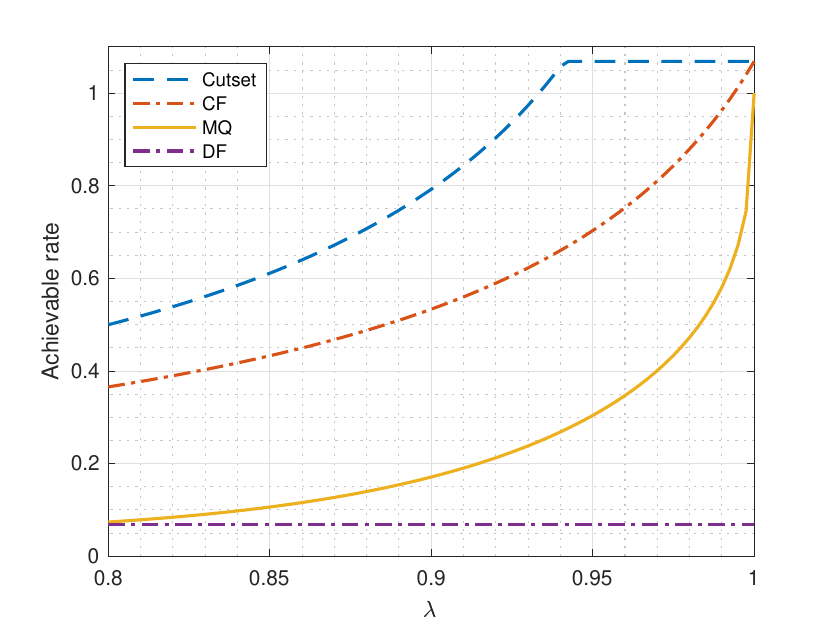}
		\caption{Fix $\alpha = -1$ and SNR $= -10$dB, vary $\lambda$.}
	\end{subfigure}
	
	\caption{Comparison of bounds on the capacity $C_{\triangleleft}(R_0;\lambda)$ of the Gaussian primitive relay channel with $\lambda$-correlated noises at the relays, $R_0 = 1$. Note that in (b) there is a dip around $\alpha = \lambda$ in the curves corresponding to the cut-set upper bound, the MQ lower bound, and the CF lower bound. This is because relaying does not help when $\alpha = \lambda$. The DF curve is symmetric about $\alpha = 0$, as the DF scheme does not take noise correlation into account, and its achievable rate is non-decreasing in $|\alpha|$.}
	\label{fig:nonperf_prim_bounds}
\end{figure*} 

As $|\lambda| \to 1$, the density of $\breve {\bm Y}_0$ becomes increasingly concentrated around the line with slope $1$ passing through the origin, which is (almost everywhere) disjoint with $\mc R_\Delta$, therefore $P_{e,\triangleleft} \to 0$, corroborating the result for the perfect noise correlation in the previous subsection.

As discussed above, the scalar MQ construction induces an effective binary channel from the transmitted message $M$ to the decoded message $\hat M$. Furthermore, from \eqref{eq:pe_0=pe_1} and \eqref{eq:pe=pe_0} we see that the induced binary channel is symmetric with crossover probability $P_{e,\triangleleft}$. By applying standard channel coding over this induced channel, we obtain the following achievable rate:
\begin{equation}
	\label{eq:MQ_nonperf}
	R_{\triangleleft, \textnormal{MQ}}(1;\lambda) = I(M;\hat M) =  1-h_b(P_{e,\triangleleft}),
\end{equation}
where $h_b(\cdot)$ is the binary entropy function.

The achievable rate in \eqref{eq:MQ_nonperf} can be further improved. In the construction above, only the modulo quantization of the receiver observation $Y$ is used for decoding, while other information, e.g., its sign, is ignored. To exploit the sign information, we introduce a binary quantization of $Y$ obtained by thresholding at $0$:
\begin{equation}W = \mb I\{Y \ge 0\} = \begin{cases}
		1, ~&Y \ge 0,\\
		0,~&Y < 0.
	\end{cases}
\end{equation}
This induces a binary-input quaternary-output channel from $M$ to $(\hat M, W)$, which yields the following improved MQ coding rate:
\begin{equation}
	\label{eq:prim_nonperf_MQ_rate}
	R_{\triangleleft,\textnormal{MQ-BQ}}(R_0;\lambda) = I(M; \hat M, W).
\end{equation}

We compare the new achievability result in \eqref{eq:prim_nonperf_MQ_rate} with the known bounds on the capacity of the Gaussian primitive relay channel presented in Proposition~\ref{prop:prim_nonperf_bd}, with particular emphasis on the high noise-correlation regime. In the three subfigures of Fig.~\ref{fig:nonperf_prim_bounds}, we plot the achievable rates and the cut-set upper bound as functions of the SNR~$P$, the relay gain~$\alpha$, and the noise-correlation coefficient~$\lambda$, respectively, while keeping the other two parameters fixed. These plots collectively illustrate how the performance of MQ coding evolves across different operating conditions. 

We observe that when the correlation~$\lambda$ is sufficiently high---that is, when the relay and receiver noises are strongly aligned---MQ coding achieves a substantially higher rate than DF over a wide range of~$\alpha$ and~$P$. Although its achievable rate remains lower than that of CF, MQ coding offers an attractive tradeoff between performance and implementation simplicity. Notably, as $\lambda \to 1$, the rate gap between MQ and CF diminishes, indicating that MQ effectively captures the essential structure of the optimal strategy in the near-perfect correlation regime. Taken together, these results demonstrate that MQ coding delivers near-optimal performance in the very high-correlation scenarios while maintaining low encoding and decoding complexity.

\subsection{Gaussian Primitive Diamond Channel}
\label{sec:diam_relay_nonperfect}

\begin{figure*}[t]
	\centering
	\begin{subfigure}[t]{0.4\linewidth}
		\centering
		\includegraphics[width=\linewidth]{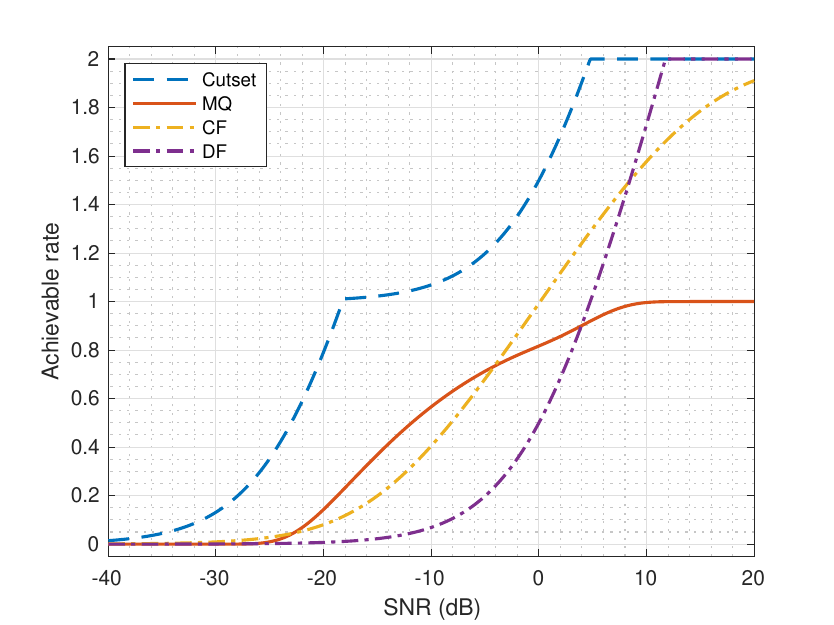}
		\caption{Fix $\lambda = 0.99$, vary SNR.}
	\end{subfigure}        
	\begin{subfigure}[t]{0.4\linewidth}
		\centering
		\includegraphics[width=\linewidth]{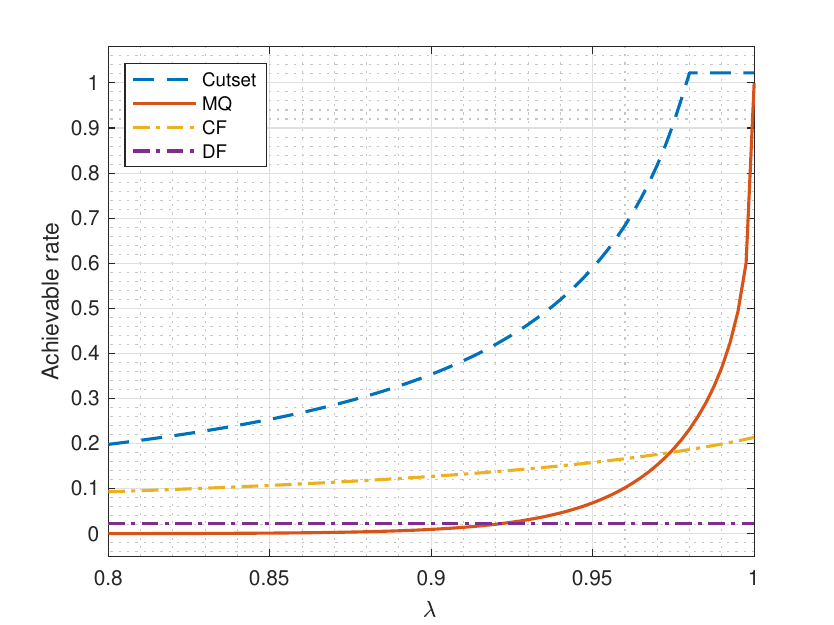}
		\caption{Fix SNR $= -15$dB, vary $\lambda$.}
	\end{subfigure}
	\caption{Comparison of capacity bounds of the anti-symmetric Gaussian primitive diamond channel with $\lambda$-correlated noises at the relays, $\bar R = 1$.}
	\label{fig:nonperf_sym_diam_bounds}
	
\end{figure*}

The Gaussian primitive diamond channel with perfect noise correlation has been introduced in Section \ref{sec:diamond_model}. In this subsection, we consider the scenario where the noises at the relays are not perfectly correlated, i.e., the assumption in \eqref{eq:diam_noise_distr} is replaced by
\begin{equation}
	\begin{bmatrix}
		Z_1 \\ Z_2
	\end{bmatrix} \sim \mc N\left( \mf 0, \begin{bmatrix}
		1 & \lambda \\ \lambda & 1
	\end{bmatrix}\right), ~\text{where}~\lambda \in (-1,1).
\end{equation}
We first present several known capacity bounds for this channel; the proof is provided in Appendix \ref{app:diam_nonperf_bd}.
\begin{proposition}
	\label{prop:diam_nonperf_bd}
	The capacity of the Gaussian primitive diamond channel with $\lambda$-correlated noise is bounded from above by the cut-set bound
	\begin{equation}
		\textnormal{Cutset}_\diamond(R_1,R_2;\lambda) = \min\left\{\begin{matrix}
			R_1 + R_2 \\ R_1 + \psi(\beta^2 P)\\ R_2 + \psi(\alpha^2 P) \\
			\psi \left(\frac{(\alpha^2 + \beta^2 - 2\lambda \alpha \beta)P}{1-\lambda^2} \right)
		\end{matrix}\right\}.
		\label{eq:diam_nonperf_cutset}
	\end{equation}
	It is bounded from below by the CF lower bound
	\begin{align}
		&R_{\diamond,\textnormal{CF}}(R_1,R_2;\lambda) \notag \\&= \max_{N_1,N_2 > 0} \!\!\min \left\{\begin{matrix}
			\frac12\log \frac{(\alpha^2 P + 1 + N_1)(\beta^2 P + 1 + N_2) - (\alpha \beta P + \lambda)^2}{(1+N_1)(1+N_2)-\lambda^2}\\ R_1 + \frac{1}{2}\log \frac{N_1(\beta^2 P + 1 + N_2)}{(1+N_1)(1+N_2)-\lambda^2} \\
			R_2 + \frac{1}{2}\log \frac{N_2(\alpha^2 P + 1 + N_1)}{(1+N_1)(1+N_2)-\lambda^2}\\ 
			R_1 + R_2 + \frac{1}{2}\log \frac{N_1 N_2}{(1+N_1)(1+N_2)-\lambda^2}
		\end{matrix}\right\},
		\label{eq:diam_nonperf_cf_lb}
	\end{align}
	where $N_i$ is the variance of the independent Gaussian quantization noise used by relay $i$ for $i = 1,2$. It is also bounded from below by the DF lower bound \eqref{eq:diamond_DF}.
\end{proposition}

We now construct a scalar MQ code for this channel. For simplicity, consider the anti-symmetric setting as an example (i.e., $\alpha = -\beta = 1$ and $R_1 = R_2 = \bar R$) and assume $\bar R = 1$. Similar to the primitive relay channel, the resulting one-shot MQ scheme induces a binary symmetric channel whose crossover probability diminishes as the noise correlation approaches $1$.

Fix power $P > 0$ and let $\Delta = 4\sqrt{P}$. To encode a message $M \in \{0,1\}$, the transmit signal is set to be $X = f_{\sf enc}(M) = (2M-1)\Delta/4$. It is easy to see that $X$ takes value from $\{\pm \sqrt{P}\}$, hence the power constraint is satisfied. The relays forward the $(\Delta,2)$-modulo quantization of a shifted version of their observation:
\begin{equation}
	V_1 = \mq{Y_1 + \Delta}{\Delta,2},\quad V_2 = \mq{Y_2 + \Delta/2}{\Delta,2}.
\end{equation}
The receiver declares $\hat M = (V_1 - V_2) \mod 2$. An application of Lemma \ref{lem:mq} and more algebraic manipulation yields 
\begin{equation}
	\hat M = \left(M + \mq{\tilde Z_1 - \frac{M}{2}\Delta}{\Delta,2} - \mq{\tilde Z_2 - \frac{M}{2}\Delta}{\Delta,2}\right) \mod 2,
\end{equation}
where $\tilde Z_i = Z_i + \frac{3}{4}\Delta$ for $i \in \{1,2\}$. Again, note that $\hat M \neq M$ if and only if $\mq{\tilde Z_1 - \frac{M}{2}\Delta}{\Delta,2} \neq \mq{\tilde Z_2 - \frac{M}{2}\Delta}{\Delta,2}$. Let $\tilde {\bm{Z}} = [\tilde Z_1, \tilde Z_2]^T$, we have
\begin{align}
	\Pr\{\hat M = 1| M = 0\} &=  \Pr\{\tilde{\bm Z}\in \mc R_{\Delta}\}, \\ 
	\Pr\{\hat M = 0| M = 1\} &= \Pr\left\{\tilde{\bm Z} -  {\textstyle\left[\frac{\Delta}{2}, \frac{\Delta}{2}\right]^T} \in \mc R_{\Delta}\right\},
\end{align}
where $\mc R_\Delta$ has been defined in \eqref{eq:mc_R_Delta}. By symmetry of $\tilde{\bm Z}$ and $\mc R_{\Delta}$, we conclude that 
\begin{equation}
	\Pr\{\tilde{\bm Z}\in \mc R_{\Delta}\} = \Pr\left\{\tilde{\bm Z} -  {\textstyle\left[\frac{\Delta}{2}, \frac{\Delta}{2}\right]^T} \in \mc R_{\Delta}\right\},
	\label{eq:pe_0=pe_1_diamond}
\end{equation}
and hence the overall error probability of this scheme is equal to 
\begin{align}
	P_{e,\diamond} &=  \frac{\Pr\{\hat M = 1 | M = 0\} + \Pr\{\hat M = 0| M =1\} }{2} \notag \\
	&=\Pr\{\tilde{\bm Z}\in \mc R_{\Delta} \}.	\label{eq:pe_diamond}
\end{align}

From \eqref{eq:pe_0=pe_1_diamond} and \eqref{eq:pe_diamond}, we see that the channel from $M$ to $\hat M$ is a binary symmetric channel with crossover probability $P_{e,\diamond}$. We can apply standard channel coding over this induced channel, which yields the following achievable rate:
\begin{equation}
	\label{eq:MQ_diam_nonperf}
	R_{\diamond, \textnormal{MQ}}(1;\lambda) = I(M;\hat M) =  1-h_b(P_{e,\diamond}),
\end{equation}
where $h_b(\cdot)$ is the binary entropy function.

The MQ achievable rate \eqref{eq:MQ_diam_nonperf} and the known capacity bounds are compared in Fig.~\ref{fig:nonperf_sym_diam_bounds}. We plot the bounds versus SNR for $\lambda = 0.99$ in Fig.~\ref{fig:nonperf_sym_diam_bounds}(a). It is observed that MQ scheme outperforms both CF and DF over most of the low-SNR regime ($\text{SNR} \le -5$dB). In Fig.~\ref{fig:nonperf_sym_diam_bounds}(b), the bounds are plotted versus $\lambda$ at $\text{SNR} = -15$dB. In this case, the achievable rate of MQ increases sharply with $\lambda$ when $\lambda \approx 1$, surpassing both CF and DF for $\lambda \ge 0.98$.

\section{Conclusion}

This paper proposes MQ coding as a simple, explicit, and structured coding technique for channels with primitive relay links and correlated noises. We demonstrate that for the Gaussian primitive relay channel with perfect noise correlation, MQ coding is able to transmit up to $\lfloor 2^{R_0} \rfloor$ messages without error in a single channel use, and when combined with a code for AWGN channel supported on a scaled integer lattice, it attains the full capacity $C_\triangleleft(R_0) = \psi(P) + R_0$. This provides an alternative achievability proof of the capacity without relying on high-complexity random binning or high-dimensional rate-distortion code. When the noise correlation is not perfect, MQ coding remains effective when the correlation is sufficiently high. In this regime, the achievable rate of MQ coding approaches that of CF, while maintaining much lower complexity and outperforming DF by a significant margin. 

For the Gaussian primitive diamond channel, we extend the modulo quantization idea to design a variant of the MQ coding scheme. When the noises at the relays are perfectly correlated, the MQ scheme enables zero-error transmission of up to $\lfloor 2^{R_{\textnormal{min}}} \rfloor$ messages in one channel use. Furthermore, a hybrid MQ–DF scheme that combines MQ with rate-splitting and time-sharing is developed. For the perfect noise correlation scenario, the proposed scheme strictly improves upon both CF and DF over a wide range of parameters and, in several regimes, achieves the cut-set upper bound, thereby establishing the exact capacity. For the case of non-perfectly correlated noises, when the correlation is sufficiently high, MQ coding outperforms both CF and DF in certain channel parameter ranges. This reinforces the conclusion that MQ is a practical and effective scheme even beyond the idealized perfect-correlation scenario, especially for diamond channels without a direct link from the transmitter to the receiver.

Overall, these results identify MQ as a coding primitive for exploiting noise correlation in relay networks. It provides an explicit, algebraically grounded mechanism that can be used in tandem with AWGN codes for the Gaussian primitive relay channel or DF for the Gaussian primitive diamond channel.

\appendices

\section{Proof of Proposition \ref{prop:bd_diamond}}
\label{app:bd_diamond}

The cut-set upper bound \eqref{eq:diamond_cutset} can be derived from the generic cut-set bound for the diamond network. Due to the average power constraint, we have
\begin{equation}
	\label{eq:diamond_cutset_1}
	C_{\diamond}(R_1,R_2) \le \max_{\mb E[X^2]\le P} \min\left\{ \begin{matrix}
		I(X;Y_1,Y_2)\\
		I(X;Y_1) + R_2\\
		I(X;Y_2) + R_1\\
		R_1 + R_2
	\end{matrix}\right\}.
\end{equation}
Observe that since $X$ is a deterministic function of $Y_1$ and $Y_2$, i.e., $X = \frac{1}{\alpha - \beta} (Y_1 - Y_2)$, the mutual information $I(X;Y_1,Y_2) = \infty$. Note that $I(X;Y_1)$ and $I(X;Y_2)$ are both maximized when $X \sim \mc N(0,P)$, which are equal to $\psi(\alpha^2 P)$ and $\psi(\beta^2 P)$ respectively, therefore we simplify \eqref{eq:diamond_cutset_1} and get
\begin{align}
	C_{\diamond}(R_1,R_2) &\le \min\left\{ \begin{matrix}
		\psi(\alpha^2 P) + R_2\\
		\psi(\beta^2 P) + R_1\\
		R_1 + R_2
	\end{matrix}\right\} \\
	&= \textnormal{Cutset}_\diamond(R_1, R_2),
\end{align}
thus obtaining \eqref{eq:diamond_cutset}.

The DF achievable rate \eqref{eq:diamond_DF} is obtained as follows. Since $|\alpha| \ge |\beta|$, the channel from $X$ to $Y_1$ and $Y_2$ is a degraded broadcast channel, with $Y_1$ being the stronger receiver and $Y_2$ being the weaker one. We let both relays decode a common message of rate $R_{\text{DF,comm}}$ and the stronger relay decode an additional private message of rate $R_{\text{DF,priv}}$. After decoding their respective messages, the two relays forward them to the destination. Specifically, the stronger receiver first tries to forward its own private message, and if the private message rate is below $R_1$, it uses the remaining rate to forward part of the common message. The weaker receiver now forwards the remaining part of the common message. Summarizing, the following rates are achievable for any $\gamma \in [0,1]$:
\begin{gather}
	R_{\text{DF,comm}} \le \psi(\frac{\bar{\gamma} \beta^2 P}{\gamma \beta^2 P + 1}),~~R_{\text{DF,priv}}\le \psi(\gamma \alpha^2 P),\label{eq:DF_rate_constr_1}\\
	R_{\text{DF,priv}}\le R_1,~~R_{\text{DF,comm}}+R_{\text{DF,priv}}\le R_1 + R_2.\label{eq:DF_rate_constr_2}
\end{gather}
The rate expressions in \eqref{eq:DF_rate_constr_1} come from the achievable rates of superposition coding for the Gaussian broadcast channel (see e.g., \cite[Section 5.5.1]{el2011network}), where $\gamma$ is the parameter for allocating the transmit power to the two superposition codebooks and $\bar{\gamma} = 1-\gamma$; and the rate expressions in \eqref{eq:DF_rate_constr_2} are to ensure that messages can be forwarded to the destination within the capacity constraints of the relay links. 

The destination receives both the private message and the common message, thus the overall rate of this scheme is $R = R_{\text{DF,comm}} + R_{\text{DF,priv}}$. Using Fourier-Motzkin elimination, any overall rate satisfying
\begin{equation}
	R \le \min \left\{\begin{matrix}
		R_1 + R_2\\ R_1 + \psi\left( \frac{\bar\gamma \beta^2 P }{\gamma \beta^2 P + 1}\right)\\\psi(\gamma \alpha^2 P )  + \psi\left( \frac{\bar\gamma \beta^2 P }{\gamma \beta^2 P + 1}\right)
	\end{matrix} \right\}
\end{equation}
is achievable. By further optimizing the parameter $\gamma \in [0,1]$, we obtain the achievability result \eqref{eq:diamond_DF}.

The CF achievable rate \eqref{eq:diamond_CF} is derived below. Accounting for the correlation between the relay observations, we adapt the simultaneous-decoding CF lower bound of \cite{sanderovich2008communication} as follows:
\begin{equation}
	\label{eq:diamond_CF_1}
	C_{\diamond}(R_1,R_2) \ge \min\left\{ 
	\begin{matrix}
		I(X;\hat Y_1, \hat Y_2)\\
		R_1 - I(Y_1;\hat Y_1)+ I_0\\
		R_2 - I(Y_2;\hat Y_2) + I_0\\
		R_1 + R_2  - I(Y_1;\hat Y_1)- I(Y_2;\hat Y_2) + I_0\\
	\end{matrix}\right\},
\end{equation}
for any joint pdf $p(x)p(y_1,y_2|x)p(\hat y_1|y_1)p(\hat y_2|y_2)$, and
\begin{equation}
	\label{eq:diamond_CF_2}
	I_0 = I(X;\hat Y_1,\hat Y_2) + I(\hat Y_1;\hat Y_2).
\end{equation}
In general, the optimal choice of $\hat Y_1$ and $\hat Y_2$ is not known. We use Gaussian codebook and Gaussian quantization, namely $X \sim \mc N(0,P)$, $\hat Y_1 = Y_1 + U_1$ and $\hat Y_2 = Y_2 + U_2$, where $U_1 \sim \mc N(0, N_1)$, $U_2 \sim \mc N(0, N_2)$ are such that $U_1 \indep U_2$ and $(U_1,U_2) \indep (X, Y_1, Y_2)$. Plugging the choice of $\hat Y_1$ and $\hat Y_2$ in \eqref{eq:diamond_CF_1} and optimize over $N_1$ and $N_2$, we obtain \eqref{eq:diamond_CF}. This concludes the proof.

\section{Proof of Theorem \ref{thm:capacity_diamond}}
\label{app:capacity_diamond}

The converse part of the theorem follows directly from the cut-set bound \eqref{eq:diamond_cutset}, therefore we only need to prove the achievability part. For fixed $P > 0$, define 
\begin{equation}
	\label{eq:R_alpha_beta}
	R_{\alpha}(\gamma) = \psi(\gamma \alpha^2 P), ~R_{\beta}(\gamma) = \psi\left( \frac{\bar\gamma \beta^2 P}{\gamma \beta^2 P + 1} \right),
\end{equation}
for $\gamma \in [0,1]$. We can now rewrite the achievability bound in \eqref{eq:R_MQ_DF} as
\begin{equation}
	\label{eq:R_MQ_DF_new}
	C_{\diamond}(R_1,R_2) \ge \max_{\substack{\gamma \in [0,1]\\0 \le \rho \le R_{\textnormal{min}}}} \min \left\{\begin{matrix}
		R_1 + R_2 - \rho\\ R_1 + R_{\beta}(\gamma)\\ R_{\alpha}(\gamma) + R_{\beta}(\gamma) + \rho
	\end{matrix} \right\}.
\end{equation}
In particular, since the right-hand side of \eqref{eq:R_MQ_DF_new} is a maximization over the parameters $\gamma$ and $\rho$, it is lower bounded by fixing these parameters to specific values, i.e., for all $\tilde \gamma \in [0,1]$ and $\tilde \rho \in [0, R_\textnormal{min}]$, we have
\begin{equation}
	\label{eq:R_MQ_DF_new_2}
	C_{\diamond}(R_1,R_2) \ge \min \left\{\begin{matrix}
		R_1 + R_2 - \tilde\rho\\ R_1 + R_{\beta}(\tilde\gamma)\\ R_{\alpha}(\tilde\gamma) + R_{\beta}(\tilde\gamma) + \tilde\rho
	\end{matrix} \right\}.
\end{equation}

For part (a), the lower bound \eqref{eq:R_MQ_DF_new_2} reduces to 
\begin{equation}
	\label{eq:R_tilde_gamma}
	C_{\diamond}(R_1,R_2) \ge \min \left\{\begin{matrix}
		R_1 + R_2 \\ R_1 + R_{\beta}(\tilde\gamma)\\ R_{\alpha}(\tilde\gamma) + R_{\beta}(\tilde\gamma)
	\end{matrix} \right\}
\end{equation}
by choosing $\tilde \rho = 0$. It suffices to show the existence of a $\tilde \gamma \in [0,1]$ such that the right hand side of \eqref{eq:R_tilde_gamma} equals $R_1 + R_2$. 

Consider first the case where 
\begin{equation}
	P \ge \frac{2^{2(R_1 + R_2)} - 1}{\beta^2} \ge \frac{2^{2R_2} - 1}{\beta^2} + \frac{2^{2R_2} (2^{2R_1} - 1)}{\alpha^2}, 
\end{equation}
or equivalently, $R_1 + R_2 \le \psi(\beta^2 P)$. In this case we let $\tilde\gamma = 0$ and then we have 
\begin{equation}
	\min \left\{\begin{matrix}
		R_1 + R_2 \\ R_1 + R_{\beta}(\tilde\gamma)\\ R_{\alpha}(\tilde\gamma) + R_{\beta}(\tilde\gamma) 
	\end{matrix}\right\} = \min \left\{\begin{matrix}
		R_1 + R_2 \\ R_1 + \psi(\beta^2 P )\\\psi(\beta^2 P )
	\end{matrix}\right\} = R_1 + R_2.
\end{equation}

Now consider the case where 
\begin{equation}
	\label{eq:P_high_snr_1}
	\frac{2^{2R_2} - 1}{\beta^2} + \frac{2^{2R_2} (2^{2R_1} - 1)}{\alpha^2}\le P < \frac{2^{2(R_1 + R_2)} - 1}{\beta^2}.
\end{equation}
Note that since $|\alpha| \ge |\beta|$, 
\begin{align}
	P &\ge \frac{2^{2R_2} - 1}{\beta^2} + \frac{2^{2R_2} (2^{2R_1} - 1)}{\alpha^2} \notag \\ &\ge \frac{2^{2R_2} (2^{2R_1} - 1) + 2^{2R_2} - 1}{\alpha^2} \notag \\
	&= \frac{2^{2(R_1+ R_2)}- 1}{\alpha^2}, 
\end{align}
which, together with \eqref{eq:P_high_snr_1}, implies  that 
\begin{equation}
	\label{eq:bd_P_R1+R2}
	1+\alpha^2 P  \ge 2^{2(R_1 + R_2)} > 1+\beta^2 P. 
\end{equation}
Since $P > 0$, this further implies $\alpha^2 > \beta^2$, excluding the possibility of $\alpha = \beta$. In this case, we choose $\tilde \gamma$ such that \begin{equation}
	\label{eq:tilde_gamma_equal}
	R_1 + R_2 = R_{\alpha}(\tilde\gamma) + R_{\beta}(\tilde\gamma).
\end{equation} 
We claim that with this choice of $\tilde \gamma$, we have 
\begin{equation}
	\label{eq:tilde_gamma_ge}
	R_1 + R_{\beta}(\tilde\gamma) \ge R_{\alpha}(\tilde\gamma) + R_{\beta}(\tilde\gamma),
\end{equation} 
which, combined with \eqref{eq:tilde_gamma_equal}, proves that the right hand side of \eqref{eq:R_tilde_gamma} is equal to $R_1 + R_2$. 
Solving \eqref{eq:tilde_gamma_equal}, we get
\begin{equation}
	\label{eq:tilde_gamma_explicit}
	\tilde \gamma = \frac{2^{2(R_1 + R_2)} - (1 + \beta^2 P)}{\alpha^2 P (1 + \beta^2 P) - \beta^2 P \cdot 2^{2(R_1 + R_2)}}.
\end{equation}
To verify this $\tilde \gamma \in [0,1]$, first note that both the numerator and the denominator on the right hand side of \eqref{eq:tilde_gamma_explicit} are positive. This is because by \eqref{eq:bd_P_R1+R2}, we have $2^{2(R_1 + R_2)} - 1 - \beta^2 P > 0$ and 
\begin{align}
	\alpha^2 P &(1 + \beta^2 P) - \beta^2 P \cdot 2^{2(R_1 + R_2)} \notag \\ 
	&\ge \alpha^2 P (1 + \beta^2 P) - \beta^2 P(1+\alpha^2 P) \notag \\ 
	&= (\alpha^2 - \beta^2 )P \notag 
	\\&> 0.
\end{align}
Thus $\tilde\gamma > 0$. Furthermore, to show $\tilde\gamma \le 1$ we can equivalently show 
\begin{equation}
	2^{2(R_1 + R_2)} - (1+\beta^2 P) \le \alpha^2 P (1+\beta^2 P) - \beta^2 P \cdot 2^{2(R_1 + R_2)}.
\end{equation}
After rearranging, we can simplify the above inequality as 
\begin{equation}
	\left(2^{2(R_1 + R_2)} - (1+\alpha^2 P)\right)(1+\beta^2 P) \le 0,
\end{equation}
which is true again because of \eqref{eq:bd_P_R1+R2}.

We now prove the claim in \eqref{eq:tilde_gamma_ge}. It suffices to show that $R_1 \ge R_\alpha(\tilde \gamma)$, or $2^{2R_1} -1  \ge 2^{2R_\alpha(\tilde \gamma)} -1 =\tilde\gamma \alpha^2 P$. Plugging in $\tilde\gamma$ from \eqref{eq:tilde_gamma_explicit}, this is equivalent to showing
\begin{align}
	&(2^{2R_1} -1) \left(\alpha^2 P (1+\beta^2 P) - \beta^2 P \cdot 2^{2(R_1 + R_2)}\right) \\
	&\quad\qquad \ge \alpha^2 P \left(2^{2(R_1 + R_2)} - (1 + \beta^2 P)\right) \notag\\
	&\Leftrightarrow ~ 2^{2R_1}  \left(\alpha^2 P (1+\beta^2 P) - \beta^2 P \cdot 2^{2(R_1 + R_2)}\right) \notag \\
	&\quad\qquad\ge (\alpha^2 - \beta^2 ) P \cdot 2^{2(R_1 + R_2)}\\
	&\Leftrightarrow ~ \alpha^2  (1+\beta^2 P) - \beta^2  \cdot 2^{2(R_1 + R_2)} \ge (\alpha^2 - \beta^2 )  2^{2R_2} \\
	&\Leftrightarrow ~ P\ge \frac{2^{2R_2} - 1}{\beta^2} + \frac{2^{2R_2} (2^{2R_1} - 1)}{\alpha^2}.
\end{align}
This completes the proof of the claim, since the last inequality holds due to  \eqref{eq:P_high_snr_1}.

For part (b), we first consider the case where $R_1 > R_2$. In this case $R_{\textnormal{min}} = R_2$. We choose $
\tilde\rho = R_{\textnormal{min}} = R_2$ and $\tilde\gamma = 1$ in \eqref{eq:R_MQ_DF_new_2}, which yields 
\begin{align}
	C_\diamond(R_1, R_2) &\ge \min \left\{\begin{matrix}
		R_1 + R_2 - R_2 \\ R_1 + R_{\beta}(1)\\R_{\alpha}(1) + R_{\beta}(1) +  R_2
	\end{matrix}\right\} \notag \\ &= \min \{ R_1, \psi(\alpha^2 P ) + R_2\}.\label{eq:R_low_SNR_R1>R2}
\end{align}
Rearranging $P \le \frac{2^{2(R_1 - R_2)} - 1}{\alpha^2}$, we get $R_2 + \psi(\alpha^2 P) \le R_1$, thus \eqref{eq:R_low_SNR_R1>R2} simplifies to $C_\diamond(R_1,R_2) \ge R_2 + \psi(\alpha^2 P)$.

We then consider the case when $R_2 > R_1$. In this case $ R_{\textnormal{min}} = R_1$. We choose $
\tilde\rho = R_{\textnormal{min}} = R_1$ and $\tilde\gamma = 0$ in \eqref{eq:R_MQ_DF_new_2}, which yields 
\begin{align}
	C_\diamond(R_1, R_2) &\ge \min \left\{\begin{matrix}
		R_1 + R_2 - R_1 \\ R_1 + R_{\beta}(0)\\R_{\alpha}(0) + R_{\beta}(0) +  R_2
	\end{matrix}\right\} \notag \\ &= \min \{ R_2, \psi(\beta^2 P ) + R_1\}.\label{eq:R_low_SNR_R2>R1}
\end{align}
Rearranging $P \le \frac{2^{2(R_1 - R_2)} - 1}{\alpha^2}$, we get $R_2 + \psi(\alpha^2 P) \le R_1$, thus \eqref{eq:R_low_SNR_R2>R1} simplifies to $C_\diamond(R_1,R_2) \ge R_2 + \psi(\alpha^2 P)$.

Finally, we deal with the case where $R_1 = R_2 = \bar R$. We choose $\tilde \rho = \bar R$. Then it is straightforward from \eqref{eq:R_MQ_DF_new_2} that
\begin{align}
	C_\diamond(\bar R, \bar R) \ge \min \left\{\begin{matrix}
		2\bar R - \bar R \\ \bar R + R_{\beta}(\gamma)\\R_{\alpha}(\gamma) + R_{\beta}(\gamma) +  \bar R
	\end{matrix}\right\}  \ge \bar R \label{eq:R_low_SNR_R2=R1}
\end{align}
for all values $\gamma \in [0,1]$ and $P > 0$. But then in the $P \to 0$ limit, $\psi(\alpha^2 P ) \to 0$ and $\psi(\beta^2 P)\to 0$, which implies 
\begin{equation}
	\lim_{P \to 0} C_{\diamond}(\bar R, \bar R) \le \lim_{P \to 0} \text{Cutset}_{\diamond}(\bar R, \bar R) = \bar R.
\end{equation}
This completes the proof.

\section{Proof of Proposition \ref{prop:sym_diam_simplify_bds}}
\label{app:sym_diam_simplify_bds}

Recall that for the anti-symmetric Gaussian primitive diamond channel, we have $R_1 = R_2 = \bar R$ and $\alpha = -\beta = 1$. As a result, $\alpha^2 = \beta^2 = 1$ and hence $\psi(\alpha^2 P) = \psi(\beta^2 P) = \psi(P)$. Therefore, the cut-set bound in \eqref{eq:diamond_cutset} reduces to the following form:
\begin{align}
	\textnormal{Cutset}_{\textnormal{sym},\diamond}(\bar R) 
	&=\min\left\{
	2\bar R, \bar R+ \psi(P)\right\}  \notag \\
	&=  \begin{cases}
		\bar R + \psi(P), &\psi(P) < \bar R,\\
		2\bar R, &\psi(P) \ge \bar R.\\
	\end{cases}\label{eq:app_sym_diam_cutset}
\end{align}

The DF lower bound in \eqref{eq:diamond_DF} can be simplified to
\begin{align}
	R_{\textnormal{sym},\diamond,\textnormal{DF}}(\bar R) &=\max_{\gamma \in [0,1]} \min \left\{
	\begin{matrix}
		2\bar R\\ \bar R + \psi\left(\frac{\bar\gamma  P }{\gamma  P + 1}  \right)\\\psi(\gamma P)+ \psi\left(\frac{\bar\gamma  P }{\gamma  P + 1}  \right)
	\end{matrix}
	\right\} \notag \\
	&= \max_{\gamma \in [0,1]} \min \left\{
	\begin{matrix}
		2\bar R\\ \bar R + \psi\left(\frac{\bar\gamma  P }{\gamma  P + 1}  \right)\\ \psi( P)
	\end{matrix}
	\right\}, \label{eq:app_sym_diam_DF_1}
\end{align}
where \eqref{eq:app_sym_diam_DF_1} is due to the fact that \begin{equation}
	\label{eq:app_psi_sum}
	\psi(\gamma P) + \psi\left(\frac{\bar{\gamma}  P}{\gamma P + 1}\right) = \psi(P)
\end{equation} for all $\gamma \in [0,1]$. Furthermore, we note that $\psi(\frac{\bar{\gamma} P}{\gamma  P + 1}) = \psi(P) - \psi(\gamma P)$ is a monotonically decreasing function of $\gamma$ for fixed $P$, hence it is maximized at $\gamma = 0$ with optimal value $\psi(P)$. We can thus bound the right-hand side of \eqref{eq:app_sym_diam_DF_1} from below by fixing $\gamma = 0$:
\begin{align}
	\max_{\gamma \in [0,1]} \min \left\{
	\begin{matrix}
		2\bar R\\ \bar R + \psi\left(\frac{\bar\gamma  P }{\gamma  P + 1}  \right)\\ \psi( P)
	\end{matrix}
	\right\}  
	&\ge \min \left\{
	\begin{matrix}
		2\bar R,\\ \bar R + \psi(P)\\ \psi(P) 
	\end{matrix}
	\right\} \notag \\
	&= \min \left\{
	2\bar R, \psi(P) 
	\right\}. \label{eq:app_max_min_ge}
\end{align}
On the other hand, since taking the minimum over a subset cannot yield a smaller value, we have
\begin{align}
	\max_{\gamma \in [0,1]} \min \left\{
	\begin{matrix}
		2\bar R\\ \bar R + \psi\left(\frac{\bar\gamma  P }{\gamma  P + 1}  \right)\\ \psi( P)
	\end{matrix}
	\right\}  
	&\le \max_{\gamma \in [0,1]} \min \left\{
	2\bar R, \psi( P)
	\right\} \notag \\
	&= \min \left\{
	2\bar R, \psi(P) 
	\right\}.\label{eq:app_max_min_le}
\end{align}
Therefore, we conclude that 
\begin{align}
	R_{\textnormal{sym},\diamond,\textnormal{DF}}(\bar R)  &= \min \{2\bar R, \psi(P)\} \notag \\
	&=  \begin{cases}
		\psi(P), &\psi(P) < 2\bar R,\\
		2\bar R, &\psi(P) \ge 2\bar R.\\
	\end{cases}
\end{align}

We can now simplify the MQ-DF lower bound \eqref{eq:R_MQ_DF} to the following:
\begin{align}
	\label{eq:sym_diam_1}
	R_{\textnormal{sym},\diamond,\textnormal{MQ-DF}}(\bar R) =\max_{\substack{\gamma\in [0,1]\\0\le \rho\le \bar R}} \min
	\left\{ 
	\begin{matrix}
		2\bar R - \rho\\
		\bar R+ \psi(\frac{\bar{\gamma} P}{\gamma  P + 1})\\
		\psi(  P) + \rho
	\end{matrix}\right\},
\end{align}
where the last step uses \eqref{eq:app_psi_sum}. Using a similar argument as in \eqref{eq:app_max_min_ge} and  \eqref{eq:app_max_min_le}, we have 
\begin{align}
	\max_{\substack{\gamma\in [0,1]\\0\le \rho\le \bar R}} \min
	\left\{ 
	\begin{matrix}
		2\bar R - \rho\\
		\bar R+ \psi(\frac{\bar{\gamma} P}{\gamma  P + 1})\\
		\psi(  P) + \rho
	\end{matrix}\right\} &\ge \max_{0\le \rho\le \bar R} \min
	\left\{ 
	\begin{matrix}
		2\bar R - \rho\\
		\bar R+ \psi(P)\\
		\psi(  P) + \rho
	\end{matrix}\right\}  \notag \\
	&= \max_{0\le \rho\le \bar R} \min
	\left\{ 
	\begin{matrix}
		2\bar R - \rho\\
		\psi(  P) + \rho
	\end{matrix}\right\} 
\end{align}
and 
\begin{align}
	\max_{\substack{\gamma\in [0,1]\\0\le \rho\le \bar R}} \min
	\left\{ 
	\begin{matrix}
		2\bar R - \rho\\
		\bar R+ \psi(\frac{\bar{\gamma} P}{\gamma  P + 1})\\
		\psi(  P) + \rho
	\end{matrix}\right\} &\le \max_{\substack{\gamma\in [0,1]\\0\le \rho\le \bar R}}\min
	\left\{ 
	\begin{matrix}
		2\bar R - \rho\\
		\psi(  P) + \rho
	\end{matrix}\right\}  \notag \\
	&= \max_{0\le \rho\le \bar R} \min
	\left\{ 
	\begin{matrix}
		2\bar R - \rho\\
		\psi(  P) + \rho
	\end{matrix}\right\}. 
\end{align}
Thus we have
\begin{equation}
	\label{eq:app_sym_diam_MQ_DF_1}
	R_{\textnormal{sym},\diamond,\textnormal{DF}}(\bar R) = \max_{0\le \rho\le \bar R} \min
	\left\{ 
	\begin{matrix}
		2\bar R - \rho\\
		\psi(  P) + \rho
	\end{matrix}\right\}. 
\end{equation}
It remains to solve the optimization on the right-hand side of \eqref{eq:app_sym_diam_MQ_DF_1}. We consider two separate cases. First, when $\psi(P) \ge 2\bar R$, $\psi(P) +\rho \ge 2\bar R - \rho$ for all $\rho \in [0, \bar R]$, therefore 
\begin{equation}
	\max_{0\le \rho\le \bar R } \min
	\left\{ 
	\begin{matrix}
		2\bar R - \rho\\
		\psi(P)+ \rho
	\end{matrix}\right\} =  \max_{0\le \rho\le \bar R }2\bar R - \rho = 2\bar R.
\end{equation}
The maximum is attained at $\rho^*=0$.

Otherwise when $\psi(P) < 2\bar R$, the objective of the outer maximization can be explicitly solved as
\begin{equation}  \min
	\left\{ 
	\begin{matrix}
		2\bar R - \rho\\
		\psi(P)+ \rho
	\end{matrix}\right\} = \begin{cases}
		\psi(P)+ \rho,&0 \le \rho \le  \frac{2\bar R - \psi(P)}{2},\\
		2\bar R - \rho,&  \frac{2\bar R - \psi(P)}{2} \le \rho \le \bar R.\\
	\end{cases}
\end{equation}
Note that this objective function is first increasing and then decreasing in its argument $\rho$, therefore it is maximized exactly at its turning point: $\rho^* =  \frac{2\bar R - \psi(P)}{2}$. Plugging this into the right-hand side of \eqref{eq:app_sym_diam_MQ_DF_1}, we get 
\begin{equation}
	\max_{0\le \rho\le \bar R } \min
	\left\{ 
	\begin{matrix}
		2\bar R - \rho\\
		\psi(P)+ \rho
	\end{matrix}\right\} = \psi(P) + \rho^*=  \bar R + \frac{\psi(P)}{2}
\end{equation}
when $\psi(P) < 2\bar R$. Concluding, we have
\begin{equation}
	\label{eq:sym_diam_4}
	R_{\textnormal{sym},\diamond,\textnormal{MQ-DF}}(\bar R) = \begin{cases}
		\bar R + \frac{\psi(P)}{2}, &\psi(P) < 2\bar R,\\
		2\bar R, &\psi(P) \ge 2\bar R.
	\end{cases}
\end{equation}

\section{Proof of Proposition \ref{prop:sym_cf}}
\label{app:sym_cf}

As discussed in Section \ref{sec:bounds_sym}, the CF achievability \eqref{eq:R_MQ_DF} simplifies to \eqref{eq:sym_cf_1} for the anti-symmetric case of the Gaussian primitive diamond channel. We aim to analyze this achievable rate in more detail. Define
\begin{equation}
	\label{eq:tildeRN}
	\tilde R(N) = \min \left\{\begin{matrix}
		\frac12 \log \left(1+\frac{2P}{N}\right)\\
		\bar R + \frac12 \log \left(1 + \frac{P-1}{N+2}\right)\\
		2\bar R + \frac12 \log \left(\frac{N}{N+2}\right)
	\end{matrix}\right\},
\end{equation}
i.e., $\tilde R(N)$ is the objective function of the maximization on the right-hand side of \eqref{eq:sym_cf_1}. We claim that the maximum of $\tilde R(\cdot)$ is attained at $N^*$, the unique positive solution of the equation below:
\begin{equation}
	\frac12 \log \left(1+\frac{2P}{N^*}\right) = 2\bar R + \frac12 \log \left(\frac{N^*}{N^*+2}\right).   
	\label{eq:rho_N_star}
\end{equation}
We further claim that 
\begin{equation}
	\tilde R(N^*) = \frac12 \log \left(1+\frac{2P}{N^*}\right) = 2\bar R + \frac12 \log \left(\frac{N^*}{N^*+2}\right),
	\label{eq:tilde_R_N_star}
\end{equation}
or equivalently, 
\begin{equation}
	\bar R + \frac12 \log \left(1 + \frac{P-1}{N^*+2}\right)\ge 2\bar R + \frac12 \log \left(\frac{N^*}{N^*+2}\right).
	\label{eq:rho_N_star_ge}
\end{equation}

We first prove \eqref{eq:rho_N_star_ge}, which implies \eqref{eq:tilde_R_N_star}. Solving \eqref{eq:rho_N_star}, we get
\begin{equation}\label{eq:N_star}
	N^* = \frac{P+1+ \sqrt{(P-1)^2 + 4P\cdot 2^{4\bar R}}}{2^{4\bar R} - 1}.
\end{equation}
Rearranging \eqref{eq:rho_N_star_ge} and plugging in the value of $N^*$ from \eqref{eq:N_star}, it suffices to show that 
\begin{align}
	0 &\le \frac12 \log \left( 1 + \frac{P+1}{N^*}\right) - \bar R \notag \\&= \frac{1}{2} \log \frac{(P+1)\cdot 2^{4\bar R} + \sqrt{(P-1)^2 + 4P\cdot 2^{4\bar R}}}{2^{2\bar R}\left(P+1+\sqrt{(P-1)^2 + 4P\cdot 2^{4\bar R}}\right)} ,
\end{align}
or equivalently,
\begin{align}
	(P+1)&\cdot 2^{4\bar R} + \sqrt{(P-1)^2 + 4P\cdot 2^{4\bar R}} \notag \\&\ge 2^{2\bar R}\left(P+1+\sqrt{(P-1)^2 + 4P\cdot 2^{4\bar R}}\right).
\end{align}
Subtracting the right-hand side from the left-hand side, we can see that the difference is nonnegative:
\begin{align}
	&(P+1)\cdot 2^{4\bar R} + \sqrt{(P-1)^2 + 4P\cdot 2^{4\bar R}} \notag \\ &- 2^{2\bar R}\left(P+1+\sqrt{(P-1)^2 + 4P\cdot 2^{4\bar R}}\right) \notag \\
	& = \left((P+1)\cdot 2^{2\bar R} - \sqrt{(P-1)^2 + 4P\cdot 2^{4\bar R}} \right) (2^{2\bar R} - 1) \notag \\
	& = \frac{(P+1)^2\cdot 2^{4\bar R} - \left((P-1)^2 + 4P\cdot 2^{4\bar R}\right)}{(P+1) 2^{2\bar R} + \sqrt{(P-1)^2 + 4P\cdot 2^{4\bar R}}} (2^{2\bar R} - 1) \notag \\
	& = \frac{(P-1)^2\cdot (2^{4\bar R} - 1)\cdot(2^{2\bar R} - 1)}{(P+1) 2^{2\bar R} + \sqrt{(P-1)^2 + 4P\cdot 2^{4\bar R}}} ~\ge~ 0,
\end{align}
where in the last step we note that $2^{4\bar R} - 1 > 2^{2\bar R} - 1 > 0$ for all $\bar R > 0$. This concludes the proof of \eqref{eq:rho_N_star_ge}. 

We now prove that $N^*$ given in \eqref{eq:N_star} maximizes $\tilde R(\cdot)$. 
For any $N < N^*$, note that
\begin{align}
	\tilde R(N) &\le 2\bar R + \frac{1}{2}\log\left(\frac{N}{N+2}\right) \\ &\le 2\bar R + \frac{1}{2}\log\left(\frac{N^*}{N^*+2}\right) \label{eq:tilde_R_N_monotone_11}\\&= \tilde R(N^*), \label{eq:tilde_R_N_monotone_12}
\end{align}
where \eqref{eq:tilde_R_N_monotone_11} is due to the fact that $\frac12\log(\frac{N}{N+2})$ is an increasing function in $N$, and \eqref{eq:tilde_R_N_monotone_12} is due to \eqref{eq:tilde_R_N_star}. 
For any $N > N^*$, note that
\begin{align}
	\tilde R(N) &\le \frac{1}{2}\log\left(1+\frac{2P}{N}\right)\\ &\le \frac{1}{2}\log\left(1+\frac{2P}{N^*}\right) \label{eq:tilde_R_N_monotone_21}\\&= \tilde R(N^*), \label{eq:tilde_R_N_monotone_22}
\end{align}
where \eqref{eq:tilde_R_N_monotone_21} is due to the fact that $\frac12\log(1+\frac{2P}{N})$ is a decreasing function in $N$, and \eqref{eq:tilde_R_N_monotone_22} is due to \eqref{eq:tilde_R_N_star}. Summarizing, we conclude the proof that the $N^*$ in \eqref{eq:N_star} maximizes $\tilde R(\cdot)$.

We now prove part (a) of Proposition \ref{prop:sym_cf}. Recall that $R_{\text{sym},\diamond,\text{MQ-DF}}(\bar R)$ is defined in \eqref{eq:sym_diam_MQ_DF_2}. It suffices to prove 
\begin{equation}
	\begin{cases}
		\tilde R(N^*) < \bar R + \frac{\psi(P)}{2},~ &P < 2^{4\bar R} - 1,\\
		\tilde R(N^*) < 2\bar R, ~&P \ge 2^{4\bar R} - 1.
	\end{cases}
\end{equation}
The case of $P \ge 2^{4\bar R} - 1$ is straightforward by noting that 
\begin{equation}
	\tilde R(N^*) = 2\bar R + \frac12 \log \left(\frac{N^*}{N^*+2}\right) < 2\bar R.
\end{equation}
For the case of $P < 2^{4\bar R} - 1$, note that
\begin{align}
	\tilde R(N^*) &= \frac12 \log\left(1+\frac{2P}{N^*}\right) \notag \\
	&= \frac12 \log\left(1+\frac{2P(2^{4\bar R} - 1)}{P+1+  \sqrt{(P-1)^2 + 4P\cdot 2^{4\bar R}}}\right) \notag\\
	&< \frac12 \log\left(1+\frac{2P(2^{4\bar R} - 1)}{P+1+ 2\sqrt{P}\cdot 2^{2\bar R}}\right). \label{eq:rho_N_star_lb}
\end{align}
Therefore
\begin{align}
	&R_{\textnormal{sym},\diamond,\textnormal{MQ-DF}}(\bar R) - R_{\textnormal{sym},\diamond,\textnormal{CF}}(\bar R) \notag\\
	&>\frac12 \log(2^{2\bar R} \sqrt{P+1}) - \frac12 \log\left(1+\frac{2P(2^{4\bar R} - 1)}{P+1+ 2\sqrt{P}\cdot 2^{2\bar R}}\right) \notag  \\
	&= \frac12 \log\left(\frac{(2^{2\bar R} \sqrt{P+1})(P+1+ 2\sqrt{P}\cdot 2^{2\bar R})}{2P\cdot2^{4\bar R} + 2\sqrt{P}\cdot 2^{2\bar R} - (P-1)}\right) \notag \\
	&= \frac12 \log\left(\frac{(2^{2\bar R} \sqrt{P+1})(P+1+ 2\sqrt{P}\cdot 2^{2\bar R})}{D}\right),
	\label{eq:app_A_diff_denom_D}
\end{align}
where we define
\begin{equation}
	D = 2P\cdot2^{4\bar R} + 2\sqrt{P}\cdot 2^{2\bar R} - (P-1).
	\label{eq:D_app_prop_5}
\end{equation}
Note that for all $P > 0$ and $\bar R > 0$, 
\begin{equation}
	D > 2P + 2\sqrt{P} - (P-1) = (\sqrt{P} + 1)^2 > 0.
\end{equation}
We now consider two different cases: $P \ge 1$ and $P < 1$. When $P \ge 1$, we have that 
\begin{align}
	&\frac12 \log\left(\frac{(2^{2\bar R} \sqrt{P+1})(P+1+ 2\sqrt{P}\cdot 2^{2\bar R})}{D}\right) \notag \\&> \frac12 \log\left(\frac{(2^{2\bar R} \sqrt{P})(P+1+ 2\sqrt{P}\cdot 2^{2\bar R})}{D}\right)\notag  \\
	&= \frac12 \log\left(\frac{2P\cdot2^{4\bar R} + (P+1)\sqrt{P}\cdot 2^{2\bar R} }{D}\right) \notag \\
	&= \psi\left( \frac{(P-1)(\sqrt{P}\cdot 2^{2\bar R} + 1) }{D}\right)\label{eq:app_A_P_ge_1}
\end{align}
Since $P \ge 1$ and $\bar R > 0$, we have $P - 1 \ge 0$ and $\sqrt{P}\cdot 2^{2\bar R} + 1 > 0$. Also, since $\psi$ maps positive reals to positive reals, the right-hand side of \eqref{eq:app_A_P_ge_1} is greater than zero.

When $P < 1$, we have
\begin{align}
	&\frac12 \log\left(\frac{(2^{2\bar R} \sqrt{P+1})(P+1+ 2\sqrt{P}\cdot 2^{2\bar R})}{D}\right) \notag \\
	&> \frac12 \log\left(\frac{2^{2\bar R} (P+1+ 2\sqrt{P}\cdot 2^{2\bar R})}{D}\right) \notag \\
	&= \psi\left(\frac{2\sqrt{P}(1-\sqrt{P})2^{4\bar R} + (P-2\sqrt{P}+1) 2^{2\bar R} + (P-1)}{D}\right) \notag \\
	&= \psi\left(\frac{(1-\sqrt{P})\!\cdot\!(2^{2\bar R} - 1)\!\cdot\!(\sqrt{P}(2\cdot 2^{2\bar R} + 1) +1 )}{D}\right).\label{eq:app_A_P_le_1}
\end{align}
Given $P< 1$ and $\bar R > 0$, we have $1 - \sqrt{P} > 0$, $2^{2\bar R} - 1 > 0$ and $\sqrt{P}(2\cdot 2^{2\bar R} + 1)+1 > 0$. Hence the argument of $\psi$ on the right-hand side of \eqref{eq:app_A_P_le_1} is positive, which therefore imply that the entire right-hand side of \eqref{eq:app_A_P_le_1} is positive. Summarizing, the right-hand side of \eqref{eq:app_A_diff_denom_D} is positive for both $P \ge 1$ and $P < 1$, thus $R_{\textnormal{sym},\diamond,\textnormal{MQ-DF}}(\bar R) > R_{\textnormal{sym},\diamond,\textnormal{CF}}(\bar R)$ always holds, proving part (a).

For part (b), recall that from \eqref{eq:tilde_R_N_star} and \eqref{eq:N_star}, we have 
\begin{gather}
	\label{eq:R_sym_diam_CF}
	R_{\textnormal{sym},\diamond,\textnormal{CF}}(\bar R) = \psi\left(\frac{2P}{N^*}\right),\\N^* = \frac{P+1+ \sqrt{(P-1)^2 + 4P\cdot 2^{4\bar R}}}{2^{4\bar R} - 1}.
\end{gather}
Also, $R_{\textnormal{sym},\diamond,\textnormal{MQ}}(\bar R) = \bar R = \psi(2^{2\bar R} - 1)$. Since $\psi$ is a strictly increasing function, it suffices to show that 
\begin{equation}
	\frac{2P}{N^*} \begin{cases}
		> 2^{2\bar R} - 1,  & P > 1,\\
		= 2^{2\bar R} - 1,  & P = 1,\\
		< 2^{2\bar R} - 1,  & P < 1.\\
	\end{cases}
\end{equation}
Equivalently, we need to show
\begin{equation}
	N^* \begin{cases}
		<\frac{2P}{ 2^{2\bar R} - 1},  & P > 1,\\
		=\frac{2P}{ 2^{2\bar R} - 1},  & P = 1,\\
		>\frac{2P}{ 2^{2\bar R} - 1},  & P < 1.\\
	\end{cases}
\end{equation}
We now consider the cases separately.

(i) $P > 1$. Observe that 
\begin{align}
	(2^{4\bar R} - 1) N^* &= P+1+ \sqrt{(P-1)^2 + 4P\cdot 2^{4\bar R}} \notag \\&\le (P+1) + (P-1) + 2\sqrt{P}\cdot 2^{2\bar R} \notag \\
	&< 2P ( 2^{2\bar R} + 1) \notag \\
	&= (2^{4\bar R} - 1) \cdot \frac{2P}{2^{2\bar R} - 1},
\end{align}
where the second line is due to the fact that $\sqrt{a+b} \le \sqrt{a} + \sqrt{b}$ for all $a,b \ge 0$, and the third line is because $\sqrt{P} < P$ for all $P > 1$. Since $R > 0$, $2^{4\bar R} - 1 > 0$, we thus conclude that $N < \frac{2P}{2^{2\bar R} - 1}$.

(ii) $P = 1$. By plugging in the value of $P$ in the expression of $N$, we have
\begin{align}
	N^* = \frac{2+  2\cdot 2^{2\bar R}}{2^{4\bar R} - 1} = \frac{2(2^{2\bar R} + 1)}{(2^{2\bar R} - 1)(2^{2\bar R} + 1)} = \frac{2P}{2^{2\bar R} - 1}.
\end{align}

(iii) $P < 1$. Similarly we have 
\begin{align}
	(2^{4\bar R} - 1) N^* &= P+1+ \sqrt{(P-1)^2 + 4P\cdot 2^{4\bar R}} \notag \\
	&> (P+1) +  2\sqrt{P}\cdot 2^{2\bar R} \notag \\
	&> 2P  +  2P \cdot 2^{2\bar R} \notag \\
	&= (2^{4\bar R} - 1) \cdot \frac{2P}{2^{2\bar R} - 1},
\end{align}
where the second line is due to the fact that $(P-1)^2 > 0$ when $P < 1$, and the third line is because $\sqrt{P} > P$ for all $P > 1$. Since $R > 0$, $2^{4\bar R} - 1 > 0$, we thus conclude that $N > \frac{2P}{2^{2\bar R} - 1}$.

\section{Proof of Proposition \ref{prop:sym_mq_cf}}
\label{app:sym_mq_cf}

Observe that we can rewrite \eqref{eq:sym_mq_cf} as the following:
\begin{align}
	\label{eq:app_sym_mq_cf}
	&R_{\textnormal{sym},\diamond,\textnormal{MQ-CF}}(\bar R)  \notag \\
	&= \max_{\substack{N > 0\\\rho \in [0,\bar R]}} \!\!\min \left\{\begin{matrix}
		\frac12 \log \left(1+\frac{2P}{N}\right) + \rho\\
		\bar R + \frac12 \log \left(1 + \frac{P-1}{N+2}\right)\\
		2\bar R + \frac12 \log \left(\frac{N}{N+2} \right)- \rho
	\end{matrix}\right\}\\
	&= \max_{\substack{N > 0\\\rho \in [0,\bar R]}} \min \left\{\begin{matrix}
		\frac12 \log \left(1+\frac{2P}{N}\right) \\
		(\bar R - \rho) + \frac12 \log \left(1 + \frac{P-1}{N+2}\right)\\
		2(\bar R -\rho) + \frac12 \log \left(\frac{N}{N+2} \right)
	\end{matrix}\right\} + \rho\\
	&= \max_{\bar \rho \in [0,\bar R]} \max_{N > 0} \,\min \left\{\begin{matrix}
		\frac12 \log \left(1+\frac{2P}{N}\right) \\
		\bar \rho + \frac12 \log \left(1 + \frac{P-1}{N+2}\right)\\
		2 \bar\rho + \frac12 \log \left(\frac{N}{N+2} \right)
	\end{matrix}\right\} + R - \bar\rho \label{eq:mq_cf_rewrite_1}\\
	&= \max_{\bar \rho \in [0,\bar R]}  R_{\textnormal{sym}, \diamond, \textnormal{CF}}(\bar\rho) + R - \bar\rho,\label{eq:mq_cf_rewrite_2}
\end{align}
where we introduce the new variable $\bar \rho = \bar R - \rho$ in \eqref{eq:mq_cf_rewrite_1}, and \eqref{eq:mq_cf_rewrite_2} is due to \eqref{eq:sym_cf_1}. We can interpret the parameter $\bar\rho$ as the portion of each relay-link rate that is allocated to CF. For simplicity of notation, we define
\begin{equation}
	g(\bar\rho) =  R_{\textnormal{sym}, \diamond, \textnormal{CF}}(\bar\rho) + \bar R - \bar\rho,
\end{equation}
and thus $R_{\textnormal{sym},\diamond,\textnormal{MQ-CF}}(\bar R) = \max_{\bar\rho \in [0, \bar R]} g(\bar \rho)$.

When $P \le 1$, we first show that $g(\bar \rho)$ is upper bounded by $\bar R$ for all $\bar\rho \in [0, \bar R]$, and then show that this upper bound is attained for a specific choice of $\bar \rho$. By Proposition \ref{prop:sym_cf}, we have $R_{\textnormal{sym}, \diamond, \textnormal{CF}}(\bar \rho) \le \bar \rho $ for all $\bar \rho > 0$. For $\bar \rho = 0$, it is also easy to verify that $R_{\textnormal{sym}, \diamond, \textnormal{CF}}(0) = \psi(0) = 0$. Therefore,
\begin{align}
	g(\bar \rho)\le \bar \rho + (\bar R - r) = \bar R
\end{align}
for all $\bar \rho \in [0, \bar R]$, and hence 
\begin{equation}
	R_{\textnormal{sym},\diamond,\textnormal{MQ-CF}}(\bar R) = \max_{\bar\rho \in [0, \bar R]} g(\bar \rho) \le R.
\end{equation}
On the other hand, 
\begin{equation}
	R_{\textnormal{sym},\diamond,\textnormal{MQ-CF}}(\bar R) = \max_{\bar\rho \in [0, \bar R]} g(\bar \rho) \ge g(0) = \bar R.
\end{equation}
This concludes $R_{\textnormal{sym},\diamond,\textnormal{MQ-CF}}(\bar R) = \bar R$ for $P \le 1$.

When $P > 1$, we claim that $g(\bar \rho)$ is maximized when $\bar \rho = \bar R$, and thus the maximum is $g(\bar R) = R_{\textnormal{sym}, \diamond, \textnormal{CF}}(\bar R)$. We first note that 
\begin{equation}
	R_{\textnormal{sym}, \diamond, \textnormal{CF}}(\bar \rho) =  \psi\left( \frac{2P(2^{4\bar \rho} - 1)}{P+1 + \sqrt{(P-1)^2 + 4P \cdot 2^{4\bar \rho}}} \right)
\end{equation}
due to \eqref{eq:tilde_R_N_star} and \eqref{eq:N_star}. Now we differentiate $g$. Define $u: \mb R \to \mb R$ and $v:\mb R_+ \to \mb R$ such that $u(x) = 2^{4x}$ and $v(x) = 1+\frac{2P(x-1)}{P+1+\sqrt{(P-1)^2 + 4P x}}$, then
\begin{equation}
	g(\bar\rho) = \frac12 \log v(u(\bar\rho)) + \bar R - \bar \rho.
\end{equation}
By chain rule, 
\begin{align}
	g'(\bar \rho) &= \frac{v'(u(\bar \rho)) \cdot u'(\bar \rho)}{2v(u(\bar \rho))\ln2 }  -1.
\end{align}
Note that since $u'(\bar \rho) =  4\ln 2 \cdot 2^{4\bar \rho} =  4\ln 2 \cdot u(\bar \rho)$, we have
\begin{equation}
	g'(\bar \rho) = \frac{2u(\bar \rho)}{v(u(\bar \rho))} v'(u(\bar \rho)) - 1.
\end{equation}
To show that $g$ attains maximum at $\bar R$, it suffices to show that $g'(\bar \rho) > 0$ for all $\bar \rho \in  [0, \bar R]$. Note that since $\bar \rho \ge 0$, $u(\bar \rho) \ge 1$ and $v(u(\bar \rho)) \ge 1$, it suffices to show
\begin{equation}
	v(u(\bar \rho))\cdot g'(\bar \rho) =  2u(\bar \rho)\cdot v'(u(\bar \rho)) - v(u(\bar \rho)) > 0.
\end{equation}
Define $w = \sqrt{(P-1)^2 + 4P u}$. Omitting all $\bar \rho$'s, we have
\begin{align}
	2u\cdot v'(u)& = 2u\cdot \frac{2Pw(P+1+w) - 4P^2(u-1)  }{w(P+1+w)^2}, \\
	v(u) &= \frac{w+1+2Pu - P}{P + 1 + w} \notag \\&= \frac{(w+1+2Pu - P)w(P+1+w)}{w(P+1+w)^2}.
\end{align}
and hence 
\begin{equation}
	\label{eq:2uv'-v}
	2u \cdot v'(u) - v(u) = \frac{S_1 + S_2 }{w(P+1+w)^2}
\end{equation}
where $S_1 = (2Pu + P-w-1)w(P+1+w)$, $S_2 =  8P^2u(1-u)$. Since $u > 0$ and $P > 1$, we have that $w > P - 1 > 0$ and hence $P+1+w > w > 0$, making the denominator of the right hand side of \eqref{eq:2uv'-v} positive. It only remains to check that $S_1 + S_2 > 0$. Now note that by the way we defined $w$, we have 
\begin{equation}
	u = \frac{w^2 - (P-1)^2}{4P},~1-u = \frac{(P+1)^2 - w^2}{4P},
\end{equation}
therefore
\begin{align}
	S_1 &= (2P\frac{w^2 - (P-1)^2}{4P}  + P-w-1)w(P+1+w) \notag \\
	&= \frac{w(w+P-3)(w-P+1)(P+1+w)}{2} 
\end{align}
and 
\begin{align}
	S_2 &= 8P^2 \cdot\frac{w^2 - (P-1)^2}{4P} \frac{(P+1)^2 - w^2}{4P} \notag \\
	&= \frac{(w+P-1)(w-P+1)(P+1+w)(P+1-w)}{2}.
\end{align}
Note that $S_1$ and $S_2$ have common factors $S_0=(w-P+1)(P+1+w)/2$, therefore
\begin{align}
	S_1 + S_2 &= S_0\left( w(w+P-3) + (w+P-1)(w-P+1)\right) \notag \\
	&= S_0(P-1)(1+w)
\end{align}
Recall that $w > P-1$ and $P + 1 + w > 0$, we thus have $S_0 > 0$. Also since $P > 1$ and $w > 0$, we can finally conclude that $S_1 + S_2 > 0$, which proves that the function $g$ is strictly increasing on $[0,\bar R]$. As a consequence, we have proved that $R_{\textnormal{sym}, \diamond, \textnormal{MQ-CF}}(\bar R) = R_{\textnormal{sym}, \diamond, \textnormal{CF}}(\bar R)$ when $P > 1$.

\section{Proof of Proposition \ref{prop:prim_nonperf_bd}}
\label{app:prim_nonperf_bd}

For the Gaussian primitive relay channel with $\lambda$-correlated noises in Section \ref{sec:prim_relay_nonperfect}, the cut-set bound states that 
\begin{equation}
	C_\triangleleft(R_0) \le \min\{I(X;Y,Y_0),I(X;Y) + R_0\}.
\end{equation}
Recall that $Y = X + Z$ and $Y_0 = \alpha X + Z_0$, and that the noise variables $Z$ and $Z_0$ are standard Gaussian variables with correlation efficient $\lambda$. Now, the mutual information expressions $I(X;Y,Y_0)$ and $I(X;Y)$ are maximized when $X \sim \mc N(0,P)$ (see, for example, \cite[Chapter 12.2]{cover1999elements}). In this case, the random vector $(X,Y,Y_0)$ is jointly Gaussian with zero mean and covariance matrix 
\begin{equation}
	\begin{bmatrix}
		P & P & \alpha P\\
		P & P+ 1 & \alpha P + \lambda\\
		\alpha P & \alpha P + \lambda & \alpha^2 P + 1
	\end{bmatrix}.
\end{equation}
By explicitly computing $I(X;Y,Y_0)$ and $I(X;Y)$, we obtain the following simplified cut-set bound:
\begin{align}
	C_\triangleleft(R_0) &\le \min\left\{ \psi\left( \frac{(\alpha - \lambda)^2 P}{1-\lambda^2} + P\right), \psi(P) + R_0 \right\} \notag \\
	&= \textnormal{Cutset}_\triangleleft(R_0)
\end{align}

To lower bound the capacity, we first consider the CF scheme. Recall that the best achievable CF rate is given by the following optimization problem:
\begin{equation}
	\begin{aligned}
		&\sup_{X - Y_0 - \hat Y_0} && I(X;Y,\hat Y_0),\\
		&\quad\ \st&& I(Y_0;\hat{Y}_0|Y) \le R_0.
	\end{aligned}
\end{equation}
In particular, we can choose $X\sim \mc N(0,P)$ and $\hat Y_0 = Y + W$ where $W \sim \mc N(0,N_W)$ is a Gaussian random variable independent of all other variables such that $I(Y_0;\hat Y_0|Y) = R_0$. We can explicitly compute the value of $N_W$ as
\begin{equation}
	N_W = \frac{1-\lambda^2 + (\alpha - \lambda)^2 \frac{P}{P+1}}{2^{2R_0} -1}.
\end{equation}
Plugging in this value, we obtain the following CF lower bound of the capacity:
\begin{align}
	C_\triangleleft(R_0) &\le I(X;Y,\hat Y_0) \notag \\
	&= \psi(P) + \psi\left( \frac{(2^{2R_0} - 1)(\alpha - \lambda)^2 P}{ 2^{2R_0} (1-\lambda^2)(P+1) + (\alpha - \lambda)^2 P}\right) \notag \\
	& = R_{\triangleleft,\textnormal{CF}}(R_0) 
\end{align}

We can also obtain a lower bound of the capacity via DF. In particular, superposition coding is applied, where the transmitter generates $2^{n\tilde R}$ sequences $\tilde x^n(\tilde m)$ according to $\mc N(0, \gamma P)$ and $2^{n \breve R}$ sequences $ {\breve x}^n(\breve m)$ according to $\mc N(0, \bar \gamma P)$ randomly and independently. To send the message pair $(\tilde m, \breve m)$, the transmitter sends $x^n(\tilde m, \breve m) =\tilde x^n(\tilde m) + \breve x^n(\breve m)$. The relay recovers $\tilde m$ from $y_0^n = \alpha  x^n(\tilde m, \breve m) +  z_0^n =  \alpha \tilde x^n(\tilde m) +( \alpha \breve x^n(\breve m) + z_0^n)$ by treating $\alpha \breve x^n(\breve m)$ as part of the noise. It then forwards $\tilde m$ to the receiver through the noiseless link. Provided that $\tilde R < \min\{\psi(\frac{\alpha^2 \gamma P}{\alpha^2 \bar \gamma P + 1}), R_0\}$, the relay is able to decode $\tilde m$ with vanishing error probability and relay it to the receiver. Upon receiving the index $\tilde m$, the receiver is able to subtract $\tilde x^n(\tilde m)$ from its observation $y^n = \tilde x^n(\tilde m) + \breve x^n(\breve m) + z^n$, and then recover $\breve m$ from $\breve x^n(\breve m) + z^n$. The decoding error vanishes when $\breve R <\psi(\bar\gamma P)$. We refer to the maximum overall rate for this scheme as the DF lower bound, which is given in the following:   
\begin{align}
	&C_\triangleleft(R_0) \le I(X;Y,\hat Y_0) \notag \\
	&= \max_{\gamma \in [0,1]} \psi(\bar\gamma P) + \min\left\{\psi\left(\frac{\alpha^2 \gamma P}{\alpha^2 \bar \gamma P + 1}\right), R_0\right\} \notag \\
	&= \begin{cases}
		\psi(P),  & \alpha^2 \le 1,  \notag \\
		\psi(\alpha^2 P),  & 1 < \alpha^2 \le \frac{2^{2R_0}-1}{P},  \notag \\
		\psi \left(P+\frac{(\alpha^2-1)(2^{2R_0} - 1)}{\alpha^2}\right),  & \alpha^2 > \max \left\{1, \frac{2^{2R_0}-1}{P}\right\},
	\end{cases} \notag \\
	&= R_{\triangleleft,\textnormal{DF}}(R_0). 
\end{align}

\section{Calculation Error in \cite{sanderovich2008communication}}
\label{app:bpsk_Gaussian_CF}

In \cite[Section VI-B]{sanderovich2008communication}, the following simple BPSK scheme is considered for the anti-symmetric Gaussian primitive diamond channel with $\lambda = 0$ and $\bar R = 1$. The transmitted signal is
\begin{equation}
	X = \begin{cases}
		-\sqrt{P},  & M = 0,\\
		\sqrt{P},  & M = 1.
	\end{cases}
\end{equation}
The relays forward $V_i = \mb I\{Y_i \ge 0\}$. This induces an effective channel from $X$ to $(V_1, V_2)$, and by coding over this channel, the following rate can be achieved:
\begin{equation}
	R_{\textnormal{sym},\diamond,\textnormal{bpsk}}(1) = I(X;V_1,V_2).
\end{equation}
We now explicitly compute this rate. Note that
\begin{equation}
	I(X;V_1, V_2) = H(V_1, V_2) - H(V_1, V_2|X). \label{eq:i_x_v1v2}
\end{equation}
Conditioning on $X = -\sqrt{P}$, 
\begin{align}
	\Pr\{V_i = 1| X = -\sqrt{P}\} = Q(\sqrt{P}), ~i \in \{1,2\}
\end{align}
where 
\begin{equation}
	Q(x) = \frac{1}{\sqrt{2\pi}} \int_{x}^\infty e^{-\frac12 z^2} dz.
\end{equation}
By symmetry, 
\begin{equation}
	\Pr\{V_i = 0| X = \sqrt{P}\} = Q(\sqrt{P}), ~i \in \{1,2\}.
\end{equation}
Since $V_1$ and $V_2$ are independent conditioning on $X$, we have 
\begin{equation}
	H(V_1, V_2|X) = H(V_1|X) + H(V_2|X).
\end{equation}
Define $t = Q(\sqrt{P})$. We can further expand
\begin{align}
	H(V_1, V_2|X) &= H(V_1|X) + H(V_2|X) \notag \\
	&= -2( t \log t+   (1-t) \log (1-t )). \label{eq:h_v1v2_x}
\end{align} 
Meanwhile, observe that 
\begin{align}
	\Pr\{(V_1,V_2)=(0,0)\} &= \Pr\{(V_1,V_2)=(1,1)\} \notag \\ &= \frac{\left(Q(\sqrt{P})\right)^2 + \left(1-Q(\sqrt{P})\right)^2}{2}, \\
	\Pr\{(V_1,V_2)=(0,1)\} &= \Pr\{(V_1,V_2)=(1,0)\} \notag \\ &= Q(\sqrt{P})  (1-Q(\sqrt{P})).
\end{align}
Therefore, we have
\begin{align}
	H(V_1, V_2) &= - ( t^2 + (1-t)^2 ) \log \frac{t^2 + (1-t)^2}{2} \notag \\ 
	&\qquad ~~- 2 t(1-t) \log (t(1-t)) \notag \\
	&= -2\bigg( (t^2 - t + \half ) \log (t^2 - t + \half) \notag \\
	&\qquad~~+ (t - t^2) \log t + (t - t^2) \log(1-t) \bigg) \label{eq:h_v1v2_2}.
\end{align}
Plugging \eqref{eq:h_v1v2_x} and \eqref{eq:h_v1v2_2} back into \eqref{eq:i_x_v1v2}, we have
\begin{equation}
	R_{\textnormal{sym},\diamond,\textnormal{bpsk}}(1) = I(X;V_1,V_2) = G(Q(\sqrt{P})),
\end{equation}
where
\begin{align}
	G(x)&= 2\bigg(x^2 \log x + (1-x)^2 \log(1-x) \notag \\ &\qquad -(x^2 - x + \half)\log(x^2 - x + \half)\bigg).
\end{align}
However in \cite{sanderovich2008communication}, it is mistakenly stated that $ R_{\textnormal{sym},\diamond,\textnormal{bpsk}}(1) = G(Q(\sqrt{2P}))$. 

On the other hand, the Gaussian CF rate can be obtained by setting $\lambda = 0$ in \eqref{eq:nonperf_diamond_CF}, or alternatively given by the following formula as stated in \cite[Section VI-A]{sanderovich2008communication}:
\begin{equation}
	R_{\textnormal{sym},\diamond,\textnormal{CF}}(1) = \psi\!\left(2P\!\left(\!1 - \frac{\sqrt{P^2 + 2^{4\bar R} (1+2P)} - P}{2^{4\bar R}}\!\right)\! \right).
\end{equation}

\begin{figure}[t!]
	\centering
	\includegraphics[width=0.9\linewidth]{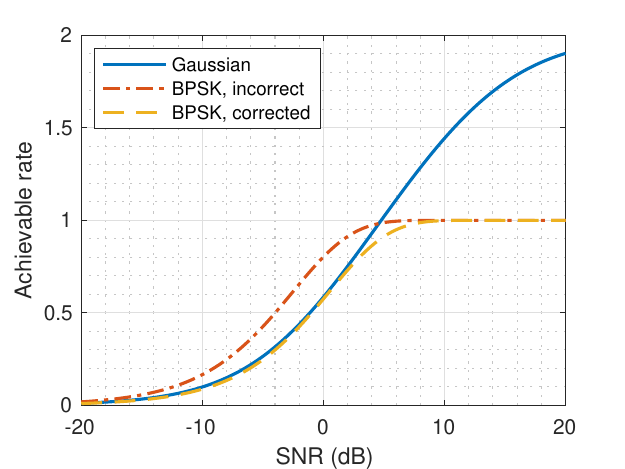}
	\caption{Comparison of the CF achievable rates with Gaussian versus BPSK signalling and quantization for the anti-symmetric Gaussian primitive diamond channel with $\lambda = 0$ and $\bar R = 1$.}
	\label{fig:app_bpsk}
\end{figure}

In Fig.~\ref{fig:app_bpsk}, we compare the achievable rates of CF with Gaussian versus BPSK signalling and quantization. We also plot the incorrect achievable rate in \cite{sanderovich2008communication} as a reference. The curve corresponding to the Gaussian CF rate is always above the one corresponding to the BPSK CF rate, thus invalidating the counterexample in Section VI-B of \cite{sanderovich2008communication}.

\section{Proof of Proposition \ref{prop:diam_nonperf_bd}}
\label{app:diam_nonperf_bd}

We first notice that the DF scheme does not take into consideration the correlation between the noises, hence the lower bound \eqref{eq:diamond_DF} does not change for this scenario. Meanwhile, the general simultaneous CF lower bound \eqref{eq:diamond_CF_1} also holds. Again, we restrict to Gaussian codebook and Gaussian quantization and set $X \sim \mc N(0,P)$, $\hat Y_1 = Y_1 + U_1$ and $\hat Y_2 = Y_2 + U_2$, where $U_1 \sim \mc N(0, N_1)$, $U_2 \sim \mc N(0, N_2)$ are such that $U_1 \indep U_2$ and $(U_1,U_2) \indep (X, Y_1, Y_2)$. Observe that the covariance of the random vector $(X,Y_1, Y_2, \hat Y_1, \hat Y_2)$ is given by
\begin{equation}
	\begin{bmatrix}
		P & \alpha P & \beta P & \alpha P & \beta P \\
		\alpha P & \alpha^2 P + 1 & \alpha \beta P + \lambda & \alpha^2 P + 1 & \alpha \beta P + \lambda\\
		\beta P & \alpha \beta P + \lambda & \beta^2 P + 1 & \alpha \beta P + \lambda & \beta^2 P + 1 \\
		\alpha P & \alpha^2 P + 1 & \alpha \beta P + \lambda & \alpha^2 P\!+\! 1\! +\! N_1 & \alpha \beta P + \lambda\\
		\beta P & \alpha \beta P + \lambda & \beta^2 P + 1 & \alpha \beta P + \lambda & \beta^2 P \!+\! 1 \!+\! N_2 \\
	\end{bmatrix}
\end{equation}
Upon noticing
\begin{align}
	I(X;\hat Y_1, \hat Y_2) = \frac12 \log \bigg( &\frac{(\alpha^2 P + 1 + N_1)(\beta^2 P + 1 + N_2) }{(1+N_1)(1+N_2) - \lambda^2} \notag \\&- \frac{(\alpha \beta P + \lambda)^2}{(1+N_1)(1+N_2) - \lambda^2}  \bigg),
\end{align}
and
\begin{align}
	I(Y_1;\hat Y_1) &= \frac12 \log \left(\frac{\alpha^2 P + 1+N_1}{N_1} \right),\\
	I(Y_2;\hat Y_2) &= \frac12 \log \left( \frac{\beta^2 P + 1 + N_2}{N_2} \right),
\end{align}
and
\begin{align}
	I_0 &=   I(X;\hat Y_1,\hat Y_2) + I(\hat Y_1;\hat Y_2) \notag \\ 
	&= \frac12 \log \left(\frac{(\alpha^2 P + 1 + N_1)(\beta^2 P + 1 + N_2)}{(1+N_1)(1+N_2) - \lambda^2} \right),
\end{align}
we obtain the following bound by simplifying \eqref{eq:diamond_CF_1} and optimizing over $N_1$ and $N_2$:
\begin{align}
	\label{eq:nonperf_diamond_CF}
	&R_{\diamond,\textnormal{CF}}(R_1,R_2)  \notag \\
	&= \max_{N_1,N_2 > 0} \min \left\{\begin{matrix}
		\frac12\log \frac{(\alpha^2 P + 1 + N_1)(\beta^2 P + 1 + N_2) - (\alpha \beta P + \lambda)^2}{(1+N_1)(1+N_2)-\lambda^2}\\ R_1 + \frac{1}{2}\log \frac{N_1(\beta^2 P + 1 + N_2)}{(1+N_1)(1+N_2)-\lambda^2} \\
		R_2 + \frac{1}{2}\log \frac{N_2(\alpha^2 P + 1 + N_1)}{(1+N_1)(1+N_2)-\lambda^2}\\ 
		R_1 + R_2 + \frac{1}{2}\log \frac{N_1 N_2}{(1+N_1)(1+N_2)-\lambda^2}
	\end{matrix}\right\}.
\end{align}
We comment that an alternative CF lower bound for the Gaussian primitive diamond channel has been derived in \cite{katz2024gaussian}. Numerical simulation suggests that the two lower bounds coincide. 

For the upper bound, the general cut-set bound for the diamond channel in \eqref{eq:diamond_cutset_1} still holds. However, since the noises at the relays are not perfectly correlated, the input $X$ is not a deterministic function of the relay observations, therefore the mutual information term $I(X;Y_1, Y_2)$ is finite:
\begin{equation}
	I(X;Y_1,Y_2) = \frac{1}{2} \log \left(1 + \frac{(\alpha^2 + \beta^2 - 2\lambda \alpha \beta)P}{1-\lambda^2} \right)
\end{equation}
The cut-set upper bound for this scenario is therefore
\begin{equation} \label{eq:nonperf_diamond_cutset}
	\textnormal{Cutset}(R_1,R_2) \le \min\left\{\begin{matrix}
		R_1 + R_2 \\ R_1 + \psi(\beta^2 P)\\ R_2 + \psi(\alpha^2 P) \\
		\psi \left(\frac{(\alpha^2 + \beta^2 - 2\lambda \alpha \beta)P}{1-\lambda^2} \right)
	\end{matrix}\right\}.
\end{equation}

\balance

\bibliographystyle{IEEEtran}
\bibliography{IEEEabrv,references}

\end{document}